%% file: Biometricsfull2022-09-13.tex
\documentclass[12pt]{article}

\pdfoutput=1

\RequirePackage{geometry}
\usepackage{amsmath}
\usepackage{graphicx,psfrag,epsf}
\usepackage{natbib}
\usepackage{mathtools,amssymb}
\usepackage{amsthm}
\usepackage{soul} 
\usepackage{float} 
\usepackage{hyperref}

\input{GrandMacros}

\mathtoolsset{showonlyrefs=true} 

\newcommand{\blind}{0} 

\def\spacingset#1{\renewcommand{\baselinestretch}%
{#1}\small\normalsize} \spacingset{1}

\newtheorem{theorem}{Theorem}

\newtheorem{lemma}{Lemma}
\newtheorem{assumption}{Assumption}

\begin{document}

\def\mytitle{Penalized Estimation of Frailty-Based Illness-Death Models for Semi-Competing Risks}

\if0\blind
{
  \title{\bf \mytitle}
  \author{Harrison T. Reeder, Junwei Lu, and Sebastien J. Haneuse\hspace{.2cm}\\
    Department of Biostatistics, Harvard T.H. Chan School of Public Health}
  \maketitle
} \fi

\if1\blind
{
  \title{\bf \mytitle}
  \author{}
  \maketitle
} \fi

\begin{abstract}
\noindent Semi-competing risks refers to the time-to-event analysis setting where the occurrence of a non-terminal event is subject to whether a terminal event has occurred, but not vice versa. Semi-competing risks arise in a broad range of clinical contexts, including studies of preeclampsia, a condition that may arise during pregnancy and for which delivery is a terminal event. Models that acknowledge semi-competing risks enable investigation of relationships between covariates and the joint timing of the outcomes, but methods for model selection and prediction of semi-competing risks in high dimensions are lacking. Moreover, in such settings researchers commonly analyze only a single or composite outcome, losing valuable information and limiting clinical utility---in the obstetric setting, this means ignoring valuable insight into timing of delivery after preeclampsia has onset. To address this gap we propose a novel penalized estimation framework for frailty-based illness-death multi-state modeling of semi-competing risks. Our approach combines non-convex and structured fusion penalization, inducing global sparsity as well as parsimony across submodels. We perform estimation and model selection via a pathwise routine for non-convex optimization, and prove statistical error rate results in this setting. We present a simulation study investigating estimation error and model selection performance, and a comprehensive application of the method to joint risk modeling of preeclampsia and timing of delivery using pregnancy data from an electronic health record.
\end{abstract}

\noindent%
\vfill

{\tiny
\noindent This is the accepted version of the following article: \\

\noindent Reeder, H. T., Lu, J., \& Haneuse, S. (2023). Penalized estimation of frailty‐based illness–death models for semi‐competing risks. Biometrics, 79(3), 1657-1669.

\noindent which has been published in final form at \url{https://doi.org/10.1111/biom.13761}. This article may be used for noncommercial purposes in accordance with the \href{http://www.wileyauthors.com/self-archiving}{Wiley Self-Archiving Policy}.

}



\newpage
\spacingset{1.15} 

\section{Introduction}
\label{sec:intro}

Semi-competing risks refers to the time-to-event analysis setting where a non-terminal event of interest can occur before a terminal event of interest, but not vice versa \citep{fine2001semi}. Semi-competing risks are ubiquitous in health research, with example non-terminal events of interest for which death is a semi-competing terminal event including hospital readmission \citep{lee2015bayesian} and cancer progression \citep{jazic2016beyond}. An example where death is not the terminal outcome is preeclampsia (PE), a pregnancy-associated hypertensive condition that complicates between 2-8\% of all pregnancies and represents a leading cause of maternal and fetal/neonatal mortality and morbidity worldwide \citep{jeyabalan2013epidemiology}. PE can develop during the pregnancy starting at 20 weeks of gestation, but once an individual has given birth they can no longer develop PE, so PE onset and delivery form semi-competing risks. Clinically, the timing of these events is of vital importance. Once PE arises, maternal health risks increase as the pregnancy continues, while giving birth early to alleviate these risks may in turn pose risks to neonatal health and development. Therefore, we are motivated to develop risk models that identify covariates affecting risk and timing of PE while also characterizing the timing of delivery after PE has onset.

Semi-competing risks data represent a unique opportunity to learn about outcomes jointly, by (1) modeling the interplay between the events and baseline covariates, and (2) predicting the covariate-specific risk of experiencing combinations of the outcomes across time. Unfortunately, analysts commonly collapse this joint outcome, considering either the non-terminal or terminal event alone, or a composite endpoint \citep{jazic2016beyond}. While this enables the use of prediction methods for univariate binary or time-to-event outcomes, modeling risk for one outcome is both a lost opportunity and a severe misalignment with how health-related decisions are actually made; as the PE setting illustrates, clinical care is informed by the joint timing of PE onset and subsequent delivery, not just risk of PE.

Instead, frailty-based illness-death multi-state models \citep{xu2010statistical, lee2015bayesian} characterize the dependency of semi-competing risks and covariates, while also enabling absolute joint risk prediction across time \citep{putter2007tutorial}. These methods comprise three cause-specific hazard submodels for: (i) the non-terminal event; (ii) the terminal event without the non-terminal event; and, (iii) the terminal event after the non-terminal event. Different covariates can affect each hazard differently, and the interplay of these submodels determines the overall covariate-outcome relationship. A person-specific random frailty shared across the submodels captures residual dependence between the two events.

Motivated by application to PE, we consider the task of developing joint risk models for semi-competing risks, specifically in high-dimensional settings such as electronic health records-based studies. Two questions framing model development emerge: (1) which covariates should be included in each submodel, and (2) can information about covariate effects be shared across submodels? To our knowledge only two published papers consider variable selection for these (and related) models, each with important limitations. \citet{sennhenn2016structured} propose $\ell_1$-penalized estimation for general multistate models, with parameter-wise penalties inducing sparsity in each submodel and a fused penalty coercing effects for a given covariate to be the same across submodels. This framework, however, does not permit a shared frailty in the model specification, focuses solely on $\ell_1$-penalization, and uses a Newton-type algorithm that does not scale to high dimensions. Instead, \citet{chapple2017bayesian} propose a Bayesian spike-and-slab variable selection approach for frailty illness-death models. This framework, however, does not consider linking coefficients across submodels, and is computationally intensive even in low dimensions.

In this paper we propose a novel high-dimensional estimation framework for penalized parametric frailty-based illness-death models. A critical challenge in this setting, however, is that the likelihood-based loss function is non-convex. This renders the development of theoretical results and efficient computational tools particularly difficult. Moreover, to our knowledge no prior literature has examined theoretical properties of penalized frailty-based illness-death models.
In relevant work, \citet{loh2015regularized} prove error bounds for non-convex loss functions with non-convex penalties, but the conditions underlying their result do not directly apply to this setting. 
Taking into account these various issues, the contributions of this paper are threefold. First, we propose a framework for selecting sparse covariate sets for each submodel via individual non-convex penalties, while inducing parsimony via a fused penalty on effects shared across submodels. Second, we develop a proximal gradient optimization algorithm with a pathwise routine for tuning the model over a grid of regularization parameters. Finally, we prove a high-dimensional statistical error rate for the penalized frailty-based illness-death model estimator. We present a simulation study investigating estimation and model selection properties, and develop a joint risk model for PE and delivery using real pregnancy outcome data from electronic health records.

\section{Penalized Illness-Death Model Framework} \label{sec:methods}
\subsection{Illness-Death Model Specification} \label{subsec:model}
Let $T_1$ and $T_2$ denote the times to the non-terminal and terminal events, respectively. As outlined in \citet{xu2010statistical}, the illness-death model characterizes the joint distribution of $\bT=(T_1,T_2)$ by three hazard functions: a cause-specific hazard for the non-terminal event; a cause-specific hazard for the terminal event in the absence of the non-terminal event; and, a hazard for the terminal event conditional on $T_1=t_{1}$. These three hazards can be structured as a function of covariates, denoted $\bX$, and an individual-specific random frailty, denoted $\gamma$, to add flexibility in the dependence structure between $T_1$ and $T_2$, as follows:
\begin{align}
\label{eq:h1formal}
	 h^c_1(t_1\mid \bX_1,\gamma)\ & =\ \lim_{\Delta \downarrow 0} \Delta^{-1} \Pr(T_1\in [t_1,t_1+\Delta)\mid T_1 \geq t_1,T_2 \geq t_1,\bX_1,\gamma),\quad t_1 > 0, \\
\label{eq:h2formal}
	 h^c_2(t_2\mid \bX_2,\gamma)\ & =\ \lim_{\Delta \downarrow 0} \Delta^{-1} \Pr(T_2\in [t_2,t_2+\Delta)\mid T_1 \geq t_2,T_2 \geq t_2,\bX_2,\gamma),\quad t_2 > 0, \\
\label{eq:h3formal}
	 h^c_3(t_2\mid t_1,\bX_3,\gamma)\ & =\ \lim_{\Delta \downarrow 0} \Delta^{-1} \Pr(T_2\in [t_2,t_2+\Delta)\mid T_1 =t_1,T_2 \geq t_2,\bX_3,\gamma),\quad t_2 > t_1 > 0,
\end{align}
where $\bX_1$,  $\bX_2$, and $\bX_3$ are each subsets of $\bX$. Practically, in order to make progress, one must specify some form of structure regarding the dependence of each hazard on $\bX$ and $\gamma$. In this paper we focus on the class of multiplicative-hazard regression models, of the form:
\begin{align}
	\label{eq:h1model}
	 h^c_1(t_1\mid \bX_1,\gamma)\ & =\ \gamma h_{1}(t_1\mid \bX_1) = \gamma h_{01}(t_1)\exp(\bX_1\trans\bbeta_1) , \quad t_1>0, \\
	 \label{eq:h2model}
	 h^c_2(t_2\mid \bX_2,\gamma)\ & =\ \gamma h_{2}(t_2\mid \bX_2) = \gamma  h_{02}(t_2)\exp(\bX_2\trans\bbeta_2) , \quad t_2>0, \\ 
	\label{eq:h3model}
	 h^c_{3}(t_2\mid t_1,\bX_3,\gamma)\ & =\ \gamma h_{3}(t_2\mid t_1,\bX_3) = \gamma  h_{03}(t_2\mid t_1)\exp(\bX_3\trans\bbeta_3), \quad t_2>t_1>0,
\end{align}
where $h_{0g}$ is a transition-specific baseline hazard function, $g=1,2,3$, and $\bbeta_g \in \bbR^{d_g}$ is a $d_g$-vector of transition-specific log-hazard ratio regression coefficients. For ease of notation of the hazard functions, we subsequently suppress conditionality on $\gamma$ and $\bX_g$.

Within this class of models, analysts must make several choices about the structure of the specific model to be adopted.  First, it must be decided how exactly the non-terminal event time $T_1$ affects $h_{03}(t_2\mid t_1)$ in submodel \eqref{eq:h3model}. The so-called Markov structure sets $h_{03}(t_2\mid t_1) = h_{03}(t_2)$, meaning the baseline hazard is independent of $t_1$. Alternatively, the semi-Markov structure sets $h_{03}(t_2\mid t_1) = h_{03}(t_2-t_1)$, so the time scale for $h_3$ becomes the time from the non-terminal event to the terminal event (sometimes called the sojourn time) \citep{putter2007tutorial,xu2010statistical}. Semi-Markov specification also allows functions of $t_1$ as covariates in the $h_3$ submodel, to further capture dependence between $T_1$ and $T_2$. While the choice of structure changes the interpretation of $h_3$ and $\bbeta_3$, the proposed penalization framework and theory allow either specification. Rather than focus on one or the other, we use semi-Markov specification when writing equations to simplify notation, as well as in the simulation study. However, in our application to pregnancy data we use Markov models to facilitate interpretation of the resulting estimates on the scale of gestational age.

A second important choice concerns the form of the three baseline hazard functions. Given the overarching goals of this paper, we focus on parametric specifications with a fixed-dimensional $k_g$-vector of unknown parameters, $\bphi_g = (\phi_{g1},\dots,\phi_{gk_g})\trans$, for the $g$th baseline hazard. For example, one could adopt a form arising from some specific distribution such as the Weibull distribution: $h_{0g}(t) = \exp(\phi_{g1} + \phi_{g2})\cdot t^{\exp(\phi_{g1}) - 1}$. More flexible options include the piecewise constant baseline hazard defined as $h_{0g}(t) = \sum_{j=1}^{k_g} \exp(\phi_{gj}) \mathbb{I}(t^{(j)} \leq{} t < t^{(j+1)})$, with a user-defined set of breakpoints $0 = t^{(1)} <\dots < t^{(k_g)} < t^{(k_g+1)}=\infty$. Web Appendix~\ref{sec:splines} describes other possible flexible spline-based baseline hazard specifications.
  
Finally, one must choose a distribution for $\gamma$, the individual-specific frailties. These terms serve to capture additional within-subject correlation between $T_1$ and $T_2$ beyond covariate effects, and increase flexibility beyond the assumed baseline hazard specification and Markov or semi-Markov model structure \citep{xu2010statistical}. In this, the frailties play a role that is analogous to that of random effects in generalized linear mixed models. As discussed in Web Appendix~\ref{subsec:predfrail}, inclusion of frailties also helps to characterize variability in individualized risk predictions.
While in principle one could adopt any distribution for $\gamma$, we focus on the common choice of $\gamma \sim \text{Gamma}(e^{-\sigma},e^{-\sigma})$. This distribution has mean 1 and variance $e^{\sigma}$, and uniquely yields a closed form marginal likelihood, as shown in \eqref{eq:likelihoodi:frailty}. This log-variance parameter $\sigma$ can be interpreted as characterizing residual variability of the outcomes beyond the specified baseline hazard and transition model structure.

\subsection{The Observed Data Likelihood}
\label{subsec:joint}

Now, let $C$ denote the right-censoring time. The observable outcome data for the $i$th subject is then $\mathcal{D}_i = \{Y_{1i}, \Delta_{1i}, Y_{2i}, \Delta_{2i}, \bX_{i}\}$ where $Y_{1i}=\min(C_i,T_{1i},T_{2i})$, $Y_{2i}=\min(C_i,T_{2i})$, $\Delta_{1i}=\bbI(Y_{1i}=T_{1i})$, and $\Delta_{2i}=\bbI(Y_{2i}=T_{2i})$. We denote the corresponding observed outcome values as $y_{1i}$, $y_{2i}$, $\delta_{1i}$, and $\delta_{2i}$. Given specification of models \eqref{eq:h1model}, \eqref{eq:h2model} and \eqref{eq:h3model}, let $\bphi = (\bphi_{1}\trans,\bphi_{2}\trans,\bphi_{3}\trans)\trans$ denote the $k \times 1$ vector of baseline hazard components, with $k = k_1+k_2+k_3$, and $\bbeta = (\bbeta_{1}\trans,\bbeta_{2}\trans,\bbeta_{3}\trans)\trans$ the $d\times 1$ vector of log-hazard ratios, with $d = d_1+d_2+d_3$. Finally, let $\bpsi = (\bbeta\trans, \bphi\trans,  \sigma)\trans$ denote the full set of $d + k + 1$ unknown parameters.

To develop the observed data likelihood, we assume independencies between: frailty and covariates, $\gamma \indep \bX$; frailty and censoring time given covariates, $\gamma \indep C \mid  \bX$; and, censoring time and event times, given covariates and frailty, $C\indep \bT \mid  (\gamma,\bX)$. Illustrating under semi-Markov specification, the $i$th likelihood contribution given $\gamma_i$ is $\mathcal{L}(\bbeta,\bphi\mid \gamma_i,\mathcal{D}_i) = h^c_{1}(y_{1i})^{\delta_{1i}}h^c_{2}(y_{1i})^{(1-\delta_{1i})\delta_{2i}}h^c_{3}(y_{2i}-y_{1i})^{\delta_{1i}\delta_{2i}} \exp\{ -H^c_1(y_{1i}) - H^c_2(y_{1i}) - \delta_{1i}H^c_3(y_{2i}-y_{1i}) \}$,
where $H^{c}_g(t) = \int_{0}^t h^{c}_g(s)ds$. Finally, integrating out the gamma-distributed frailty, the $i$th marginal likelihood contribution takes the closed form
\begin{equation}\label{eq:likelihoodi:frailty}
\begin{aligned} 
\mathcal{L}(\bpsi\mid \mathcal{D}_i)\ =\ & h_{1}(y_{1i})^{\delta_{1i}}h_{2}(y_{1i})^{(1-\delta_{1i})\delta_{2i}}h_{3}(y_{2i}-y_{1i})^{\delta_{1i}\delta_{2i}} (1+e^{\sigma})^{\delta_{1i}\delta_{2i}}
\\ & \times (1+e^{\sigma}\{ H_1(y_{1i}) + H_2(y_{1i}) + H_3(y_{2i}-y_{1i}) \})^{-\exp(-\sigma)-\delta_{1i}-\delta_{2i}}.
\end{aligned}
\end{equation} 

\subsection{Penalization for Sparsity and Model Parsimony}
\label{subsec:pen}

Given an i.i.d sample of size $n$, let $\ell(\bpsi) = -n^{-1}\sum_{i=1}^n \log \mathcal{L}(\bpsi\mid \mathcal{D}_i)$ denote the negative log-likelihood. Penalized likelihood estimation follows via the introduction of a penalty function $P_{\blambda}(\bpsi)$, yielding a new objective function of the form $Q_{\blambda}(\bpsi) = \ell(\bpsi) + P_{\blambda}(\bpsi).$

Letting $\bX=\bX_1=\bX_2=\bX_3$, the illness-death model's hazards allow each $X_j$ to potentially have three different coefficients: a cause-specific log-hazard ratio for the non-terminal event ($\beta_{1j}$), a cause-specific log-hazard ratio for the terminal event ($\beta_{2j}$), and a log-hazard ratio for the terminal event given the non-terminal event has occurred ($\beta_{3j}$). We propose a structured $P_{\blambda}(\bpsi)$ simultaneously targeting two properties: (i) sparsity, by identifying important non-zero covariate effects, and (ii) parsimony, by identifying relationships between the effects of each covariate across the three submodels. We propose the general form
\begin{equation}\label{eq:genpen}
	P_{\blambda}(\bpsi)\ =\ \sum_{g=1}^3 \sum_{j=1}^{d_{g}}p_{\lambda_1}(\lvert\beta_{gj}\rvert)  + \sum_{g\neq g'}\sum_{j=1}^{d_g} p_{\lambda_2}\left(\lvert\beta_{gj}-\beta_{g'j}\rvert\right).
\end{equation}
The first component, regulated by $\lambda_1$, induces sparsity by setting unimportant covariate effects to zero. The second component, regulated by $\lambda_2$, induces parsimony by regularizing the cofficients of each covariate $X_j$ towards being similar or shared across submodels. 

For each component, one could, in principle, consider any of a wide array of well-known penalties, such as the Lasso $\ell_1$ penalty \citep{tibshirani1996regression} or non-convex penalties like smoothly-clipped absolute deviation (SCAD) \citep{fan2001variable}:
\begin{align} \label{eq:scad}
	p_{\lambda}(\lvert\beta\rvert) = \begin{cases} \lambda\lvert\beta\rvert, & \lvert\beta\rvert\leq{} \lambda, \\ -\left\{ \beta^2 - 2 \xi \lambda\lvert\beta\rvert + \lambda^2 \right\}/\{2(\xi -1)\}, & \lambda < \lvert\beta\rvert \leq{} \xi \lambda, \\ (\xi +1)\lambda^2/2, & \lvert\beta\rvert > \xi \lambda, \end{cases}
\end{align}
with $\xi >2$ controlling the level of non-convexity. This penalty behaves like the Lasso near zero, but flattens out for larger values, reducing bias on truly non-zero estimates. 

The motivation for this form of penalty is both clinical and statistical, and we emphasize that depending on the application, the fusion penalties between $\bbeta_1$, $\bbeta_2$ and/or $\bbeta_3$ can be included or omitted from $P_{\blambda}(\bpsi)$. Clinically, fusion penalties are valuable when there is subject matter knowledge indicating that covariates likely have similar effects in two or more submodels. For example, $\bbeta_2$ and $\bbeta_3$ both represent log-hazard ratios for the terminal event, with $\bbeta_2$ representing cause-specific effects in the absence of the non-terminal event, and $\bbeta_3$ representing effects conditional on the occurrence of the non-terminal event. Therefore, fusing $\bbeta_2$ and $\bbeta_3$ imposes a structure where covariate effects on the terminal event are similar whether or not the non-terminal event has occurred. Relatedly, $\bbeta_1$ and $\bbeta_2$ both represent cause-specific log-hazard ratios, for the non-terminal and terminal event respectively. Therefore, in settings where both non-terminal and terminal events represent negative health outcomes, like cancer progression and death, fusing these components induces each covariate to have similar or shared cause-specific hazard ratio estimates for the two events. In any case, well-chosen structured fusion penalties can be used to encode clinically meaningful subject matter knowledge into the estimation framework. 

Statistically, fusion penalties may also be valuable when there is relatively little information on one of these submodels, and the goal is to impose structure and stabilize estimation. For example, in settings where the non-terminal event is rare relative to the terminal event, there will be more information available for estimating $\bbeta_2$ relative to $\bbeta_1$ or $\bbeta_3$. Therefore, adding a fusion penalty between $\bbeta_2$ and either $\bbeta_1$ and/or $\bbeta_3$ regularizes the more variable estimates of $\bbeta_1$ and/or $\bbeta_3$ towards the more precise estimates of $\bbeta_2$, effectively borrowing information across submodels. As with all regularized estimation, this directly reflects a bias-variance trade off: imposing structure on covariate effects across hazards to reduce variance, or leaving effects unstructured across hazards to reduce bias.


\section{Optimization}

Practically, minimizing the objective function $Q_{\blambda}(\bpsi)$ with respect to $\bpsi$ poses several interconnected challenges: the loss function $\ell$ is non-convex due to the marginalized random frailty; the penalty functions $p_{\lambda}$ may also be non-convex; and, the fusion penalty component does not admit standard algorithms for general fused Lasso tailored to linear regression \citep{tibshirani2011solution}. Finally,  the combination of penalties requires tuning a two-dimensional regularization parameter $\blambda$. In this section, we propose a comprehensive optimization routine to simultaneously and efficiently handle these challenges.

\subsection{Proximal Gradient Descent with a Smoothed Fusion Penalty}

Proximal gradient descent iteratively minimizes objective functions like $Q_{\blambda}$ defined as the sum of a differentiable loss function and a non-differentiable penalty. When $P_{\blambda}(\bpsi)$ is the standard Lasso $\ell_1$-penalty $\lambda\|\bbeta\|_{1}$, the algorithm reduces to standard gradient descent with an added soft-thresholding operation. To leverage this property, we combine two techniques to recast $Q_{\blambda}$ from a loss with a complex penalty into a loss with a simple Lasso penalty.

First we decompose each $p_\lambda$ in \eqref{eq:genpen} into the sum of a smooth concave term and a simple $\ell_1$ penalty term, of the form $p_\lambda(\lvert x\rvert)  = \ptilde_\lambda(\lvert\beta\rvert) + \lambda\lvert\beta\rvert$,
where $\ptilde_{\lambda}$ is a Lipschitz-smooth concave function \citep{zhao2018pathwise}. The goal is to treat the smooth component as part of the likelihood, leaving only a simpler $\ell_1$ penalty. This decomposition can be done to both the parameterwise penalties $p_{\lambda_1}$ and the fusion penalties $p_{\lambda_2}$, rewriting \eqref{eq:genpen} as
\begin{equation}\label{eq:genpendecomp}
	P_{\blambda}(\bpsi) = \sum_{g=1}^3 \sum_{j=1}^{d_{g}}\ptilde_{\lambda_1}(\lvert\beta_{gj}\rvert) + \sum_{g\neq g'}\sum_{j=1}^{d_g} \ptilde_{\lambda_2}\left(\lvert\beta_{gj}-\beta_{g'j}\rvert\right) + \lambda_{1}\|\bbeta\|_{1} + \Omega_{\lambda_2}(\bpsi),
\end{equation}
where $\Omega_{\lambda_2}(\bpsi) = \sum_{g\neq g'}\sum_{j=1}^{d_g} \lambda_2 \lvert\beta_{gj}-\beta_{g'j}\rvert$ denotes the fusion $\ell_1$ penalty.

However, this fusion penalty $\Omega_{\lambda_2}(\bpsi)$ still complicates optimization, so next we use Nesterov smoothing to substitute it with a smoothed, differentiable surrogate   \citep{chen2012smoothing}. Defining $\bD_{\lambda_2}$ as a contrast matrix such that $\Omega_{\lambda_2}(\bpsi) = \|\bD_{\lambda_2}\bpsi\|_1$, the surrogate is
\begin{equation}\label{eq:pensmooth}
\Omegatilde_{\lambda_2,\mu}(\bpsi) = \max_{\|\bz\|_{\infty}\leq{} 1}\left( \bz\trans \bD_{\lambda_2}\bpsi - \mu\|\bz\|^2_2/2 \right) =  (\bz^*)\trans \bD_{\lambda_2}\bpsi - \mu\|\bz^*\|^2_2/2,
\end{equation}
where $\bz^* = \bS(\bD_{\lambda_2}\bpsi/\mu)$ and $\bS(\bx)$ is the vector-valued projection operation onto the unit box, defined at the $j$th element by $[\bS(\bx)]_j = \sign(x_j)\max(1, \lvert x_j\rvert)$ and $\mu>0$ is a user-chosen smoothness parameter. Smaller $\mu$ yields a tighter approximation, with the gap between penalty and surrogate bounded by $\Omega_{\lambda_2}(\bpsi) - \mu J / 2 \leq{} \Omegatilde_{\lambda_2,\mu}(\bpsi) \leq{} \Omega_{\lambda_2}(\bpsi)$, where $J$ is the number of pairwise fusion terms. Web Appendix~\ref{sec:algodetails} details tuning methods for $\mu$.

Together, \eqref{eq:genpendecomp} and \eqref{eq:pensmooth} recast the objective function $Q_{\blambda}(\bpsi)$ as an $\ell_1$-penalized objective:
\begin{equation}\label{eq:optimpensep}
Q_{\blambda,\mu}(\bpsi) = \widetilde{\ell}_{\blambda,\mu}(\bpsi) + \lambda_1\|\bbeta\|_1,
\end{equation}
where $\widetilde{\ell}_{\blambda,\mu}(\bpsi) = \ell(\bpsi) + \sum_{g=1}^3 \sum_{j=1}^{d_g}\ptilde_{\lambda_1}(\lvert\beta_{gj}\rvert) +  \sum_{g\neq g'}\sum_{j=1}^{d_g} \ptilde_{\lambda_2}(\lvert\beta_{gj}-\beta_{g'j}\rvert) +  \Omegatilde_{\lambda_2,\mu}(\bpsi).$ Towards optimizing \eqref{eq:optimpensep}, define the vector-valued soft thresholding operation $\bS_{\lambda}(\bx)$ at the $j$th element by $[\bS_{\lambda}(\bx)]_j = \sign(x_j)\max(0, \lvert x_j\rvert - \lambda)$. Then the $m$th step of the iterative proximal gradient algorithm is $\bpsi^{(m)} \leftarrow \bS_{\lambda_1}\left\{ \bpsi^{(m-1)} - r^{(m)} \cdot \nabla \widetilde{\ell}_{\blambda,\mu}(\bpsi^{(m-1)}) \right\}$, where $r^{(m)}$ is an adaptive step size determined by backtracking line search \citep[see, e.g., `Algorithm 3' of][]{wang2014optimal}. Iterations continue until change in objective function $\lvert Q_{\blambda,\mu}(\bpsi^{(m)}) - Q_{\blambda,\mu}(\bpsi^{(m-1)})\rvert$ falls below a given threshold (e.g., $10^{-6}$ in this paper).

\subsection{Tuning Regularization Parameters via Pathwise Grid Search}

For non-convex penalized problems with a single regularization parameter $\lambda$, recent path-following routines apply proximal gradient descent or coordinate descent over a decreasing sequence of regularization parameters \citep{wang2014optimal,zhao2018pathwise}. At each new $\lambda$, these routines initialize at the solution of the prior $\lambda$. The result is a sequence of estimates across a range of penalization levels, also called a regularization path. Under certain conditions these `approximate path-following' approaches yield high-quality local solutions with attractive theoretical properties, even when the loss and/or regularizer are non-convex. Heuristically, many non-convex objective functions are locally convex in the neighborhood of well-behaved optima, and so incrementally optimizing over a sequence of small changes to $\lambda$ ensures that each local solution remains in the convex neighborhood of the solution under the previous $\lambda$.

Therefore, we develop a pathwise approach to the penalized illness-death model \eqref{eq:optimpensep} with a search routine over a two-dimensional grid of the sparsity parameter $\lambda_1$ and fusion parameter $\lambda_2$ (Figure~\ref{fig:gridsearchplot}). This routine consists of an outer loop decrementing $\lambda_1$ as in standard pathwise algorithms, and an inner loop comprising a branching pathwise search over increasing $\lambda_2$ values. The resulting grid search of the model space slowly grows the number of non-zero coefficient estimates as $\lambda_1$ decreases in the outer loop, and then explores how the resulting non-zero coefficients fuse as $\lambda_2$ increases in the inner loop. Assuming sparsity of the regression coefficients, a straightforward choice to begin the pathwise regularization grid search routine is to set $\bbeta=\mathbf{0}$, and set the remaining parameters to the unadjusted MLE estimates fit without covariates. Optimization at each point initializes from the solution at the prior relevant step. 

\begin{figure}
\centering
	\includegraphics[width=\textwidth]{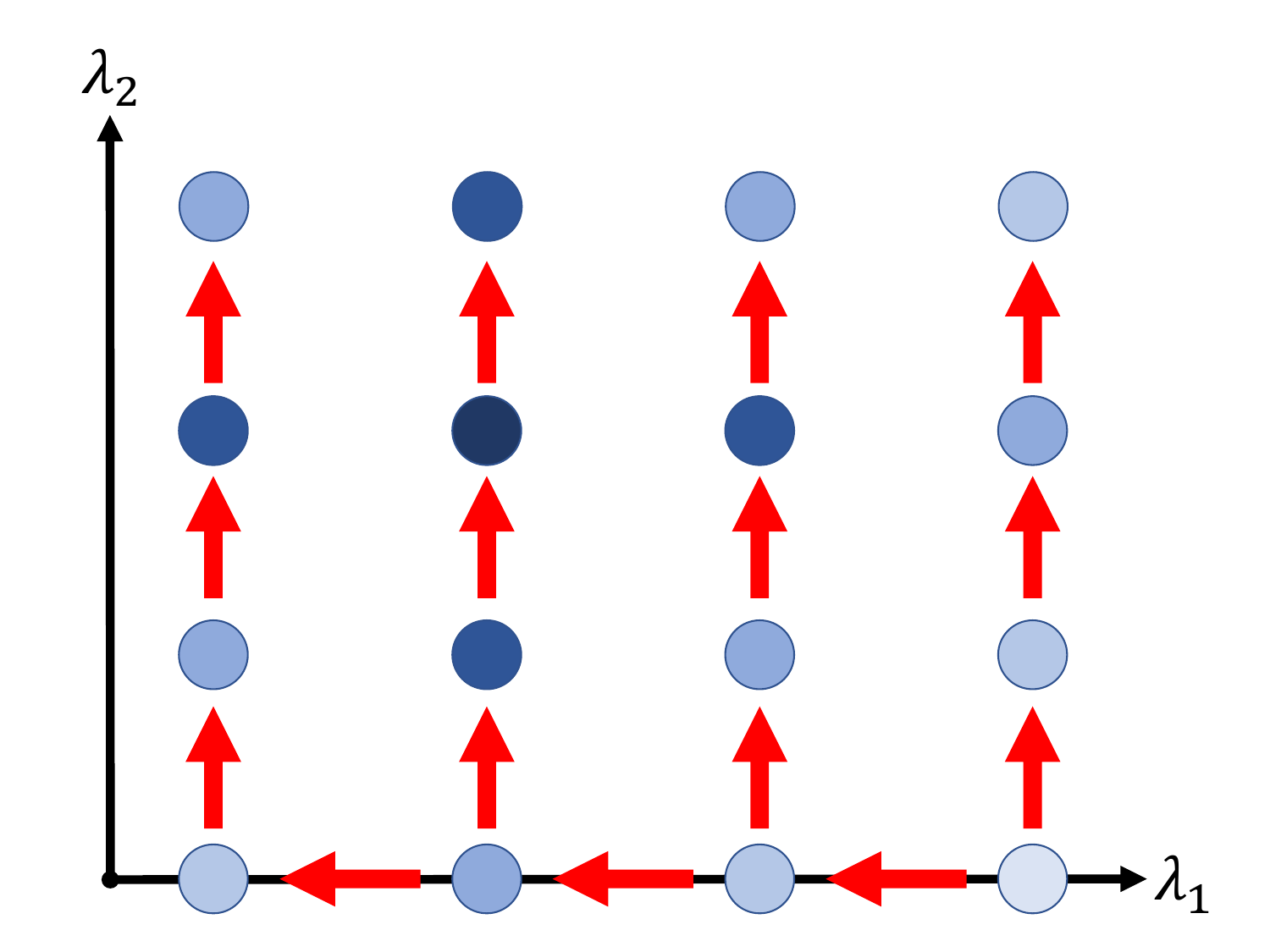}
\caption{Schematic depicting path-following grid search routine over $(\lambda_1,\lambda_2)$. Each dot represents a $(\lambda_1,\lambda_2)$ pair for which the penalized estimator is fitted, with darker shading corresponding to better model fit metric (e.g., AIC or BIC). The arrows illustrate the path of the search routine, with optimization at each grid point starting at the solution of the previous point. This figure appears in color in the electronic version of this article, and any mention of color refers to that version.\label{fig:gridsearchplot}}
\end{figure}

The grid should be as fine as computational costs allow; \citet{wang2014optimal} recommend that successive values of $\lambda_1$ differ by no more than a factor of 0.9, and for $\lambda_2$ we chose four grid points in simulations and seven for the data application. Final choice of $(\lambda_1,\lambda_2)$ follows by minimizing a performance metric computed at each grid point, depicted by shading at each point in Figure~\ref{fig:gridsearchplot}. Metrics such as Bayesian Information Criterion (BIC) or Akaike Information Criterion (AIC) may be computed using model degrees of freedom estimated by the number of unique covariate estimates \citep{sennhenn2016structured}.


\section{Theoretical Results}\label{sec:theory}

In this section, we derive the statistical error rate for estimation of the true parameter vector $\bpsi^*$ in a gamma frailty illness-death model with non-convex penalty, encompassing high-dimensional settings where $d > n$ with sparsity level denoted $\|\bpsi^*\|_{0} = s$. This work builds on the framework of \citet{loh2015regularized}, extended to the additional complexities of parametric gamma-frailty illness-death models. We develop a set of sufficient conditions for this setting under which we prove the statistical rate, and verify that such conditions hold with high probability under several common model specifications.

This statistical investigation focuses on the estimator 
\begin{equation}\label{eq:optimtheory}
\bpsihat\ =\ \argmin_{\|\bpsi\|_1 \leq{} R_{1}} Q_{\lambda}(\bpsi) = \ell(\bpsi) + \sum_{j=1}^{d_{1}} p_{\lambda} \left(\lvert\beta_{1j}\rvert\right) + \sum_{g=2}^3\sum_{j=1}^{d_{g}} p_{\lambda} \left(\lvert\beta_{gj} - \beta_{1j}\rvert\right).
\end{equation} 
We note that while the general framework \eqref{eq:genpen} allows a penalty on every element and pairwise difference of the regression parameter vectors $\bbeta_1$, $\bbeta_2$, and $\bbeta_3$, the estimator based on \eqref{eq:optimtheory} specifically penalizes $\bbeta_1$ and its pairwise differences with $\bbeta_2$ and $\bbeta_3$. This facilitates the theoretical analysis while retaining the property that the elementwise differences between $\bbeta_g$'s are sparse. Moreover, to accommodate the role of non-convexity the constraint $\|\bpsi\|_{1} \leq{} R_1$ is imposed on the parameter space over which solutions are sought~\citep{loh2015regularized}. See Assumption~\ref{ass:vec} below and its remark for detailed discussion. 
 
We start by listing assumptions used to derive the statistical rate of the estimator of $\bpsi^*$ based on \eqref{eq:optimtheory}. To unify outcome notation across all three submodels $g=1,2,3$, let
\begin{align}\label{eq:ytilde}
\Ytilde_{gi} = & \begin{cases} Y_{1i}, & g=1,2, \\ Y_{2i}-Y_{1i}, & g=3, \end{cases} &\text{and}& & \Deltatilde_{gi} = & \begin{cases} \Delta_{1i}, & g=1, \\ (1-\Delta_{1i})\Delta_{2i}, & g=2, \\ \Delta_{1i}\Delta_{2i}, & g=3,\end{cases}
\end{align}
where $Y_{gi}$ and $\Delta_{gi}$ are defined as in Section \ref{subsec:joint}.

\begin{assumption}[Bounded Data]\label{ass:bound}
There exists some administrative maximum time $\tau_{Y}$ such that $0<Y_{1} \leq{} Y_{2} \leq{} \tau_{Y}<\infty$. Additionally, there exists some positive covariate bound $\tau_{X}$ such that $\tau_{X} \geq \|\bX\|_{\infty}$, where $\|\bX\|_{\infty} = \max_{j=1,\dots,d}\lvert X_{j}\rvert$.
\end{assumption}
Assumption~\ref{ass:bound} ensures boundedness of the observed data. This assumption will invariably be satisfied in real world data applications, especially in time-to-event studies where person-time is censored and there are practical limits on covariate values.

\begin{assumption}[Bounded True Parameter]\label{ass:vec}
There exists a $R_2 > 0$ such that $\|\bpsi^*\|_{1} \leq{} R_{2}$.
\end{assumption}
Assumption~\ref{ass:vec} characterizes the overall length of the true parameter vector in terms of $\ell_1$-norm. Combined with the side constraint introduced in \eqref{eq:optimtheory} and setting $R = R_1+ R_2$, this ensures by the triangle inequality that there is an overall bound $\|\bpsi - \bpsi^*\|_1 \leq R$ for each iterate and all stationary points of the optimization routine. 

\begin{assumption}[Bounded Minimum Population Hessian Eigenvalue]\label{ass:hess}
There exists a $\rho>0$ such that $\min_{\bpsi: \|\bpsi-\bpsi^*\|_2 \leq{} R}\lambda_{\min}\{\Sigma(\bpsi)\} \ge \rho$, where $\lambda_{\min}\{\Sigma(\bpsi)\}$ is the minimum eigenvalue of $\Sigma(\bpsi) = \bbE\{\nabla^2 \ell(\bpsi)\}$, the population Hessian matrix at $\bpsi$.
\end{assumption}
Assumption~\ref{ass:hess} characterizes the positive-definiteness of the expected Hessian matrix of the loss function as a function of $\bpsi$, and guarantees curvature of the population loss function in a neighborhood around the truth.

\begin{assumption}[Baseline Hazard Function Sufficient Conditions]\label{ass:cond}
For $g,r=1,2,3$, $j=1,\dots,k_{g}$, and $l=1,\dots,k_{r}$, and for all $\left\{\bpsi:\|\bpsi - \bpsi^*\|_2 \leq{} R\right\}$,
\begin{enumerate}
	\item[(a)] $H_{0g}(t)$, $\partial H_{0g}(t) / \partial\phi_{gj}$, and $\partial^2 H_{0g}(t) / (\partial\phi_{gj}\partial\phi_{rl})$ are bounded functions on $0 \leq{} t \leq{} \tau$ for any $\tau > 0$. \label{ass:cond1}
	\item[(b)] $\Var\left\{\Deltatilde_{gi}\partial \log h_{0g}(\Ytilde_{gi}) / \partial\phi_{gj}\right\}$ is finite. \label{ass:cond2}
	\item[(c)] Each log-hazard second derivative factorizes into the form $\partial^2 \log h_{0g}(t) / (\partial\phi_{gj}\partial\phi_{rl}) = w^{gr}_{jl}(\bpsi)z^{gr}_{jl}(t)$, where $w^{gr}_{jl}(\bpsi)$ is only a function of $\bpsi$ and $z^{gr}_{jl}(t)$ is only a function of $t$. In addition, every $\Var\left\{\Deltatilde_{gi}z^{gr}_{jl}(\Ytilde_{gi})\right\}$ is finite. \label{ass:cond3}
\end{enumerate}
\end{assumption}
Assumption~\ref{ass:cond} outlines conditions regarding the baseline hazard functions in the illness-death model specification. Collectively, these conditions are imposed to control the maximum deviations of the gradient $\|\nabla \ell(\bpsi^*)\|_{\infty}$, and Hessian $\|\nabla^2 \ell(\bpsi) - \Sigma(\bpsi)\|_{\max}$ for all $\bpsi$ over the $\ell_2$-ball $\|\bpsi - \bpsi^*\|_2 \leq{} R$, where $\|\cdot\|_{\max}$ is the matrix elementwise absolute maximum. Specifically, the gradient and Hessian of the empirical loss function $\ell$ may be unbounded, which complicates our analysis; under a Weibull specification, for example, several elements of the gradient $\nabla \ell(\bpsi)$ involve the term $\log \ytilde_{gi}$, which diverges approaching 0. As such, Assumptions~\ref{ass:cond}b and \ref{ass:cond}c are used to control the unbounded quantities, while Assumption~\ref{ass:cond}a bounds remaining terms. Note, these conditions are satisfied by commonly-used baseline hazard choices; Web Appendix~\ref{sec:lemmas} contains proofs for piecewise constant and Weibull specifications. Lastly, while the conditions in Assumption~\ref{ass:cond} are presented under a semi-Markov model, analogous conditions can be expressed for a Markov model.

We now present the main theorem on the statistical rate of the estimator in \eqref{eq:optimtheory}. We take $p_{\lambda}$ to be the SCAD penalty defined in \eqref{eq:scad} to streamline the statement in terms of SCAD's non-convexity parameter $\xi$, though the result holds for other penalty functions including the Lasso and minimax concave penalty (MCP) \citep{zhang2010nearly}, as described in Web Appendix~\ref{sec:proof}.
\begin{theorem}\label{th:main}
Under Assumptions~\ref{ass:bound}, \ref{ass:vec}, and \ref{ass:hess} and sparsity level $s = \|\bpsi^*\|_{0}$, consider a gamma frailty illness-death model satisfying Assumption~\ref{ass:cond} with SCAD penalization as in \eqref{eq:optimtheory}. Suppose the SCAD non-convexity parameter $\xi$ satisfies $3/\{4(\xi -1)\} < \rho$, where $\rho$ is the population Hessian eigenvalue bound defined in Assumption \ref{ass:hess}. Then choosing $\lambda = c\sqrt{\log(dn)/n}$ for sample size $n$, parameter dimensionality $d$, and sufficiently large constant $c$, any stationary point $\bpsihat$ of \eqref{eq:optimtheory} will have a statistical rate that varies with $s$, $n$, and $d$ as
\begin{equation}
\|\bpsihat-\bpsi^*\|_{2} = O_P\left(\sqrt{s\log(dn) / n}\right).
\end{equation}
\end{theorem}
The proof of this theorem and detailed discussion are left to Web Appendix~\ref{sec:proof}. In particular, due to the complexities outlined in the discussion of Assumption~\ref{ass:cond}, the proof relies on a weaker version of the so-called Restricted Strong Convexity condition than that of \citet{loh2015regularized}.
Lastly, we note that by this result, consistency of the estimator $\bpsihat$ follows in the high-dimensional regime under scaling condition $s\log(dn)/n \rightarrow 0$.


\section{Simulation Studies} \label{sec:simulation}

In this section, we present a series of simulation studies to investigate the performance of the proposed methods in terms of estimation error and selection of the covariate effects $\bbeta$, comparing various penalty specifications with ad hoc methods like forward selection.

\subsection{Set-up and Data Generation}
We consider eight simulation scenarios, each based on a true semi-Markov illness death model with gamma frailty variance $e^\sigma = 0.5$. The eight scenarios arise as all combinations of two specifications for each of the baseline hazard functions, the overall covariate dimensionality and the values of the true regression parameters. The specifications under consideration are detailed in Table~\ref{tab:simsettings} of Web Appendix~\ref{sec:addlsims}, and summarized below. We repeated simulations under the given settings for three sample sizes: $n = 300, 500, 1000$.

The true baseline hazard specifications were piecewise constant with breakpoints at 5, 15, and 20, specified to yield particular marginal event rates for the non-terminal event. Under the `Low Non-Terminal Event Rate' setting approximately 17\% of subjects are observed to experience the non-terminal event, while under the `Moderate Non-Terminal Event Rate' setting this number was 30\%. Both specifications represent complex non-monotonic hazards not be well-approximated by Weibull parameterization, to examine the impact of such misspecification on regression parameter selection and estimation error.

We considered both low- and high-dimensional regimes under sparsity, with 25 and 350 covariates respectively, having 10 true non-zero coefficients in each submodel ranging in magnitude from 0.2 to 1. Crucially, the high-dimensional setting always has more regression parameters than observations, as $d=350\times 3=1050 > n$. Each simulated covariate vector $\bX_i$ was a centered and unit-scaled multivariate normal, with AR($0.25$) serial collinearity. To assess the performance of the fusion penalty, we lastly varied the extent of shared covariate effects. Under the `Shared Support' specification, the support of the non-zero effects is the same across submodels, whereas under the `Partially Non-Overlapping Support' structure the supports only partially overlap.

\subsection{Analyses}
Under each scenario, we generated 300 simulated datasets. Each dataset was then analyzed using both Lasso and SCAD-penalized models, each with and without additional fusion $\ell_1$-penalties linking all three hazards. Each analysis was performed using both Weibull and piecewise constant baseline hazard specifications. For the latter, we set  $k_g=3$ and chose breakpoints at quantiles of the data, so they also did not overlap exactly with the true data generating mechanism. In all cases, penalized models were fit over a grid comprising $21$ values for $\lambda_1$ in the high-dimensional setting and $29$ in the low-dimensional setting, and $4$ values for $\lambda_2$, leading to overall regularization grids of $21\times 4=84$ and $29\times 4=116$ points, respectively. At each grid point, the best estimate was selected from initializations at the previous step's solution, and 5 additional randomized starting values. 

From each fitted regularization grid, two models are reported. A model without fusion was chosen that minimizes the BIC over the subset path $(\lambda_1,0)$, and a model possibly with fusion was chosen which minimized BIC over the entire grid of values $(\lambda_1,\lambda_2)$. Therefore, the model space with fusion encompasses the model space without fusion, so any differences between the reported estimates with and without fusion reflect improvements in the BIC due to the added fusion penalty. If fusion did not improve BIC, there would be no difference.

For comparison, we considered a forward selection procedure minimizing BIC by adding one covariate to one transition hazard at each step. Finally, we fit the `oracle' MLE on the set of true non-zero coefficients, as well as the full MLE in the low-dimensional setting.

\subsection{Results} \label{subsec:simresults}
To assess estimation performance, we examine $\ell_2$-error defined as $\|\bbetahat-\bbeta^*\|_2$ in Table~\ref{tab:l2errorwb}. Across all settings, the estimation error of the regression coefficients was insensitive to the model baseline hazard specification, with comparable results for both Weibull and piecewise constant specifications. Therefore we present results for Weibull models, with piecewise specification results given in Web Appendix~\ref{sec:addlsims}.
For both $n=500$ and $n=1000$, the combination of SCAD and fusion penalties outperforms all comparators, particularly in the high-dimensional regime. Forward selection and the unfused SCAD-penalized estimator generally yielded the next best results.
Estimators with fusion penalization also performed better in the `Low Non-Terminal Event Rate' setting, likely because fusion links estimates across submodels, allowing `borrowing' of information about $h_1$ and $h_3$ from $h_2$ when the non-terminal event is rare.
Fusion penalized estimators also performed comparably well even if the covariate supports of each submodel only partially overlapped, relative to complete overlap.
However, Lasso penalized models did poorly relative to other comparators, which likely reflects elevated regularization-induced bias in the individual estimates.

\begin{table}
\caption{\label{tab:l2errorwb}Mean $\ell_2$ estimation error of $\bbetahat$, Weibull baseline hazard specification. Maximum likelihood estimates only available for low-dimensional setting.}
\begin{center}
\resizebox{\textwidth}{!}{\begin{tabular}{lccccccc}
\hline
$n=500$ & Oracle & MLE & Forward & Lasso & SCAD & Lasso + Fusion & SCAD + Fusion \\ 
\hline
\multicolumn{8}{l}{\textbf{Moderate Non-Terminal Event Rate}} \\ 
\multicolumn{8}{l}{\textit{~~Shared Support}} \\ 
~~~~Low-Dimension & 0.75 & 1.34 & 1.48 & 2.06 & 1.35 & 1.49 & 0.97 \\ 
~~~~High-Dimension & 0.76 & --- & 3.10 & 2.77 & 2.37 & 2.21 & 1.20 \\ 
\multicolumn{8}{l}{\textit{~~Partially Non-Overlapping Support}} \\
~~~~Low-Dimension & 0.73 & 1.30 & 1.33 & 1.87 & 1.26 & 1.61 & 1.13 \\ 
~~~~High-Dimension & 0.74 & --- & 3.47 & 2.49 & 2.34 & 2.28 & 1.45 \\ 
\multicolumn{8}{l}{\textbf{Low Non-Terminal Event Rate}} \\ 
\multicolumn{8}{l}{\textit{~~Shared Support}} \\ 
~~~~Low-Dimension & 0.92 & 1.81 & 1.89 & 2.24 & 1.91 & 1.48 & 1.20 \\ 
~~~~High-Dimension & 0.88 & --- & 5.13 & 2.56 & 2.23 & 2.34 & 1.28 \\ 
\multicolumn{8}{l}{\textit{~~Partially Non-Overlapping Support}} \\
~~~~Low-Dimension & 0.80 & 1.53 & 1.55 & 2.05 & 1.50 & 1.71 & 1.27 \\ 
~~~~High-Dimension & 0.80 & --- & 3.65 & 2.42 & 2.20 & 2.33 & 1.55 \\ 
\hline
$n=1000$ & Oracle & MLE & Forward & Lasso & SCAD & Lasso + Fusion & SCAD + Fusion \\ 
\hline
\multicolumn{8}{l}{\textbf{Moderate Non-Terminal Event Rate}} \\ 
\multicolumn{8}{l}{\textit{~~Shared Support}} \\ 
~~~~Low-Dimension & 0.50 & 0.82 & 0.76 & 1.53 & 0.73 & 1.39 & 0.75 \\ 
~~~~High-Dimension & 0.50 & --- & 1.25 & 2.40 & 1.52 & 1.82 & 0.81 \\ 
\multicolumn{8}{l}{\textit{~~Partially Non-Overlapping Support}} \\
~~~~Low-Dimension & 0.48 & 0.81 & 0.71 & 1.23 & 0.71 & 1.32 & 0.83 \\ 
~~~~High-Dimension & 0.49 & --- & 1.18 & 2.17 & 1.20 & 1.84 & 0.97 \\ 
\multicolumn{8}{l}{\textbf{Low Non-Terminal Event Rate}} \\ 
\multicolumn{8}{l}{\textit{~~Shared Support}} \\ 
~~~~Low-Dimension & 0.57 & 0.95 & 1.31 & 2.05 & 1.15 & 1.37 & 0.83 \\ 
~~~~High-Dimension & 0.58 & --- & 1.86 & 2.39 & 2.10 & 1.80 & 0.90 \\ 
\multicolumn{8}{l}{\textit{~~Partially Non-Overlapping Support}} \\
~~~~Low-Dimension & 0.52 & 0.88 & 0.85 & 1.52 & 0.84 & 1.43 & 0.96 \\ 
~~~~High-Dimension & 0.52 & --- & 1.42 & 2.22 & 1.80 & 1.94 & 1.06 \\ 
\hline
\end{tabular}}
\end{center}
\end{table}

To assess selection performance, Table~\ref{tab:signincwb} reports mean sign inconsistency, which counts the estimated regression coefficients that do not have the correct sign---exclusion of true non-zero coefficients, inclusion of true zero coefficients, or estimates having the opposite sign of the true coefficient. Lower values indicate better overall model selection performance. Again performance was very similar between Weibull and piecewise constant specifications, so only Weibull results are presented in the main text. Additional simulation results, and separated results on false inclusions and exclusions are included in Web Appendix~\ref{sec:addlsims}.

\begin{table}
\caption{\label{tab:signincwb}Mean count of sign-inconsistent $\bbetahat$ estimates, Weibull baseline hazard specification. Sign inconsistency counts the number of estimated regression coefficients that do not have the correct sign---exclusion of true non-zero coefficients, inclusion of true zero coefficients, or estimates having the opposite sign of the true coefficient.}
\begin{center}
\resizebox{\textwidth}{!}{\begin{tabular}{lcccccc}
\hline
$n=500$ & Oracle & Forward & Lasso & SCAD & Lasso + Fusion & SCAD + Fusion \\ 
\hline
\multicolumn{7}{l}{\textbf{Moderate Non-Terminal Event Rate}} \\ 
\multicolumn{7}{l}{\textit{~~Shared Support}} \\ 
~~~~Low-Dimension & 0.13 & 11.51 & 15.06 & 10.28 & 11.90 & 3.45 \\ 
~~~~High-Dimension & 0.12 & 35.81 & 26.40 & 35.22 & 21.89 & 18.88 \\ 
\multicolumn{7}{l}{\textit{~~Partially Non-Overlapping Support}} \\
~~~~Low-Dimension & 0.14 & 10.73 & 15.39 & 10.19 & 15.06 & 7.91 \\ 
~~~~High-Dimension & 0.13 & 34.56 & 27.21 & 38.52 & 23.37 & 24.35 \\ 
\multicolumn{7}{l}{\textbf{Low Non-Terminal Event Rate}} \\ 
\multicolumn{7}{l}{\textit{~~Shared Support}} \\ 
~~~~Low-Dimension & 0.30 & 16.23 & 19.33 & 17.16 & 12.77 & 5.45 \\ 
~~~~High-Dimension & 0.29 & 39.50 & 26.08 & 24.94 & 23.39 & 14.42 \\ 
\multicolumn{7}{l}{\textit{~~Partially Non-Overlapping Support}} \\
~~~~Low-Dimension & 0.24 & 12.70 & 17.11 & 12.42 & 16.33 & 9.47 \\ 
~~~~High-Dimension & 0.21 & 36.20 & 26.08 & 29.51 & 24.53 & 22.89 \\ 
\hline
$n=1000$ & Oracle & Forward & Lasso & SCAD & Lasso + Fusion & SCAD + Fusion \\ 
\hline
\multicolumn{7}{l}{\textbf{Moderate Non-Terminal Event Rate}} \\ 
\multicolumn{7}{l}{\textit{~~Shared Support}} \\ 
~~~~Low-Dimension & 0.01 & 4.03 & 13.52 & 3.99 & 8.48 & 1.51 \\ 
~~~~High-Dimension & 0.00 & 15.83 & 20.92 & 19.92 & 19.82 & 7.21 \\ 
\multicolumn{7}{l}{\textit{~~Partially Non-Overlapping Support}} \\
~~~~Low-Dimension & 0.02 & 4.12 & 13.80 & 4.27 & 12.58 & 4.00 \\ 
~~~~High-Dimension & 0.03 & 15.57 & 19.86 & 19.30 & 21.28 & 8.44 \\ 
\multicolumn{7}{l}{\textbf{Low Non-Terminal Event Rate}} \\ 
\multicolumn{7}{l}{\textit{~~Shared Support}} \\ 
~~~~Low-Dimension & 0.06 & 9.70 & 16.60 & 9.01 & 8.43 & 1.86 \\ 
~~~~High-Dimension & 0.06 & 23.73 & 24.30 & 22.17 & 19.08 & 6.95 \\ 
\multicolumn{7}{l}{\textit{~~Partially Non-Overlapping Support}} \\
~~~~Low-Dimension & 0.03 & 5.29 & 14.54 & 5.74 & 13.27 & 4.97 \\ 
~~~~High-Dimension & 0.03 & 16.96 & 21.98 & 22.18 & 21.98 & 8.72 \\ 
\hline
\end{tabular}}
\end{center}
\end{table}

For both $n=500$ and $n=1000$, the combination of SCAD and fusion penalties outperformed comparators, while other methods' performances varied across sample size and setting. Fusion penalized estimators exhibited notably better selection properties in the `Shared Support' setting, as fusion coerces a common block of non-zero covariates across submodels. Lasso penalized models tended to choose overly sparse models. With many of the true non-zero effects small in magnitude, regularization-induced bias may have rendered those terms indistinguishable from truly zero effects.

Lastly we summarize the results in the smallest sample setting of $n=300$, which are detailed in Web Appendix~\ref{sec:addlsims}. This is a challenging setting because small samples exacerbate the non-convexity of the marginal illness-death likelihood, and when outcomes are rare some transition submodels have a small number of observed events.
These complications affect small-sample empirical performance of frailty-based illness-death models even in the absence of high-dimensional covariates. For example, in the setting with 25 covariates per transition, average estimation error of the full MLE is substantially larger for $n=300$ than $n=500$, and more sensitive to the non-terminal event rate (see Web Table~\ref{tab:l2errorsmall}). Still, in this low-dimensional regime we again observe the combination of SCAD and fusion penalties reducing estimation error relative to comparators. 

However, increasing dimensionality to 350 covariates per transition while keeping $n=300$ degraded performance of all methods, particularly when event rates were lowest. For example, in the `Low Non-Terminal Event Rate' settings, the comparator forward selection algorithm failed for between 25 and 90 percent of simulations by adding so many covariates that optimization no longer converged. Penalized models also tended towards extremes, with SCAD-penalized models including many unnecessary covariate effects, while the Lasso models selected few or no non-zero covariate effects. A key challenge is that likelihood-based selection criteria like BIC can be distorted in small samples by non-convexity. In certain instances, the log-likelihood can become monotonic with respect to the frailty log-variance $\sigma$, yielding artificial information criteria and leading to selected models that are either completely sparse (as with Lasso-penalized models) or completely saturated (as with SCAD-penalized and forward-selected models). These serious complications manifested only in the most difficult settings combining small samples, low-to-moderate event rates, and high dimensional covariates, but show that the challenges of very small-sample estimation with frailties is compounded by high-dimensional covariates. 


\section{Data Application: Preeclampsia (PE) and Delivery} \label{sec:data}

The proposed methods are motivated by practical application to clinical settings where interest is in developing a risk model that jointly characterizes a non-terminal and terminal event. To this end, we consider modeling PE onset and the timing of delivery using the electronic health records of an urban, academic medical center in Boston, Massachusetts. We analyze 2127 singleton live births recorded in 2019 among individuals without pre-existing hypertension who received the majority of their prenatal care and delivered at the academic medical center. Restricting to those without hypertension targets the modeling task, as PE superimposed on chronic hypertension has distinct clinical features compared to other forms of the disease \citep{jeyabalan2013epidemiology}. $189$ (8.9\%) individuals developed PE, with median diagnosis time of 37.9 weeks (Inter-Quartile Range [IQR] 35.0-39.0). The median time to delivery was 38.0 weeks (IQR 35.4-39.3) among those who developed PE and 39.4 weeks (IQR 38.6-40.3) among those who did not. Note, because PE is only diagnosed after 20 weeks of gestation, for modeling purposes this is used as the time origin, with $T_1$ and $T_2$ defined as time from week 20 until PE onset and delivery, respectively. 

We considered a set of $33$ potential covariates, including demographics recorded at patient intake, baseline lab values annotated by the medical center with a binary indicator for abnormality, and maternal health history derived from ICD-10 diagnostic codes associated with delivery (summarized in Web Table~\ref{tab:table1}). We fit Markov illness-death models so that $\bbeta_2$ and $\bbeta_3$ are both interpretable on the gestational age timescale, under both Weibull and piecewise constant baseline hazards. We adopted SCAD penalties on each regression coefficient, and an $\ell_1$ fusion penalty between $\bbeta_2$ and $\bbeta_3$ to induce a shared structure between coefficients for the timing of delivery in the absence of PE and timing of delivery given the onset of PE. We specified a grid of 55 values for $\lambda_1$ and seven for $\lambda_2$. As above, we selected the final penalized model minimizing BIC over the entire grid of $(\lambda_1,\lambda_2)$ values, while the final SCAD-only penalized model minimizes BIC over the path $(\lambda_1,0)$. Again, this means that differences between the models with and without fusion reflect improvements in BIC due to the fused penalty. For comparison, we also report the unpenalized MLE. 

Figure~\ref{fig:PE_m_bic} compares the estimated regression coefficients between unpenalized and penalized models, for each baseline hazard specification. As in the simulation studies, there appears to be little difference in the selection properties or resulting regression estimates between models with Weibull and piecewise constant baseline hazard specifications. Across all specifications, the penalized estimates chosen by BIC are highly sparse relative to the unpenalized MLE. The inclusion of a fusion penalty linking the coefficients in $h_2$ and $h_3$ further improved BIC, and fused several covariate effects for the timing of delivery before PE and given PE. Specifically, parity of 1 or more (meaning a history of at least one pregnancy lasting at least 20 weeks), and in the Weibull model, the presence of leiomyomas (benign gynecological tumors), are both estimated with shared coefficients on timing of delivery with and without PE. For comparison, models chosen by AIC are provided in Web Figure~\ref{fig:PE_m_aic}, which include more covariates but are still sparse relative to the full model. Finally, we find that in this application, beyond the estimated regression coefficients frailties did not play a large role in characterizing additional dependence between preeclampsia and delivery timing, as the estimated frailty variance is very close to 0 in all estimates.

\begin{figure}
\begin{center}
\includegraphics[width=\textwidth]{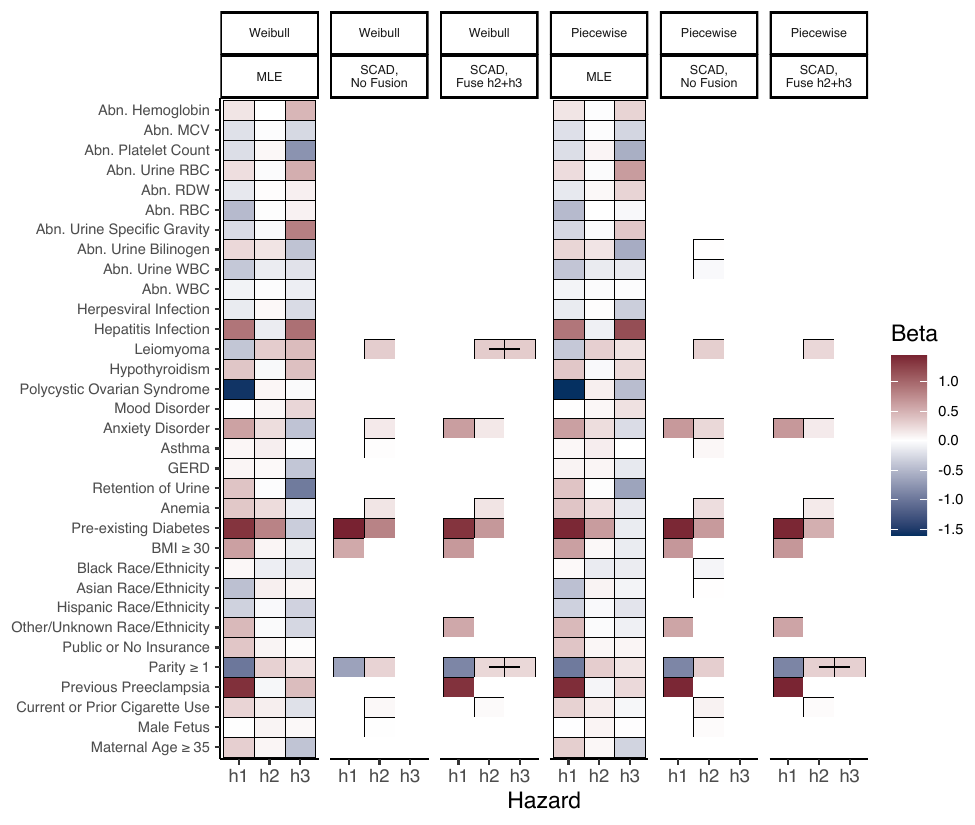}
\end{center}
\caption{Estimated coefficients, BIC-optimal SCAD-penalized estimators with and without $\ell_1$ fusion between $h_2$ and $h_3$, and MLE under Markov specification. Fused coefficients connected with a black line. Abbreviations: abnormal (Abn), white blood cell count (WBC), red blood cell count (RBC), red cell distribution width (RDW), mean corpuscular volume (MCV), gastroesophageal reflux disease (GERD). This figure appears in color in the electronic version of this article, and any mention of color refers to that version.\label{fig:PE_m_bic}}
\end{figure}

These results have strong clinical significance. In every specification and across all three hazards, the selected covariates are primarily maternal health history and behaviors, rather than demographics or baseline lab measurements. Many of the variables selected for the cause-specific hazard of PE---parity of 1 or more, BMI of at least 30, and pre-existing diabetes---align with findings of recent meta-analyses of factors affecting PE risk \citep{giannakou2018genetic}. Further illustrating the interplay of risk factors with the outcomes, fusion penalized estimates show parity of 1 or more associated both with delayed timing of PE, and accelerated timing of delivery in the presence of PE. This correspond clinically with risk of milder late-onset PE for which delivery can occur quickly with fewer risks.

As introduced previously, care decisions for PE center two challenges: identifying those at high risk of PE, and timing delivery after PE onset to balance maternal and fetal health risks. Though our methodological focus is on regularized estimation and model selection, the resulting fitted illness-death models also generate prospective risk predictions to inform these individualized clinical care decisions \citep{putter2007tutorial}. In Web Appendix~\ref{sec:jointriskpred} we present and discuss a set of four such risk profiles for sample patients using a Weibull model with fusion penalty. Specifically, from baseline the model can predict across time how likely an individual is to be in one of four categories: (i) still pregnant without PE, (ii) already delivered without PE, (iii) already delivered with PE, and (iv) still pregnant with PE. Such profiles directly address clinical needs by highlighting individuals' overall risk of developing PE, while also characterizing the timing of PE and delivery.


\section{Discussion}\label{sec:disc}

Frailty-based illness-death models enable investigation of the complex interplay between baseline covariates and semi-competing time-to-event outcomes. Estimates directly illustrate the relationships between risk factors and the joint outcomes via hazard ratios across three submodels, while individualized risk predictions generate an entire prospective outcome trajectory to inform nuanced clinical care decisions. 

The task of modeling risk of PE and timing of delivery illustrates the value, and the potential, for penalized illness-death modeling to inform clinical practice. While analysts typically default to including the same set of covariates in all three hazards, Figure~\ref{fig:PE_m_bic} illustrates that no covariate in any BIC-selected models has distinct, non-shared coefficients in all three hazards. Moreover, even in the setting of PE onset and delivery, where relatively few covariate effects appear to be shared across submodels, adding fusion regularization improved model fit metrics. We expect the impact of fusion would be even more pronounced in settings where the outcomes are more positively correlated, such as when the non-terminal event is a negative health outcome and the terminal event is death. Because frailties also tend to characterize positive correlation of the outcomes, we might also expect larger estimated frailty variance in such settings. Analysts interested in considering non-frailty models might fit regularization paths with and without the frailty, and choose a final criterion-minimizing model from amongst both frailty and non-frailty candidates.

We also note that the statistical rate result of Theorem~\ref{th:main} uses the specific choice of penalty given in \eqref{eq:optimtheory}, however we would expect similar theoretical performance under the similar penalty introduced in \eqref{eq:genpen}. The advantage of implementing the formulation as in \eqref{eq:genpen} is its interpretability for the analyst, by directly distinguishing between the role of $\lambda_1$ in determining the global level of sparsity of the regression parameters, and the role of $\lambda_2$ in determining the level of parsimony in the sharing of effects across hazards.

Though the current work focuses on penalization of the regression parameters $\bbeta$, the framework also admits penalization of the baseline hazard parameters to achieve similar goals of flexibility and structure. For example, under the Markov transition specification a penalty of the form $\sum_{j=1}^{k_3} p_{\lambda_3}(\lvert\phi_{2j}-\phi_{3j}\rvert)$ could regularize the model towards having $h_{02}(t_2)=h_{03}(t_2)$. \citet{xu2010statistical} call this the `restricted' illness-death model corresponding to the baseline hazard of the terminal event being equal before and after the non-terminal event. Moreover, while we presently focus on fixed-dimensional parametric baseline hazard specifications, in principle the estimation algorithms presented here extend to penalized baseline models of growing dimensionality, such as splines with number of basis functions dependent on sample size. The theoretical properties of such an estimator would be an interesting avenue of future research. Future work might also explore these methods and theory under other frailty distributions besides the closed form-inducing gamma.

Finally, establishing the statistical rate of the proposed penalized estimator also enables future development of post-selection inferential tools such as confidence intervals for selected coefficients. Most importantly, this methodology enables future work modeling semi-competing risks across a wide array of clinical domains, and leveraging data sources with high-dimensional covariates from electronic health records to genomic data.

\section*{Acknowledgments}
The authors thank the Associate Editor and Reviewers for helpful comments, and thank Drs. Michele Hacker, Anna Modest, and Stefania Papatheodorou for data access at Beth Israel Deaconess Medical Center, Boston, MA and valuable discussion. This study received exemption status by institutional review boards at Beth Israel Deaconess Medical Center and Harvard TH Chan School of Public Health. HTR is supported by NIH grants T32LM012411, T32GM074897, and F31HD102159, and JL by R35CA220523 and U01CA209414. JL is supported by NSF grant 1916211.

\section*{Data Availability Statement}
Data used in this paper to illustrate the proposed methods are not shared due to privacy restrictions.

\bibliographystyle{biom}
\bibliography{../../sources}

\newpage

\appendix

\makeatletter
\renewcommand{\thetable}{\thesection.\@arabic\c@table}
\@addtoreset{table}{section}
\makeatother

\makeatletter
\renewcommand{\thefigure}{\thesection.\@arabic\c@figure}
\@addtoreset{figure}{section}
\makeatother

\makeatletter
\renewcommand{\theequation}{\thesection.\@arabic\c@equation}
\@addtoreset{equation}{section}
\makeatother

\makeatletter
\renewcommand{\thelemma}{\thesection.\@arabic\c@lemma}
\@addtoreset{lemma}{section}
\makeatother

\vspace{-0.5in}
\section*{Appendix Introduction} \label{sec:appintro}

In this appendix we present additional details and results beyond what could be presented in the main manuscript. To distinguish the two documents, alpha-numeric labels are used in this document while numeric labels are used in the main paper. Section~\ref{sec:addldata} provides additional results from the data application. Section~\ref{sec:jointriskpred} describes the use of fitted illness-death models for individualized risk prediction, and presents an example using the data application.  Section~\ref{sec:proof} provides the proof of Theorem~\ref{th:main}. Section~\ref{sec:gradhess} defines likelihood-related functions for frailty-based parametric illness death models. Section~\ref{sec:lemmas} defines and proves technical lemmas used in the main result. Section~\ref{sec:addlsims} presents additional simulation details and results. Section~\ref{sec:algodetails} presents additional algorithmic details for optimization. Section~\ref{sec:splines} summarizes other spline-based hazard specifications. 

\section{Additional Data Application Results}\label{sec:addldata}
In this section, we present further additional results for the data application.

\begin{table}[H]
\caption{Characteristics of the study population, overall and by observed preeclampsia diagnosis and delivery outcome. Abbreviations: abnormal (Abn), white blood cell count (WBC), red blood cell count (RBC), red cell distribution width (RDW), mean corpuscular volume (MCV), gastroesophageal reflux disease (GERD).\label{tab:table1}}
\begin{center}
\resizebox{\textwidth}{!}{\begin{tabular}{lccc}
  \hline
 & Total Births & Births with Preeclampsia & Births without Preeclampsia\\ 
  \hline
Total & 2127 (100\%) & 189 (100\%) & 1938 (100\%) \\
  \hline
Maternal Age $\geq$ 35 & 658 (30.9\%) & 61 (32.3\%) & 597 (30.8\%) \\ 
  Male Fetus & 1074 (50.5\%) & 94 (49.7\%) & 980 (50.6\%) \\ 
  Current or Prior Cigarette Use & 232 (10.9\%) & 29 (15.3\%) & 203 (10.5\%) \\ 
  Previous Preeclampsia & 51 (2.4\%) & 10 (5.3\%) & 41 (2.1\%) \\ 
  Parity $\geq$ 1 & 1075 (50.5\%) & 58 (30.7\%) & 1017 (52.5\%) \\ 
  Public or No Insurance & 761 (35.8\%) & 81 (42.9\%) & 680 (35.1\%) \\ 
  Other/Unknown Race/Ethnicity & 431 (20.3\%) & 54 (28.6\%) & 377 (19.5\%) \\ 
  Hispanic Race/Ethnicity & 194 (9.1\%) & 12 (6.3\%) & 182 (9.4\%) \\ 
  Asian Race/Ethnicity & 209 (9.8\%) & 10 (5.3\%) & 199 (10.3\%) \\ 
  Black Race/Ethnicity & 310 (14.6\%) & 29 (15.3\%) & 281 (14.5\%) \\ 
  BMI $\geq$ 30 & 550 (25.9\%) & 75 (39.7\%) & 475 (24.5\%) \\ 
  Pre-existing Diabetes & 49 (2.3\%) & 14 (7.4\%) & 35 (1.8\%) \\ 
  Anemia & 283 (13.3\%) & 31 (16.4\%) & 252 (13\%) \\ 
  Retention of Urine & 49 (2.3\%) & 5 (2.6\%) & 44 (2.3\%) \\ 
  GERD & 96 (4.5\%) & 11 (5.8\%) & 85 (4.4\%) \\ 
  Asthma & 175 (8.2\%) & 22 (11.6\%) & 153 (7.9\%) \\ 
  Anxiety Disorder & 166 (7.8\%) & 25 (13.2\%) & 141 (7.3\%) \\ 
  Mood Disorder & 125 (5.9\%) & 15 (7.9\%) & 110 (5.7\%) \\ 
  Polycystic Ovarian Syndrome & 30 (1.4\%) & 2 (1.1\%) & 28 (1.4\%) \\ 
  Hypothyroidism & 140 (6.6\%) & 21 (11.1\%) & 119 (6.1\%) \\ 
  Leiomyoma & 269 (12.6\%) & 14 (7.4\%) & 255 (13.2\%) \\ 
  Hepatitis Infection & 27 (1.3\%) & 7 (3.7\%) & 20 (1\%) \\ 
  Herpesviral Infection & 101 (4.7\%) & 8 (4.2\%) & 93 (4.8\%) \\ 
  Abn. WBC & 462 (21.7\%) & 41 (21.7\%) & 421 (21.7\%) \\ 
  Abn. Urine WBC & 95 (4.5\%) & 6 (3.2\%) & 89 (4.6\%) \\ 
  Abn. Urine Bilinogen & 47 (2.2\%) & 5 (2.6\%) & 42 (2.2\%) \\ 
  Abn. Urine Specific Gravity & 10 (0.5\%) & 1 (0.5\%) & 9 (0.5\%) \\ 
  Abn. RBC & 322 (15.1\%) & 19 (10.1\%) & 303 (15.6\%) \\ 
  Abn. RDW & 79 (3.7\%) & 8 (4.2\%) & 71 (3.7\%) \\ 
  Abn. Urine RBC & 88 (4.1\%) & 8 (4.2\%) & 80 (4.1\%) \\ 
  Abn. Platelet Count & 55 (2.6\%) & 4 (2.1\%) & 51 (2.6\%) \\ 
  Abn. MCV & 153 (7.2\%) & 14 (7.4\%) & 139 (7.2\%) \\ 
  Abn. Hemoglobin & 187 (8.8\%) & 15 (7.9\%) & 172 (8.9\%) \\ 
   \hline
\end{tabular}}
\end{center}
\end{table}

\subsection{Additional Estimation Results}

\begin{figure}[H]
\begin{center}
\includegraphics[width=6.5in]{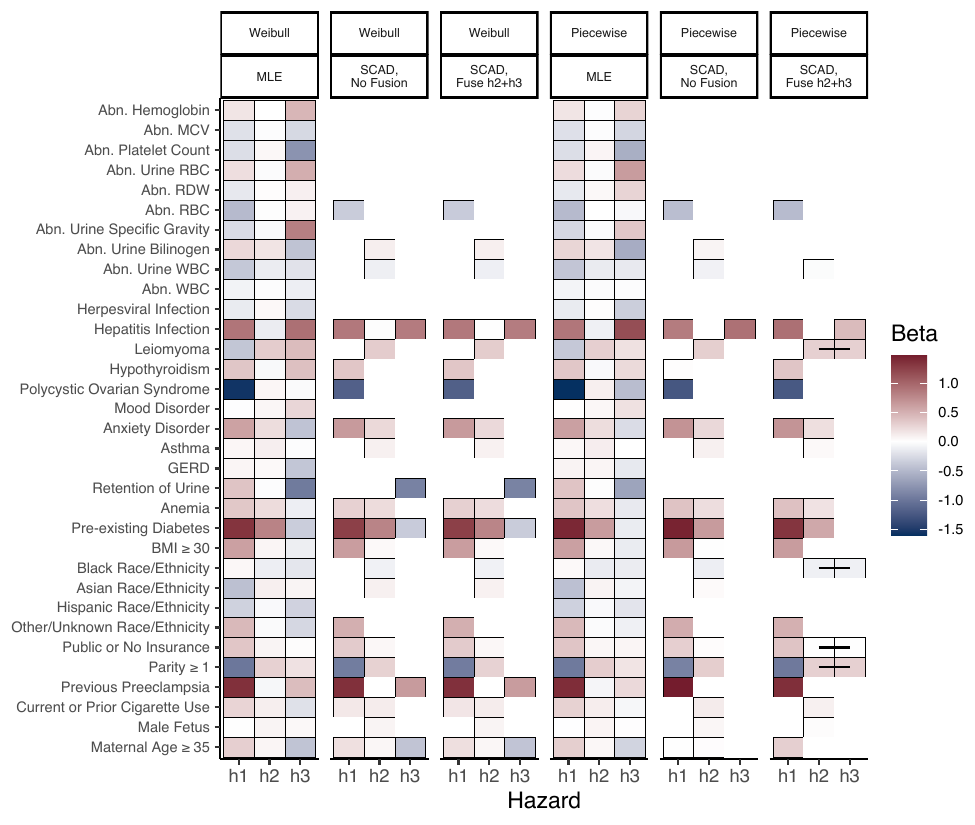}
\end{center}
\caption{Estimated coefficients, AIC-optimal SCAD-penalized estimators with and without $\ell_1$ fusion between $h_2$ and $h_3$, and MLE under Markov specification. Fused coefficients connected with a black line. \label{fig:PE_m_aic}}
\end{figure}

\begin{table}[H]
\caption{Information criterion values for models presented in Figures~\ref{fig:PE_m_bic} (BIC) and \ref{fig:PE_m_aic} (AIC). Lower values reflect improved model fit.\label{tab:infocrit}}
\begin{center}
\begin{tabular}{lccc}
  \hline
\textit{BIC-Optimal Selection} & $\lambda_1$ & $\lambda_2$ & BIC \\ 
  \hline
~~Weibull MLE & 0 & 0 & -219.039  \\ 
~~Weibull, SCAD, No Fusion & 0.027 & 0& -762.464  \\
~~Weibull, SCAD, Fuse $h_2+h_3$ & 0.016 & 0.015 & -784.840  \\
~~Piecewise MLE & 0 & 0 & 93.984  \\  
~~Piecewise, SCAD, No Fusion & 0.017 & 0 & -467.450  \\
~~Piecewise, SCAD, Fuse $h_2+h_3$ & 0.017 & 0.015 & -487.490 \\ 
   \hline
\textit{AIC-Optimal Selection} & $\lambda_1$ & $\lambda_2$ & AIC   \\
\hline
~~Weibull MLE & 0 & 0 & -813.599 \\ 
~~Weibull, SCAD, No Fusion & 0.009 & 0 & -907.093 \\
~~Weibull, SCAD, Fuse $h_2+h_3$ & 0.009 & 0.001 & -907.164 \\
~~Piecewise MLE  & 0 & 0 & -568.524 \\  
~~Piecewise, SCAD, No Fusion & 0.011 & 0 & -668.483 \\
~~Piecewise, SCAD, Fuse $h_2+h_3$ & 0.010 & 0.010 & -675.313 \\ 
   \hline
\end{tabular}
\end{center}
\end{table}

\newpage
\section{Individualized Joint Outcome Risk Prediction}
\label{sec:jointriskpred}

Fitted illness-death models can be used to generate clinically meaningful predictions of individualized joint risk, and the timing of the non-terminal and terminal events. In the preeclampsia setting, the model can predict across time how likely an individual is to be in one of four categories: 
\begin{itemize}
\item[(1)] still pregnant without preeclampsia
\item[(2)] already delivered without preeclampsia
\item[(3)] already delivered with preeclampsia, and 
\item[(4)] still pregnant with preeclampsia. 
\end{itemize}
These four probabilities comprise an individualized “risk profile,” and are derived for the illness-death model in \citet{putter2007tutorial} by integrating over regions of the joint density of the semi-competing outcomes $(T_1,T_2)$. The resulting formulas can be concisely represented for both Markov and semi-Markov illness-death models, by defining
\begin{equation}
H_3(t\mid t_1,\bX_{3}) = \begin{cases} H_3(t\mid \bX_{3}) - H_3(t_1\mid \bX_{3}) & \text{Markov} \\ H_3(t - t_1\mid\bX_{3}) & \text{semi-Markov}. \end{cases}
\end{equation}
Then for fixed frailty $\gamma$, these four probabilities (numbered as above) are:
\begin{align}
	\label{eqn:fitprofile:1}
	\pi^{(1)}(t\mid \bX, \gamma) = & \exp\{-\gamma[H_1(t\mid \bX_1) + H_2(t\mid \bX_2)]\} \\
	\label{eqn:fitprofile:2}
	\pi^{(2)}(t\mid \bX, \gamma) = & \int_0^t  \gamma h_2(t_2\mid \bX_2)\exp\{-\gamma[H_1(t_2\mid \bX_1) + H_2(t_2\mid \bX_2)]\} dt_2  \\
\begin{split}\label{eqn:fitprofile:3}
	\pi^{(3)}(t\mid \bX, \gamma) = &  \int_0^t  \gamma h_1(t_1\mid \bX_1) \exp \{ -\gamma [H_1(t_1\mid \bX_1) + H_2(t_1\mid \bX_2)]\} \\
	& \times (1-\exp \{-\gamma H_3(t\mid t_1,\bX_{3})\}) dt_1
\end{split}\\
	\label{eqn:fitprofile:4}
	\pi^{(4)}(t\mid \bX, \gamma) = & \int_0^t \gamma h_1(t_1\mid \bX_1) \exp\{-\gamma [H_1(t_1\mid \bX_1) + H_2(t_1\mid \bX_2) + H_3(t\mid t_1,\bX_{3})\} dt_1,
\end{align}
and collectively denoted $\bpi(t\mid  \bX, \gamma) = \{\pi^{(1)}(t\mid \bX, \gamma) ,\pi^{(2)}(t\mid \bX, \gamma) ,\pi^{(3)}(t\mid \bX, \gamma) ,\pi^{(4)}(t\mid \bX, \gamma) \}$. These probabilities sum to 1, and the integrals can be computed numerically using standard software.

\subsection{Sample Predictions from Data Application}

Example individualized risk profiles are shown in Figure~\ref{fig:riskprof}, generated from the above AIC-selected Weibull model with fusion penalty (shown above in column three of Figure~\ref{fig:PE_m_aic}). The four panels of Figure~\ref{fig:riskprof} correspond with sample individuals having covariates outlined in Table~\ref{tab:samplecov}. At each time point, the height of each colored area of the plot gives the probability that the individual will be in the corresponding outcome category at that time, stacking from top to bottom $\pi^{(1)}(t\mid\bX,\gamma)$,  $\pi^{(2)}(t\mid\bX,\gamma)$,  $\pi^{(3)}(t\mid\bX,\gamma)$, and $\pi^{(4)}(t\mid\bX,\gamma)$. For example, the model predicts that at 34 weeks of gestation, individual D has about a 10\% chance of having developed preeclampsia and still being pregnant (blue), a 10\% chance of having developed preeclampsia and already given birth (purple), a 2\% chance of having given birth without preeclampsia (red), and a 78\% chance of being pregnant without preeclampsia (grey).

From these detailed risk patterns, we can also read simple overall probabilities of developing preeclampsia by looking at the combined height of the blue and purple bars at the right end of the plot. For patients A-D, the predicted overall probability of developing preeclampsia by week 40 is about 2\%, 22\%, 52\% and 47\%, respectively.

\begin{figure}[H]
\begin{center}
\includegraphics[width=6.5in]{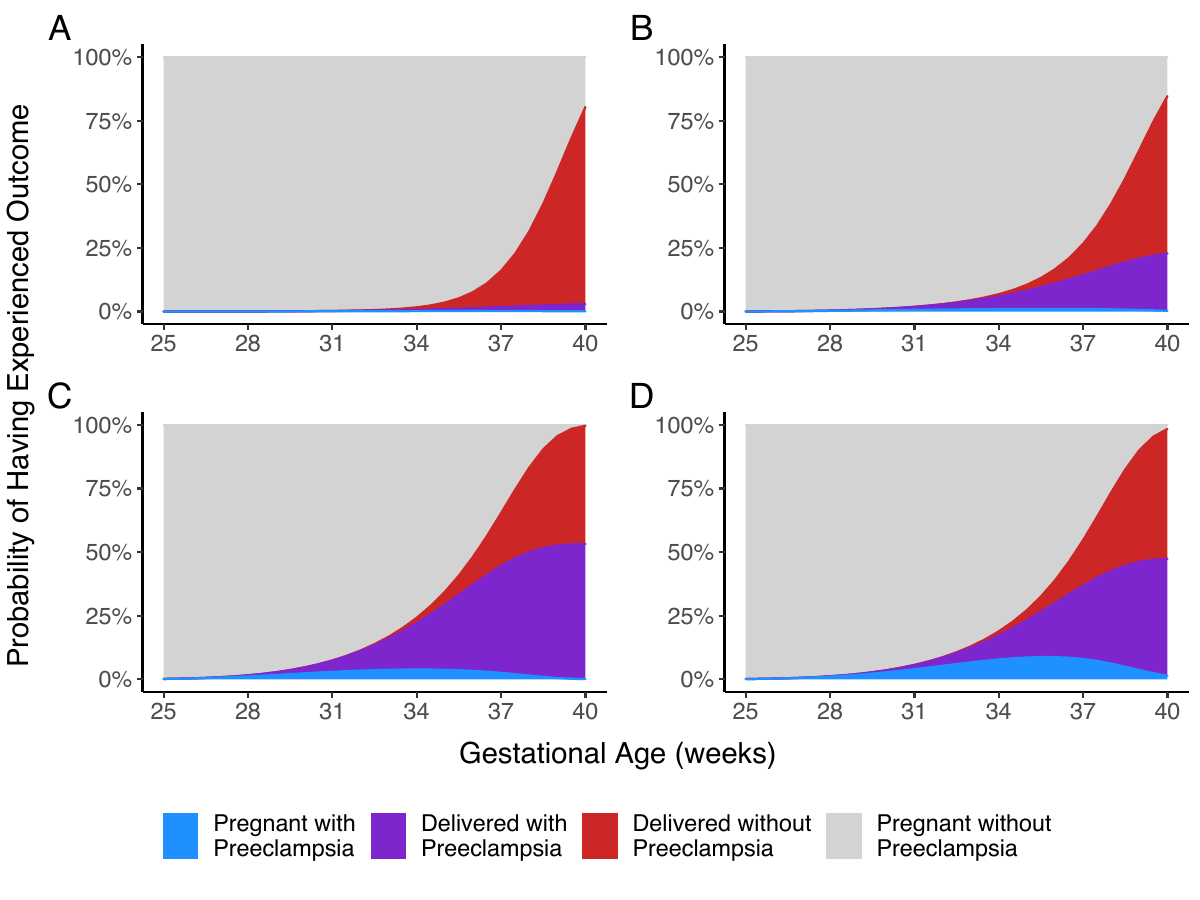}
\end{center}
\caption{Sample predicted risk profiles for four sample individuals. AIC-selected Weibull model with SCAD and fusion penalty under Markov specification. Frailty value fixed at $\gamma=1$.\label{fig:riskprof}}
\end{figure}

\begin{table}[H]
\caption{Characteristics of sample individuals used in Figure~\ref{fig:riskprof}.\label{tab:samplecov}}
\begin{center}
\begin{tabular}{rllll}
  \hline
 & A & B & C & D \\ 
  \hline
  Abn. RBC & No & No & No & No \\ 
  Abn. Urine Bilinogen & No & No & No & No \\ 
  Abn. Urine WBC & Yes & No & No & No \\ 
  Hepatitis Infection & No & No & No & No \\ 
  Leiomyoma & No & No & No & No \\ 
  Hypothyroidism & No & No & Yes & Yes \\ 
  Polycystic Ovarian Syndrome & No & No & No & No \\ 
  Anxiety Disorder & No & No & No & No \\ 
  Asthma & No & No & Yes & No \\ 
  Retention of Urine & No & No & No & Yes \\ 
  Anemia & No & Yes & No & No \\ 
  Pre-existing Diabetes & No & No & Yes & Yes \\ 
  BMI $\geq$ 30 & No & Yes & Yes & Yes \\ 
  Race/Ethnicity & White & Other/Unknown & White & White \\
  Public or No Insurance & No & No & No & Yes \\ 
  Parity $\geq$ 1 & Yes & No & Yes & No \\ 
  Previous Preeclampsia & No & No & Yes & No \\ 
  Current or Prior Cigarette Use & No & No & No & No \\ 
  Male Fetus & Yes & No & No & Yes \\ 
  Maternal Age $\geq$ 35 & Yes & No & Yes & No \\ 
   \hline
\end{tabular}
\end{center}
\end{table}

\subsection{Use of Frailties in Individualized Risk Prediction}\label{subsec:predfrail}

As described in the main text, the patient-specific frailty $\gamma$  accounts for residual within-patient dependence between $(T_1,T_2)$. Thus, it represents a potentially important component of variation in individualized risk predictions \citep{putter2015frailties}. However, $\gamma$ is latent and, therefore, cannot be observed for an individual to plug into $\bpi(t\mid  \bX, \gamma)$, and must be fixed to a chosen value as in Figure~\ref{fig:riskprof} with $\gamma=1$. In practice we might also compare predicted risk profiles for an individual across different values of $\gamma$, and in this way, the frailty can be viewed as a way of characterizing individual-level variability in the predicted risk profile. For example, Figure~\ref{fig:riskproffrailty} shows predicted risks across frailty values for each sample subject at the fixed timepoint of 37 weeks of gestation. Note that now, the x-axis does not represent time, but differing choices of frailty $\gamma$ plugged into the above risk profile formulae. Intuitively, we see that the predicted probabilities of experiencing some combination of the outcomes by week 37 tend to be smaller for smaller frailty values, and larger for larger frailty values.  Though individuals' frailties are unobserved, these plots can be used to communicate to patients the variability of possible outcome probability estimates depending on unmeasured latent characteristics \citep{lee2020timetoevent}.

Alternatively, this prediction framework also enables risk profile estimates to be marginalized over the frailty distribution by simply integrating over $f_\gamma(\gamma\mid \sigma)$, yielding the marginal profile $\bpi(t \mid  \bX) = \int \bpi(t \mid  \bX, \gamma)f_\gamma(\gamma\mid \sigma) d\gamma$. Corresponding marginal risk profile formulas derived under a gamma frailty distribution are given here for reference:
\begin{align}
	\label{eqn:fitprofile:1marg}
	\pi^{(1)}(t\mid \bX) = & \{1 + e^{\sigma}[H_1(t_2\mid \bX_1) + H_2(t_2\mid \bX_2)]\}^{-\exp(-\sigma)} \\
	\label{eqn:fitprofile:2marg}
	\pi^{(2)}(t\mid \bX) = & \int_0^t  h_2(t_2\mid \bX_2)\{1 + e^{\sigma}[H_1(t_2\mid \bX_1) + H_2(t_2\mid \bX_2)]\}^{-\exp(-\sigma)-1} dt_2  \\
\begin{split}\label{eqn:fitprofile:3marg}
	\pi^{(3)}(t\mid \bX) = &  \int_0^t  h_1(t_1\mid \bX_1) \{1 + e^{\sigma}[H_1(t_1\mid \bX_1) + H_2(t_1\mid \bX_2)]\}^{-\exp(-\sigma)-1} dt_1 \\
	& - \int_0^t h_1(t_1\mid \bX_1)\{1 + e^{\sigma}[H_1(t_1\mid \bX_1) + H_2(t_1\mid \bX_2) + H_3(t\mid t_1,\bX_3)]\}^{-\exp(-\sigma)-1} dt_1
\end{split}\\
	\label{eqn:fitprofile:4marg}
	\pi^{(4)}(t\mid \bX) = & \int_0^t h_1(t_1\mid \bX_1)\{1 + e^{\sigma}[H_1(t_1\mid \bX_1) + H_2(t_1\mid \bX_2) + H_3(t\mid t_1,\bX_3)]\}^{-\exp(-\sigma)-1} dt_1
\end{align}
Because in this data application the estimated frailty variance was small, there is little difference between the population-averaged risk profiles for the sample subject compared to Figure~\ref{fig:riskprof}, so the plot is omitted.

\begin{figure}[H]
\begin{center}
\includegraphics[width=6.5in]{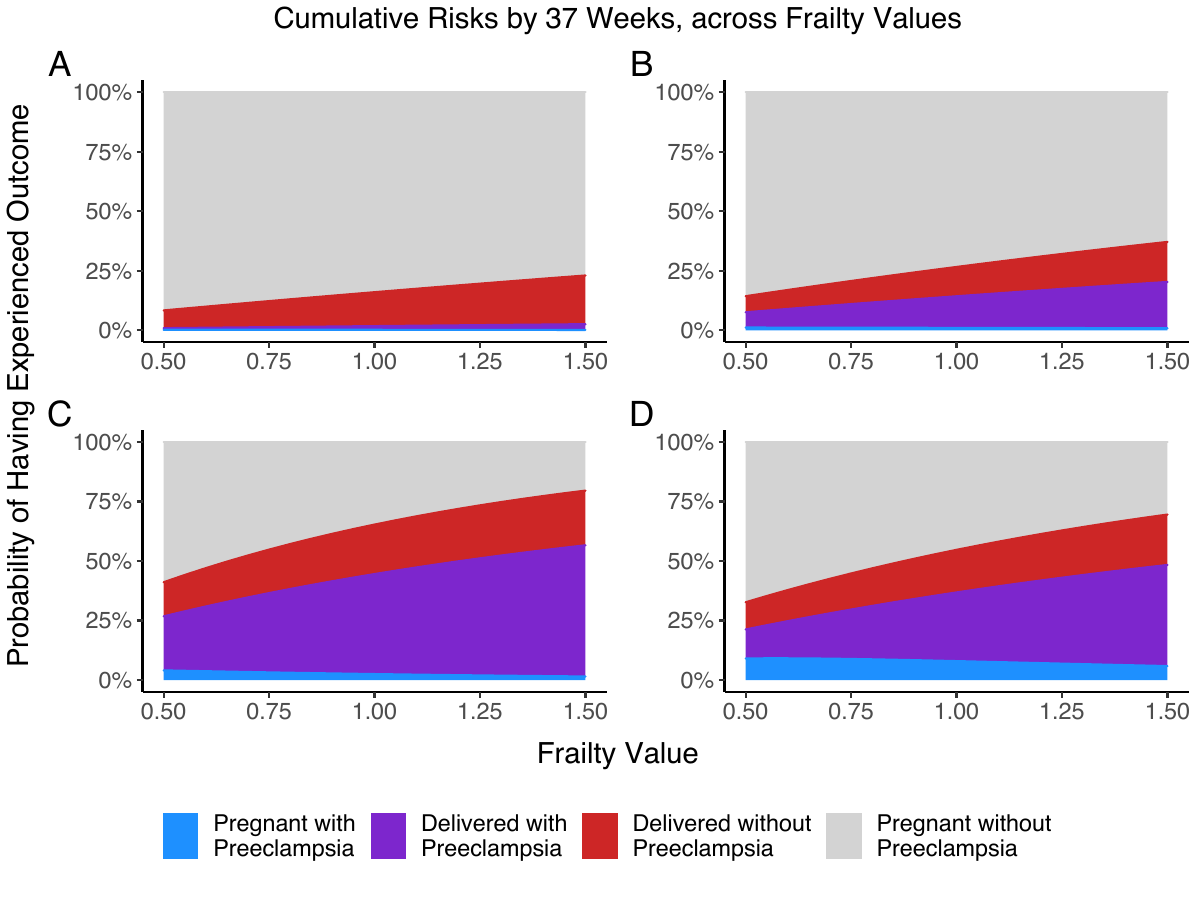}
\end{center}
\caption{Predicted risk at 37 weeks of gestation across values of frailty $\gamma$ for four sample individuals. AIC-selected Weibull model with SCAD and fusion penalty under Markov specification.\label{fig:riskproffrailty}}
\end{figure}

\newpage
\section{Proof of Main Result} \label{sec:proof}

In this section, we present a proof of Theorem~\ref{th:main}. This proof follows the strategy of Theorem 1 of  \citet{loh2015regularized}, with several differences. In particular, the restricted strong convexity (RSC) condition in \citet{loh2015regularized} no longer applies in this setting, because under certain baseline hazard specifications such as the Weibull, the resulting baseline parameters $\bphi$ have heavy-tailed deviations of the gradient and Hessian that decay at a slower rate (see Lemmas~\ref{lem:weibcond}, \ref{lem:gradbound}, and \ref{lem:hessbound}). Therefore, we introduce and apply the alternative RSC condition given in \eqref{eq:newrsc}, which in Lemma~\ref{lem:newrsc} is verified to hold with high probability under the corresponding Assumptions listed in Section~\ref{sec:theory} of the main text.

Finally, we note that the same statistical rate can be immediately obtained for other choices of penalty function $p_\lambda$ besides SCAD, as described by \citet{loh2015regularized}, Section~2.2. In particular, the theorem supposition of $3/\{4(\xi -1)\} < \rho$ for SCAD penalties relates the level of penalty non-convexity to the Hessian eigenvalue lower bound, and can be replaced by $3/\{4\xi\} < \rho$ for the MCP penalty of \citet{zhang2010nearly}, or omitted entirely for the lasso penalty. The theorem can be correspondingly adjusted below.

\begin{proof}[Proof of Theorem~\ref{th:main}]

To begin, define the linear transformation of the regression parameters $\bbeta' = ( (\bbeta'_{1})\trans, (\bbeta'_{2})\trans, (\bbeta'_{3})\trans ) = ( \bbeta_{1}\trans, (\bbeta_{2} - \bbeta_{1})\trans, (\bbeta_{3} - \bbeta_{1})\trans)$.
Similarly define $\bpsi' = ((\bbeta')\trans,\bphi\trans,\sigma)\trans$. Define the contrast matrix $\bM$ such that $\bpsi = \bM\bpsi'$, and corresponding loss $\ell'(\bpsi') = \ell(\bM\bpsi')$.


Under this transformation the penalized objective function $Q_{\lambda}$ given in \eqref{eq:optimtheory} can be equivalently represented by $\ell'(\bpsi') + \sum_{g=1}^3\sum_{j=1}^{d_{g}} p_{\lambda} \left(|\beta'_{gj}|\right).$

By the chain rule, $\nabla\ell'(\bpsi') = \bM\trans\nabla\ell(\bpsi)$ and $\nabla^2\ell'(\bpsi') = \bM\trans\nabla^2\ell(\bpsi)\bM$. 
Define the $j$th column of $\bM$ as $\bM_{\cdot j}$ and the induced matrix $1$-norm as the maximum absolute column sum $\|\bM\|_1 = \max_j \|\bM_{\cdot j}\|_{1}$. For this particular transformation, it can be shown that $\|\bM\|_1 = 3$. Then repeated application of the triangle inequality and H{\"o}lder's inequality illustrates that the maximal values are equivalent up to a known constant:
\begin{align}
\|\nabla\ell'(\bpsi')\|_{\infty} & \leq \|\bM\|_1\|\nabla\ell(\bpsi)\|_{\infty} = 3\|\nabla\ell(\bpsi)\|_{\infty},
\\ \|\nabla^2\ell'(\bpsi')\|_{\max} & \leq \|\bM\|^2_1\|\nabla^2\ell(\bpsi)\|_{\max} = 9\|\nabla^2\ell(\bpsi)\|_{\max}.
\end{align}
As a result, it suffices to show the desired rate for the estimator
\begin{equation} \label{eq:optimtheory2}
\bpsihat = \argmin_{\|\bpsi\|_1 \leq{} R_{1}} \ell(\bpsi) + \sum_{g=1}^3\sum_{j=1}^{d_{g}} p_{\lambda} \left(|\beta_{gj}|\right).
\end{equation} 

Defining $\bnu = \bpsihat - \bpsi^{*}$, then by Assumption~\ref{ass:vec} and the side constraint on \eqref{eq:optimtheory2}, our stationary point satisfies $\|\bnu\|_2 \leq{} \|\bnu\|_1 \leq{} R$. Then by Lemma~\ref{lem:newrsc}, under Assumptions~\ref{ass:bound}, \ref{ass:hess}, and \ref{ass:cond}, there exists a positive constant $c$ establishing the following RSC condition with high probability:
\begin{align}\label{eq:newrsc}
\langle \nabla \ell(\bpsi^* + \bnu) - \nabla \ell(\bpsi^*), \bnu \rangle \geq \rho\|\bnu\|_2^2 - c \sqrt{\frac{\log(dn)}{n}}\|\bnu\|_1^2.
\end{align}

Recall the defined penalty function $P_{\lambda}(\bpsi)=\sum_{g=1}^3\sum_{j=1}^k p_{\lambda} \left( \beta_{gj} \right)$. Setting $p_{\lambda}$ to be the SCAD penalty function \eqref{eq:scad}, then $P_{\lambda}(\bpsi) + \|\bpsi\|_2^2/\{2(\xi-1)\}$ is convex, so it follows that
\begin{align}
\left\langle \nabla P_{\lambda}(\bpsihat), \bpsi^{*}-\bpsihat \right\rangle \leq{} P_{\lambda}(\bpsi^{*}) - P_{\lambda}(\bpsihat) + \frac{1}{2(\xi-1)}\|\bpsihat - \bpsi^{*}\|_{2}^2,
\end{align}
where $\nabla P_{\lambda}(\bpsihat)$ is a subgradient of $P_{\lambda}$ at $\bpsihat$. Combining this with \eqref{eq:newrsc}, and the first order condition
\begin{align}
\langle \nabla \ell(\bpsihat) + \nabla P_{\lambda}(\bpsihat), \bpsi - \bpsihat) \geq 0 \quad \text{for all feasible $\bpsi$},
\end{align}
yields
\begin{align}
\rho\|\bnu\|_{2}^2 - c \sqrt{\frac{\log(dn)}{n}}\|\bnu\|_1^2 \leq{} -\langle \nabla \ell(\bpsi^*),\bnu\rangle + P_{\lambda}(\bpsi^*) -  P_{\lambda}(\bpsihat) + \frac{1}{2(\xi-1)}\|\bnu\|_2^2.
\end{align}
Rearranging and applying H{\"o}lder's inequality gives
\begin{align}
\left\{\rho - \frac{1}{2(\xi-1)}\right\}\|\bnu\|_{2}^2 \leq{} P_{\lambda}(\bpsi^*) -  P_{\lambda}(\bpsihat) + \|\nabla\ell(\bpsi^*)\|_{\infty}\|\bnu\|_1 + c \sqrt{\frac{\log(dn)}{n}}\|\bnu\|_1^2.
\end{align}
Under the constraint that $\|\bnu\|_1 \leq{} R$, this results in
\begin{equation}\label{eq:theoremeq1}
\begin{aligned}
\left\{\rho - \frac{1}{2(\xi-1)}\right\}\|\bnu\|_{2}^2 \leq{} & P_{\lambda}(\bpsi^*) -  P_{\lambda}(\bpsihat) 
\\ & + \left(\|\nabla\ell(\bpsi^*)\|_{\infty} + R c \sqrt{\frac{\log(dn)}{n}} \right)\|\bnu\|_1.
\end{aligned}
\end{equation}
By Lemma~\ref{lem:gradbound}, $\|\nabla\ell(\bpsi^*)\|_{\infty} = O_P\left(\sqrt{\log(d)/n}\right)$. So for suitable choices of $R$ and $\lambda = c'\sqrt{\log(dn)/n}$ with $c'$ sufficiently large, then with high probability
\begin{equation}\label{eq:theoremeq2}
\left\{\rho - \frac{1}{2(\xi-1)}\right\}\|\bnu\|_{2}^2 \leq{} P_{\lambda}(\bpsi^*) -  P_{\lambda}(\bpsihat) + \frac{\lambda}{2}\|\bnu\|_1.
\end{equation}
By the subadditive property of the SCAD penalty, $P_{\lambda}(\bnu) \leq{} P_{\lambda}(\bpsi^{*}) + P_{\lambda}(\bpsihat)$. Moreover, Lemma~4 of \citet{loh2015regularized} shows that for the SCAD penalty,  $\lambda|x| \leq p_{\lambda}(x) + x^2/\{2(\xi-1)\}$ for any $x\in \bbR$. Together, these results give that
\begin{align}
\frac{\lambda}{2}\|\bnu\|_1 \leq{} \frac{1}{2}P_{\lambda}(\bnu) + \frac{1}{4(\xi-1)}\|\bnu\|_{2}^2 \leq{} \frac{P_{\lambda}(\bpsi^{*})}{2} +\frac{P_{\lambda}(\bpsihat)}{2} + \frac{1}{4(\xi-1)}\|\bnu\|_{2}^2.
\end{align}
Combining with \eqref{eq:theoremeq2} and rearranging yields
\begin{align}\label{eq:theoremeq3}
0 \leq{} \left\{\rho - \frac{3}{4(\xi-1)}\right\}\|\bnu\|_{2}^2 \leq{} \frac{3}{2}P_{\lambda}(\bpsi^*) -  \frac{1}{2}P_{\lambda}(\bpsihat),
\end{align}
where the lower bound of 0 follows by assumption that $\rho > 3/\{4(\xi-1)\}$. Assume without loss of generality that the unpenalized $k$ baseline parameters of $\bphi$ and frailty log-variance parameter $\sigma$ are non-zero, so $s \geq (k+1)$. Then define $S$ to be the index set of the $k+1$ unpenalized elements, plus the $s-k-1$ largest elements of $\bbeta$. Then Lemma~5 of \citet{loh2015regularized} states that
\begin{align}
3P_{\lambda}(\bpsi^*) - P_{\lambda}(\bpsihat) \leq{} 3\lambda\|\bnu_{S}\|_1 - \lambda \|\bnu_{S^{c}}\|_1.
\end{align}
Substituting this into \eqref{eq:theoremeq3} gives
\begin{align}
\left\{2\rho - \frac{3}{2(\xi-1)}\right\}\|\bnu\|_{2}^2 \leq{} 3\lambda\|\bnu_{S}\|_1 - \lambda\|\bnu_{S^{c}}\|_1 \leq{} 3\lambda\|\bnu_{S}\|_1 \leq{} 3\lambda\sqrt{s}\|\bnu\|_2,
\end{align}
which yields the final result 
\begin{align}
\|\bnu\|_{2} \leq{} \frac{6\lambda\sqrt{s}}{4\rho - 3(\xi-1)^{-1}}.
\end{align} The statistical rate follows by the choice of $\lambda = O\left(\sqrt{\log(dn)/n}\right)$.
\end{proof}


\newpage
\section{Observed Data Likelihood Expressions for Gamma-Frailty Illness-Death Model}  \label{sec:gradhess}

In this section, we derive expressions for the observed data likelihood, gradient, and Hessian functions for the gamma-frailty illness-death model. As in the main text, for brevity we adopt a semi-Markov specification for $h_3$ throughout, though only simple modifications are required for corresponding formulae under the Markov specification. 

We denote $i$th subject's observed event times $(y_{1i}, y_{2i})$ and corresponding observed outcome indicators $(\delta_{1i},\delta_{2i})$. Next, denote the sum of cumulative cause-specific hazards as
\begin{align}
A_i = H_{01}(y_{1i})e^{\bX_{1i}\trans\bbeta_1} + H_{02}(y_{1i})e^{\bX_{2i}\trans\bbeta_2} + H_{03}(y_{2i}-y_{1i})e^{\bX_{3i}\trans\bbeta_3}.
\end{align}
Now, the negative log-likelihood loss can be succinctly written as
\begin{equation}\label{eq:loglikelihoodi:frailty}
\begin{aligned} 
\ell(\bpsi) = -\frac{1}{n}\sum_{i=1}^n & \Big\{\delta_{1i}\log h_{1}(y_{1i}) + (1-\delta_{1i})\delta_{2i}\log h_{2}(y_{1i})  + \delta_{1i}\delta_{2i}\log h_{3}(y_{2i}-y_{1i})
\\ & +  \delta_{1i}\delta_{2i}\log (1+e^{\sigma})
-(e^{-\sigma}+\delta_{1i}+\delta_{2i})\log(1+e^{\sigma}A_{i})\Big\}.
\end{aligned}
\end{equation}
For detailed casewise derivation of the observed data likelihood, see Appendix B of \citet{lee2015bayesian}.

Finally, as in the main text we reduce repetition by defining unifying notation for the observed outcomes:
\begin{align}\label{eq:ytildeobs}
\ytilde_{gi} ={} & \begin{cases} y_{1i}, & g\in\{1,2\}, \\ y_{2i}-y_{1i}, & g=3, \end{cases}
&\text{and}& & \deltatilde_{gi} ={} & \begin{cases} \delta_{1i}, & g=1, \\ (1-\delta_{1i})\delta_{2i}, & g=2, \\ \delta_{1i}\delta_{2i}, & g=3. \end{cases}
\end{align}

\subsection{General Case} \label{subsec:gengradhess}
\subsubsection{Gradient of Loss}

Considering just the $i$th subject's contribution, for $g=1,2,3$ and $j=1,\dots,k_{g}$ the gradient expressions are
\begin{align}
\frac{\partial \ell_i(\bpsi)}{\partial\sigma} ={} & \frac{\delta_{1i}\delta_{2i}e^{\sigma}}{1+e^{\sigma}} + \frac{\log(1+e^{\sigma}A_i)}{e^{\sigma}} - \frac{1 + e^{\sigma}(\delta_{1i} + \delta_{2i})}{1+e^{\sigma}A_i} A_i , \label{eq:gradsigma}
\\ \frac{\partial \ell_i(\bpsi)}{\partial\bbeta_g\trans} ={} & \left\{\deltatilde_{gi} - \frac{1 + e^{\sigma}(\delta_{1i} + \delta_{2i})}{1+e^{\sigma}A_i}H_{0g}(\ytilde_{gi})e^{\bX_{gi}\trans\bbeta_g}  \right\}\bX_{gi} , \label{eq:gradbetag}
\\ \frac{\partial \ell_i(\bpsi)}{\partial\phi_{gj}} ={} & \deltatilde_{gi}\left\{\frac{\partial}{\partial\phi_{gj}} \log h_{0g}(\ytilde_{gi})\right\} - \frac{\{1 + e^{\sigma}(\delta_{1i} + \delta_{2i})\}e^{\bX_{gi}\trans\bbeta_g}}{1+e^{\sigma}A_i} \left\{\frac{\partial H_{0g}(\ytilde_{gi})}{\partial\phi_{gj}}\right\}. \label{eq:gradphigj}
\end{align}

\subsubsection{Hessian of Loss}  \label{subsubsec:genhess}

Considering just the $i$th subject's contribution, using general notation over $g\in \{1,2,3\}$, $r\in \{1,2,3\}$, $j=1,\dots,k_{g}$, and $l=1,\dots,k_{r}$, the Hessian is expressions are
\begingroup
 \allowdisplaybreaks
\begin{align}
\label{eq:hesssigmasigma} 
\frac{\partial^2 \ell_i(\bpsi)}{\partial\sigma\partial\sigma} ={} & \frac{A_{i} + e^{\sigma}A_i(2A_i - \delta_{1i} - \delta_{2i})}{(1 + e^{\sigma}A_i )^2}
- \frac{\delta_{1i} \delta_{2i}e^{\sigma}}{(1 + e^{\sigma})^2}
 - e^{\sigma}\log(1 + e^{\sigma}A_i),
\\ \label{eq:hessbetabeta} 
\begin{split}
\frac{\partial^2 \ell_i(\bpsi)}{\partial\bbeta_{g}\partial\bbeta_{r}\trans} ={} & \frac{\{1 + e^{\sigma}(\delta_{1i} + \delta_{2i})\} e^{\bX_{gi}\trans\bbeta_{g}}}{1+e^{\sigma}A_i}H_{0g}(\ytilde_{gi})
\\ & \times \left\{\frac{e^{\sigma}H_{0r}(\ytilde_{ri}) e^{\bX_{ri}\trans\bbeta_r}}{1+e^{\sigma}A_i}  - \bbI(g=r)\right\}\bX_{ri}\bX_{gi}\trans,
\end{split}
\\ \label{eq:hesssigmabeta} 
\frac{\partial^{2} \ell_i(\bpsi)}{\partial{\sigma}\partial\bbeta_{r}\trans} ={} & \frac{e^{\sigma}H_{0r}(\ytilde_{ri}) e^{\bX_{ri}\trans\bbeta_r}}{1 + e^{\sigma}A_i}\left\{\frac{1+e^{\sigma}(\delta_{1i}+\delta_{2i})}{1+e^{\sigma}A_i}A_i - (\delta_{1i} + \delta_{2i}) \right\}\bX_{ri},
\\ \label{eq:hesssigmaphi} 
\frac{\partial^{2} \ell_i(\bpsi)}{\partial{\sigma}\partial\phi_{gj}} ={} & \frac{e^{\sigma}e^{\bX_{gi}\trans\bbeta_g}}{1 + e^{\sigma}A_i}\left\{\frac{\partial H_{0g}(\ytilde_{gi})}{\partial\phi_{gj}}\right\}\left\{\frac{1+e^{\sigma}(\delta_{1i}+\delta_{2i})}{1+e^{\sigma}A_i}A_i - (\delta_{1i} + \delta_{2i}) \right\},
\\ \label{eq:hessphibeta} 
\begin{split}
\frac{\partial^2 \ell_i(\bpsi)}{\partial\phi_{gj}\partial\bbeta_{r}\trans} ={} & \frac{(1 + e^{\sigma}(\delta_{1i} + \delta_{2i}))e^{\bX_{gi}\trans\bbeta_{g}}}{1+e^{\sigma}A_i}\left\{\frac{\partial H_{0g}(\ytilde_{gi})}{\partial\phi_{gj}}\right\} 
\\ & \times \left\{\frac{e^{\sigma}H_{0r}(\ytilde_{ri}) e^{\bX_{ri}\trans\bbeta_r}}{1+e^{\sigma}A_i}  - \bbI(g=r)\right\}\bX_{ri},
\end{split}
\\ \label{eq:hessphiphi} 
\begin{split}
\frac{\partial^2 \ell_i(\bpsi)}{\partial\phi_{gj}\partial\phi_{rl}} ={} & \frac{\{1 + e^{\sigma}(\delta_{1i} + \delta_{2i})\}e^{\bX_{gi}\trans\bbeta_{g}}}{1+e^{\sigma}A_i}
\\ & \times \left[\frac{e^{\sigma}e^{\bX_{ri}\trans\bbeta_r}}{1+e^{\sigma}A_i}\left\{\frac{\partial H_{0g}(\ytilde_{gi})}{\partial\phi_{gj}}\frac{\partial H_{0r}(\ytilde_{ri})}{\partial\phi_{rl}}\right\} - \left\{\frac{\partial^2 H_{0g}(\ytilde_{gi})}{\partial\phi_{gj}\partial\phi_{rl}}\right\}\right]
\\ & + \deltatilde_{gi}\left\{\frac{\partial^2}{\partial\phi_{gj}\partial\phi_{rl}}\log h_{0g}(\ytilde_{gi})\right\}.
\end{split}
\end{align}
\endgroup
\subsection{Piecewise Constant Baseline Hazard} \label{subsec:pwgradhess}

Recall that for the $g$th transition baseline hazard, the piecewise constant specification requires a user-defined set of knots $0 = t_{g}^{(1)} <\dots < t_{g}^{(k_g)} < t_{g}^{(k_g+1)}=\infty$ defining the intervals over which the hazard is constant. Using general notation over $g\in \{1,2,3\}$, $r\in \{1,2,3\}$, $j=1,\dots,k_{g}$, and $l=1,\dots,k_{r}$, the cause-specific log-baseline hazard and its first two derivatives are
\begin{align}
\label{eq:pwlogh}
	\log h_{0g}(t) ={} & \sum_{j=1}^{k_g} \phi_{gj} \mathbb{I}(t^{(j)} \leq{} t < t^{(j+1)}), \\
\label{eq:pwloghd}
	\frac{\partial}{\partial\phi_{gj}} \log h_{0g}(t) ={} & \mathbb{I}(t^{(j)} \leq{} t < t^{(j+1)}) ,\\
\label{eq:pwloghd2}
	\frac{\partial^2}{\partial\phi_{gj}\partial\phi_{rl}} \log h_{0g}(t) ={} & 0.
\end{align}
For the $g$th transition, define $B_{gj}(t) = (\min(t,t_g^{(j+1)}) - t_g^{(j)})\mathbb{I}(t \geq t_{g}^{(j)})$ to represent the amount of time spent in the $j$th interval. Then the cumulative cause-specific hazard is
\begin{align}
\label{eq:pwH}
	H_{0g}(t) ={} & \sum_{j=1}^{k_1} e^{\phi_{gj}} B_{gj}(t), \\
\label{eq:pwHd}
	\frac{\partial}{\partial\phi_{gj}} H_{0g}(t) ={} & e^{\phi_{gj}} B_{gj}(t), \\
\label{eq:pwHd2}
	\frac{\partial^2}{\partial\phi_{gj}\partial\phi_{rl}} H_{0g}(t) ={} & e^{\phi_{gj}} B_{gj}(t)\mathbb{I}(g=r,j=l).
\end{align}

\subsection{Weibull Baseline Hazard} \label{subsec:weibgradhess}

Using general notation over $g\in \{1,2,3\}$, $r\in \{1,2,3\}$, $j=1,\dots,k_{g}$, and $l=1,\dots,k_{r}$, the cause-specific log-baseline hazard and its first two derivatives are
\begin{align}
\label{eq:weiblogh}
	\log h_{0g}(t) ={} & \phi_{g1} + \phi_{g2} + (e^{\phi_{g1}}-1)\log t , \\
\label{eq:weibloghd}
	\frac{\partial}{\partial\phi_{gj}} \log h_{0g}(t) ={} & 1+ (e^{\phi_{g1}} \log t)\mathbb{I}(j=1), \\
\label{eq:weibloghd2}
	\frac{\partial^2}{\partial\phi_{gj}\partial\phi_{rl}} \log h_{0g}(t) ={} & (e^{\phi_{g1}} \log t)  \mathbb{I}(g=r,j=l=1).
\end{align}
The cause-specific cumulative hazard and its first two derivatives are then
\begin{align}
\label{eq:weibH}
	H_{0g}(t) ={} & e^{\phi_{g2}}t^{\exp(\phi_{g1})}, \\
\label{eq:weibHd}
	\frac{\partial}{\partial\phi_{gj}} H_{0g}(t) ={} & e^{\phi_{g2}}t^{\exp(\phi_{g1})}(e^{\phi_{g1}}\log t)^{\mathbb{I}(j=1)}, \\
\label{eq:weibHd2}
\begin{split}
\frac{\partial^2}{\partial\phi_{gj}\partial\phi_{rl}} H_{0g}(t) ={} & e^{\phi_{g2}}t^{\exp(\phi_{g1})}(e^{\phi_{g1}}\log t)^{\{\mathbb{I}(j=1)+\mathbb{I}(l=1)\}}\mathbb{I}(g=r)
\\ & + \{e^{\phi_{g2}}e^{\phi_{g1}}t^{\exp(\phi_{g1})}\log t\}\mathbb{I}(g=r,j=l=1).
\end{split}
\end{align}


\newpage
\section{Technical Lemmas} \label{sec:lemmas}

This section presents supporting details underlying the theoretical results. To summarize,
\begin{itemize}
\item Under Assumption~\ref{ass:bound}, Lemma~\ref{lem:pwcond} verifies that the piecewise constant baseline hazard specification satisfies Assumption~\ref{ass:cond}.
\item Under Assumption~\ref{ass:bound}, Lemma~\ref{lem:weibcond} verifies that the Weibull baseline hazard specification satisfies Assumption~\ref{ass:cond}.
\item Under Assumption~\ref{ass:bound} and Assumption~\ref{ass:cond}, Lemma~\ref{lem:gradbound} confirms a probabilistic bound on the largest gradient element.
\item Under Assumption~\ref{ass:bound} and Assumption~\ref{ass:cond}, Lemma~\ref{lem:hessbound} confirms a probabilistic bound on the largest deviation of a Hessian element from its mean.
\item Under Assumption~\ref{ass:hess} and the above Hessian bound, a Restricted Strong Convexity condition follows with high probability.
\end{itemize}

\begin{lemma}\label{lem:pwcond}
Under Assumption~\ref{ass:bound}, the piecewise constant baseline hazard specification satisfies Assumption~\ref{ass:cond}.
\end{lemma}
\begin{proof}
Assumption~\ref{ass:cond}a follows by inspection of the baseline cumulative hazard function \eqref{eq:pwH}, and its derivatives \eqref{eq:pwHd}, and \eqref{eq:pwHd2}, which are piecewise linear and therefore bounded on a closed interval. Assumption~\ref{ass:cond}c follows trivially, as any second derivatives of the log baseline hazard function \eqref{eq:pwloghd2} are the zero function.
Finally, note that any first derivative of the log baseline hazard function \eqref{eq:pwloghd} is just an indicator function, so $\Deltatilde_{gi}\left\{\partial \log h_{0g}(\Ytilde_{gi})/\partial\phi_{gj}\right\}$ is itself a Bernoulli random variable. Therefore, it must have finite variance, and Assumption~\ref{ass:cond}b is established.
\end{proof}

\begin{lemma}\label{lem:weibcond}
Under Assumption~\ref{ass:bound}, the Weibull baseline hazard specification satisfies Assumption~\ref{ass:cond}.
\end{lemma}
\begin{proof}
We start with a proof of Assumption~\ref{ass:cond}a. Note that $\{\bpsi: \|\bpsi-\bpsi^*\|_2 \le R\}$ is a compact subset of $\bbR^{d+k+1}$, and $[0,t]$ is a closed interval. Then because the baseline cumulative hazard function \eqref{eq:weibH} and its derivatives \eqref{eq:weibHd} and \eqref{eq:weibHd2} are continuous in $t$ and $\bpsi$, by the Extreme Value Theorem the functions are bounded over the given space.

To prove Assumption~\ref{ass:cond}b, note that for each so-called `scale' parameter $\phi_{g2}$, the corresponding derivative of the log baseline hazard \eqref{eq:weibloghd} is one, so
\begin{align}
\Var\left\{\Deltatilde_{gi}\frac{\partial}{\partial\phi_{g2}} \log h_{0g}(\Ytilde_{gi})\right\} = \Var\left(\Deltatilde_{gi}\right),
\end{align} which is finite, as $\Deltatilde_{gi}$ is a binary random variable.

However, for each `shape' parameter $\phi_{g1}$, \eqref{eq:weibloghd} is an unbounded function as $t\rightarrow 0$.
Using the law of total variance, then
\begin{align}
\Var\left(\Deltatilde_{gi}\log \Ytilde_{gi}\right) = & \bbE_{\Deltatilde_{gi}}\left\{\Var\left(\Deltatilde_{gi}\log \Ytilde_{gi}\mid \Deltatilde_{gi}\right)\right\} + \Var_{\Deltatilde_{gi}}\left\{\bbE\left(\Deltatilde_{gi}\log \Ytilde_{gi}\mid \Deltatilde_{gi}\right)\right\}
\\ = & \Pr\left(\Deltatilde_{gi}=1\right)\Var\left(\log \Ytilde_{gi}\mid \Deltatilde_{gi}=1\right) 
\\ & + \Pr\left(\Deltatilde_{gi}=1\right)\Pr\left(\Deltatilde_{gi}=0\right)\bbE\left(\log \Ytilde_{gi}\mid \Deltatilde_{gi}=1\right)^2.
\end{align} Thus, to show $\Var\left(\Deltatilde_{gi}\log \Ytilde_{gi}\right)$ is finite it suffices to show finiteness of $E\left(\log \Ytilde_{gi}\mid \Deltatilde_{gi}=1\right)$ and $E\left(\log^2 \Ytilde_{gi}\mid \Deltatilde_{gi}=1\right)$.

Because the three transition submodels are analogous, without loss of generality we will focus on showing finite variance in the case of $g=1$. To start, assume no censoring, and no covariates. Then the marginal distribution of $Y_{1i}$ depends on the correlation of $T_{1i}$ and $T_{2i}$ induced by $\gamma_i$. Under the assumption that $\phi_{11}=\phi_{21}$, \citet{jiang2015simulation} show that
\begin{align}
\Pr(\Delta_{1i}=1) = \left(\frac{e^{\phi_{12}}}{e^{\phi_{12}}+e^{\phi_{22}}}\right),
\end{align} and derive the conditional distribution of $Y_{1i}$ as
\begin{align}
f_{Y_{1i}<\infty}(y_{1i}\mid \Delta_{1i}=1) = & \left(\frac{e^{\phi_{12}}+e^{\phi_{22}}}{e^{\phi_{12}}}\right)\frac{e^{\phi_{12} + \phi_{11}}y_{1i}^{e^{\phi_{11}}-1}}{\{1 + e^{\sigma}(e^{\phi_{12}}+ e^{\phi_{22}})y_{1i}^{e^{\phi_{11}}}\}^{e^{-\sigma}+1} }
\\ = &  \frac{\left(e^{\phi_{12}}+e^{\phi_{22}}\right)e^{\phi_{11}}y_{1i}^{e^{\phi_{11}}-1}}{\{1 + e^{\sigma}(e^{\phi_{12}}+ e^{\phi_{22}})y_{1i}^{e^{\phi_{11}}}\}^{e^{-\sigma}+1}} .
\end{align}
By these formulas, it can be shown that
\begin{align}
\bbE\left(\log Y_{1i}\mid \Delta_{1i}=1\right) = & \int_{0}^{\infty} \log(y_{1i})f_{Y_{1i}<\infty}(y_{1i}\mid \Delta_{1i}=1)dy_{1i} <  \infty, \\
\bbE\left(\log^2 Y_{1i}\mid \Delta_{1i}=1\right) = &\int_{0}^{\infty} \log^2(y_{1i})f_{Y_{1i}<\infty}(y_{1i}\mid \Delta_{1i}=1)dy_{1i} <  \infty,
\end{align} and thus we conclude that $\Var(\Delta_{1i}\log Y_{1i})$ is finite. As long as the random censoring distribution satisfies $\Var(\log C_i) < \infty$, our conclusion remains true if we incorporate censoring. It also remains true in the presence of covariates, by adding $\bX_{i1}\trans\bbeta_{1}$ and $\bX_{i2}\trans\bbeta_{2}$ to $\phi_{12}$ and $\phi_{22}$ respectively. Finally, our conclusion holds if we allow the shape parameters to differ, though the closed form expressions become more complicated.

The proof of Assumption~\ref{ass:cond}c follows directly. When $j=1$ and $l=1$ the form of the second derivative of the log baseline hazard \eqref{eq:weibloghd2} decomposes into $w^{gr}_{jl}(\bpsi)=e^{\phi_{g1}}$ and $z^{gr}_{jl}(t)=\log t$, so $\Var\left\{\Deltatilde_{gi}z^{gr}_{jl}\left(\Ytilde_{gi}\right)\right\}=\Var\left(\Deltatilde_{gi}\log \Ytilde_{gi}\right)$, which is finite by the previous result. When either $j\neq 1$ or $l \neq 1$, \eqref{eq:weibloghd2} is the zero function and the condition follows trivially.

\end{proof}

\begin{lemma}\label{lem:gradbound}
Under Assumption~\ref{ass:bound} and Assumption~\ref{ass:cond}, then there exist positive constants $c_1,c_2$ such that with probability $1-\epsilon$ the gradient of the negative log-likelihood loss satisfies
\begin{equation}
\|\nabla \ell(\bpsi^*)\|_{\infty} \leq{} \sqrt{\frac{\log\{4(d+1)/\epsilon\}}{2c_1n}} + \sqrt{\frac{2kc_2}{n\epsilon}}.
\end{equation}
\end{lemma}
\begin{proof} By Assumption~\ref{ass:bound}, the elements of $\bX_i$ are bounded on $[-\tau_{X},\tau_{X}]$, and $0 < Y_{1i} \leq{} Y_{2i} \leq{} \tau_{Y}$. Then by Assumption~\ref{ass:cond}a, for fixed parameter $\bpsi^{*}$  the gradient functions \eqref{eq:gradsigma} and \eqref{eq:gradbetag} corresponding to $\sigma$ and $\bbeta$ are bounded over this domain. Therefore, choose positive constant $c_1$ such that $c_1 \geq \max\{\|\nabla_{\bbeta} \ell_i(\bpsi^{*})\|_{\infty}, |\nabla_{\sigma} \ell_i(\bpsi^{*})|\}$,
where $\nabla_{\bbeta} \ell_i(\bpsi^*)$ is the $i$th subject's contribution to the gradient component corresponding to $\bbeta$ evaluated at $\bpsi^*$, and $\nabla_{\sigma} \ell_i(\bpsi^*)$ is analogously defined.

Now, using the property that random variables bounded by $[-c_1,c_1]$ are sub-Gaussian with variance proxy $c_1^2$, we may apply Hoeffding's inequality to each gradient component and take a union bound over all $(d+1)$ elements, yielding 
\begin{align}\label{eq:maxgradcompbound1}
\Pr\left[ \max\{\|\nabla_{\bbeta} \ell(\bpsi^{*})\|_{\infty},|\nabla_{\sigma} \ell(\bpsi^{*})|\} > t \right] \leq{} 2(d+1)\exp\left( - \frac{nt^2}{2c_1^2}\right).
\end{align}

However, because the gradient contributions from $\bphi$ may instead be heavy tailed, we introduce a moment inequality approach to bound these random variables. The form of each such gradient element as given in \eqref{eq:gradphigj} has two terms: the first term has finite variance by Assumption~\ref{ass:cond}b, while the second term is bounded by Assumption~\ref{ass:bound} and Assumption~\ref{ass:cond}a. Because bounded random variables have finite variance, then all elements of $\nabla_{\bphi}(\bpsi^{*})$ have finite variance.

Choose positive constant $c_2 \geq \max_{j}\{\Var([\nabla_{\bphi}\ell_{i}(\bpsi^{*})]_j)\}$ to be an upper bound on the variance of the $i$th subject's contributions to all gradient elements in $\bphi$ evaluated at $\bpsi^*$. Then using Chebyshev's inequality, and taking a union bound over all $k$ elements, yields
\begin{align} \label{eq:maxgradcompbound2}
\Pr(\|\nabla_{\bphi} \ell(\bpsi^*)\|_{\infty} > t) \leq{} \frac{kc_2}{nt^2}.
\end{align}

Having addressed each major component of the gradient $\nabla\ell(\bpsi^{*})$, then combining \eqref{eq:maxgradcompbound1} and \eqref{eq:maxgradcompbound2} using a union bound, then the maximum over all of the gradient elements is bounded by
\begin{align}
\Pr(\|\nabla \ell(\bpsi^*)\|_{\infty} > t) & \leq{} \Pr\left( \max\{\|\nabla_{\bbeta} \ell(\bpsi^{*})\|_{\infty},|\nabla_{\sigma} \ell(\bpsi^{*})|\} > t \right) + \Pr(\|\nabla_{\bphi} \ell(\bpsi^*)\|_{\infty} > t)
\\ & \leq{} 2(d + 1)\exp\left( - \frac{nt^2}{2c_1^2}\right) + \frac{kc_2}{nt^2}.
\end{align}
Inverting this result, we have that with probability $1-\epsilon$, 
\begin{align}
\|\nabla \ell(\bpsi^*)\|_{\infty} \leq{} \sqrt{\frac{\log(4(d+1)/\epsilon)}{2c_1n}} + \sqrt{\frac{2kc_2}{n\epsilon}}.
\end{align}
Note that the first term implies a $\sqrt{\log(d)/n}$ rate, while the second implies a $1/\sqrt{n}$ rate, yielding the desired overall result
\begin{align}
\|\nabla \ell(\bpsi^*)\|_{\infty} = O_P\left(\sqrt{\frac{\log d}{n}}\right).
\end{align}
\end{proof}

\begin{lemma}\label{lem:hessbound}
For any scalar $u\in[0,1]$ and any $(d+k+1)$-vector $\bnu$ satisfying $\|\bnu\|_2 \leq{} R$, consider the $i$th subject's Hessian contribution evaluated at $\bpsi^* + u\bnu$. Define the matrix of elementwise deviations from its expectation as
\begin{align*}
G^{\bnu}(u) = \nabla^2 \ell(\bpsi^* + u\bnu) - \Sigma( \bpsi^* + u\bnu).
\end{align*}
Then for a grid of points $u_m = m/n$ for $m = 1,\dots,n$, there exist positive constants $c_3, c_4, c_5, c_6 < \infty$ such that
\begin{equation} \label{eq:hessbound}
\begin{aligned}
 \Pr\big(\max_{1\leq{} m \leq{} n} \|G^{\bnu}(u_m)\|_{\max} > t \big) \leq{} & 2n[(d+k+1)^2 - k^2]\exp\left(-\frac{nt^2}{2c_3^2} \right) 
 \\ & + 2nk^2\exp\left(-\frac{nt^2}{2c_4^2} \right) + \frac{4k^2c_5^2c_6}{nt^2}.
\end{aligned}
\end{equation}
\end{lemma}
\begin{proof}

This result is similar in spirit to Lemma~\ref{lem:gradbound}, in controlling the maximum deviation of a collection of random variables from their means. However, now our approach needs to also account for a grid of parameter values $\bpsi^* + u_m\bnu$, where $u_m = m/n$ for $m = 1,\dots,n$.

Let $\nabla^2_{\bbeta\bbeta}\ell_i(\bpsi)$ be the submatrix of the $i$th subject's Hessian contribution corresponding with the second derivatives of $\bbeta$ evaluated at $\bpsi$. Denote
\begin{align}
G^{\bnu}_{\bbeta\bbeta}(u) = \nabla_{\bbeta\bbeta}^2 \ell(\bpsi^* + u\bnu) - \bbE[\nabla_{\bbeta\bbeta}^2(\bpsi^* + u\bnu)],
\end{align} and define all other submatrices similarly.

\noindent\ul{Step 1: Elements corresponding to partial derivatives of $\bbeta$ and $\sigma$}

By Assumption~\ref{ass:bound}, the elements of $\bX_i$ are bounded on $[-\tau_{X},\tau_{X}]$, and $0 < Y_{1i} \leq{} Y_{2i} \leq{} \tau_{Y}$. Moreover, $\bpsi^* + u\bnu$ lies in an $\ell_2$-ball of radius $R$ around $\bpsi^*$. Then by Assumption~\ref{ass:cond1}a, the Hessian functions \eqref{eq:hesssigmasigma}, \eqref{eq:hessbetabeta}, \eqref{eq:hesssigmabeta}, \eqref{eq:hesssigmaphi}, and \eqref{eq:hessphibeta} corresponding with partial derivatives of $\sigma$ and $\bbeta$ are bounded over this domain. Choose positive upper bound $c_3$ on these $(d+k+1)^2 - k^2$ elements, such that
\begin{align}
c_3 \geq \max_{\|\bpsi-\bpsi^*\|_2\leq{} R}\big\{ & \|\nabla^2_{\bbeta\bbeta}\ell_i(\bpsi)\|_{\max}, \|\nabla^2_{\bbeta\sigma}\ell_i(\bpsi)\|_{\max}, 
\\ & \|\nabla^2_{\bbeta\bphi}\ell_i(\bpsi)\|_{\max}, \|\nabla^2_{\bphi\sigma}\ell_i(\bpsi)\|_{\max}, |\nabla^2_{\sigma\sigma}\ell_i(\bpsi)|  \big\}.
\end{align}
This means that each element is a random variable bounded by $[-c_3,c_3]$, so is sub-Gaussian with variance proxy $c_3^2$. Then using Hoeffding's inequality and taking a union bound over the Hessian components and over the $n$ points of the $u_m$ grid yields
\begin{equation}
\begin{aligned}\label{eq:hessboundbetasigma}
\Pr\big(\max_{1\leq{} m\leq{} n}\big\{ & \|G^{\bnu}_{\bbeta\bbeta}(u_m)\|_{\max}, \|G^{\bnu}_{\bbeta\sigma}(u_m)\|_{\max}, 
\\ & \|G^{\bnu}_{\bbeta\bphi}(u_m)\|_{\max}, \|G^{\bnu}_{\bphi\sigma}(u_m)\|_{\max}, |G^{\bnu}_{\sigma\sigma}(u_m)|  \big\} > t\big) \\
& \leq{} 2n[(d+k+1)^2 - k^2]\exp\left( -\frac{nt^2}{2c_3^2} \right).
\end{aligned}
\end{equation}

\noindent\ul{Step 2: Elements corresponding to second derivatives of $\bphi$}

Importantly, the remaining elements of the Hessian which correspond to the second derivatives of baseline hazard parameters $\bphi$ may be unbounded on this domain. However, under Assumption~\ref{ass:bound},  Assumption~\ref{ass:cond}a, and Assumption~\ref{ass:cond}c, then by \eqref{eq:hessphiphi} the random Hessian element corresponding to the partial derivatives of $\phi_{gj}$ and $\phi_{rl}$ takes the form
\begin{align}\label{eq:hessphiphisimp}
\frac{\partial^2 l_i(\bpsi)}{\partial\phi_{gj}\partial\phi_{rl}} =  B^{gr}_{jl}(\bpsi,\bX_i,\bY_i,\bDelta_i) + \Deltatilde_{gi}\left( w^{gr}_{jl}(\bpsi)z^{gr}_{jl}(\Ytilde) \right),
\end{align}
where each $B^{gr}_{jl}$ is a function bounded on the domain $\|\bpsi - \bpsi^*\|_2 \leq{} R$, $0 \leq{} Y_{1i} \leq{} Y_{2i} \leq{} \tau_{Y}$, $\bDelta_i \in \{0,1\}^2$, and $\|\bX_{i}\|_{\infty} \leq{} \tau_{X}$. So, the goal is to control each of term of \eqref{eq:hessphiphisimp}, and then combine the results.

Towards bounding the first term, choose positive constant $c_4$ that upper bounds all $B^{gr}_{jl}$ over the inputs:
\begin{align}
c_4 \geq \max_{g,r,j,l} \left\{ \max_{\|\bpsi - \bpsi^*\|_2 \leq{} R} \left[ \max_{\substack{ 0\leq{} Y_{1i}\leq{} Y_{2i}\leq{} \tau_{Y}, \\\bDelta_{i} \in \{0,1\}^2}} \left ( \max_{\|\bX_i\|_{\infty}\leq{} \tau_{X}}B^{gr}_{jl}(\bpsi,\bX_i,\bY_i,\bDelta_i) \right) \right] \right\}.
\end{align} 
Therefore, each $B^{gr}_{jl}$ is sub-Gaussian with with variance proxy $c_4^2$, so applying Hoeffding's inequality and a union bound over $m=1,\dots,n$ yields
\begin{equation}
\begin{aligned}\label{eq:hessboundphipart1}
& \Pr\bigg( \max_{1\leq{} m\leq{} n} \frac{1}{n}\sumin \left|B^{gr}_{jl}(\bpsi^{*} + u_{m}\bnu,\bX_i,\bY_i,\bDelta_i)-\bbE[B^{gr}_{jl}(\bpsi^{*} + u_{m}\bnu,\bX_i,\bY_i,\bDelta_i)]\right| > t \bigg)
\\ & \leq{} 2n\exp\left(-\frac{nt^2}{2c_4^2} \right).
\end{aligned}
\end{equation}

To control the second term, note that each $w^{gr}_{jl}$ is continuous, so by the Extreme Value Theorem is bounded on $\|\bpsi - \bpsi^*\|_2 \leq{} R$. Choose positive constant 
\begin{align}
c_5 \geq \max_{g,r,j,l}\left\{ \max_{\|\bpsi - \bpsi^*\|_2 \leq{} R}|w^{gr}_{jl}(\bpsi)|\right\}.
\end{align}
Next, by Assumption~\ref{ass:cond}c, each $\Var\left(\Deltatilde_{gi}z^{gr}_{jl}(\Ytilde_i) \right)$ is finite, so choose
\begin{align}
c_6 \geq \max_{g,r,j,l}\left\{\Var\left(\Deltatilde_{gi}z^{gr}_{jl}(\Ytilde_i)\right)\right\}.
\end{align}
Then by bounding $|w^{gr}_{jl}(\bpsi^* + u_m\bnu)|$ over $m=1,\dots, n$ by $c_5$, and using Chebyshev's inequality on $\Deltatilde_{gi}z^{gr}_{jl}(\Ytilde_i)$, we have
\begin{equation}
\begin{aligned}\label{eq:hessboundphipart2}
&\Pr\left( \max_{1\leq{} m \leq{} n}\frac{1}{n}\sumin \left|\Deltatilde_{gi}w^{gr}_{jl}(\bpsi^* + u_m\bnu)z^{gr}_{jl}(\Ytilde_i) - \bbE\left[\Deltatilde_{gi}w^{gr}_{jl}(\bpsi^* + u_m\bnu)z^{gr}_{jl}(\Ytilde_i)\right]\right| > t \right)
\\ & \leq{} \Pr\bigg( \left\{\max_{1\leq{} m \leq{} n}|w^{gr}_{jl}(\bpsi^* + u_m\bnu)|\right\}\frac{1}{n}\sumin |\Deltatilde_{gi}z^{gr}_{jl}(\Ytilde) - \bbE[\Deltatilde_{gi}z^{gr}_{jl}(\Ytilde_i)]| > t\bigg)
\\ & \leq{} \Pr\bigg( \frac{c_5}{n}\sumin |\Deltatilde_{gi}z^{gr}_{jl}(\Ytilde_i) - \bbE[\Deltatilde_{gi}z^{gr}_{jl}(\Ytilde_i)]| > t\bigg) \leq{} \frac{c_5^2c_6}{nt^2}.
\end{aligned}
\end{equation}

To bring these pieces together, we denote the Hessian component $G^{\bnu}_{\bphi\bphi}(u_m)$ such that
\begin{equation}
\begin{aligned}
\|G^{\bnu}_{\bphi\bphi}(u_m)\|_{\max} = & \max_{g,r,j,l} \frac{1}{n}\sumin \bigg\{ |B^{gr}_{jl}(\bpsi^* + u_m\bnu,\bX_i,\bY_i,\bDelta_i) - \bbE[B^{gr}_{jl}(\bpsi^* + u_m\bnu,\bX_i,\bY_i,\bDelta_i)]| 
\\ & + \left|\Deltatilde_{gi}\left( w^{gr}_{jl}(\bpsi^* + u_m\bnu)z^{gr}_{jl}(\Ytilde_i) \right) - \bbE\left[\Deltatilde_{gi}\left( w^{gr}_{jl}(\bpsi^* + u_m\bnu)z^{gr}_{jl}(\Ytilde_i) \right)\right]\right|\bigg\}.
\end{aligned}
\end{equation}
Then combining \eqref{eq:hessboundphipart1} and \eqref{eq:hessboundphipart2}, and taking a union bound over all $k^2$ elements of this Hessian submatrix yields
\begin{align}\label{eq:hessboundphiphi}
\Pr\left(\max_{1\leq{} m \leq{} n} \|G^{\bnu}_{\bphi\bphi}(u_m)\|_{\max} > t \right) \leq{} 2k^2\left[\frac{4c_5^2c_6}{nt^2} + 2n\exp\left(-\frac{nt^2}{2c_4^2} \right)\right].
\end{align}

\noindent\ul{Step 3: Combine result for overall bound on $\max_{1\leq{} m\leq{} n}\|G^{\bnu}(u_m)\|_{\max}.$}

From the definition of $G^{\bnu}$, we have that
\begin{align}
\|G^{\bnu}(u_m)\|_{\max} = \max\big\{ & \|G^{\bnu}_{\bphi\bphi}(u_m)\|_{\max},\|G^{\bnu}_{\bbeta\bbeta}(u_m)\|_{\max}, \|G^{\bnu}_{\bbeta\sigma}(u_m)\|_{\max}, 
\\ & \|G^{\bnu}_{\bbeta\bphi}(u_m)\|_{\max}, \|G^{\bnu}_{\bphi\sigma}(u_m)\|_{\max}, |G^{\bnu}_{\sigma\sigma}(u_m)|  \big\}.
\end{align}
So combining the results of \eqref{eq:hessboundbetasigma} and \eqref{eq:hessboundphiphi} via a union bound yields the desired result.

Equivalently, we may invert the result such that with probability $1-\epsilon$,
\begin{equation}
\begin{aligned}
\max_{1\leq{} m \leq{} n} \|G^{\bnu}(u_m)\|_{\max} \leq{} & \sqrt{\frac{2c_3^2\log(6n[(d+k+1)^2 - k^2]/\epsilon)}{n}} 
\\ & + \sqrt{\frac{2c_4^2\log(6nk^2/\epsilon)}{n}} + \sqrt{\frac{12k^2c_5^2c_6}{n\epsilon}}.
\end{aligned}
\end{equation}
\end{proof}

\begin{lemma}\label{lem:newrsc}
Under Assumptions~\ref{ass:bound}, \ref{ass:vec} and \ref{ass:hess}, and a model satisfying Assumption~\ref{ass:cond}, then with probability $1-\epsilon$, for some positive constant $c$ the following holds:
\begin{equation}
\langle \nabla \ell(\bpsi^* + \bnu) - \nabla \ell(\bpsi^*), \bnu \rangle \geq \rho\|\bnu\|_2^2 - c \sqrt{\frac{\log(dn)}{n}}\|\bnu\|_1^2 \quad \forall \|\bnu\|_2 \leq{} R.
\end{equation}
\end{lemma}
\begin{proof}

Using the integral definition of the mean-value theorem generalized to vector-valued functions, for some $u \in [0,1]$ we have
\begin{align}
\langle \nabla \ell(\bpsi^* + \bnu) - \nabla \ell(\bpsi^*), \bnu \rangle = \bnu\trans\left(\int_{0}^{1} \nabla^2 \ell(\bpsi^* + u\bnu) du\right)\bnu,
\end{align} where the integral is understood to be elementwise. 
Recall the definition $\Sigma(\bpsi)$ = $\bbE[\nabla^2 \ell(\bpsi)]$, and decompose
\begin{equation}
\begin{aligned}\label{eq:hessdecomp}
\bnu\trans\left(\int_{0}^{1} \nabla^2 \ell(\bpsi^* + u\bnu) du\right)\bnu = & \bnu\trans\left(\int_{0}^{1} \Sigma(\bpsi^* + u\bnu) du\right)\bnu 
\\ & + \bnu\trans\left(\int_{0}^{1} \left[\nabla^2 \ell(\bpsi^* + u\bnu) - \Sigma(\bpsi^* + u\bnu) \right] du\right)\bnu.
\end{aligned}
\end{equation}
To yield our final RSC condition, we bound each of these two righthand terms.

\noindent \ul{Step 1: Bound minimum eigenvalue of the population Hessian in a ball}

Starting with the first term on the righthand side of \eqref{eq:hessdecomp}, and recalling that $\lambda_{\min}(\cdot)$ denotes the matrix minimum eigenvalue,
\begin{align}\label{eq:hessbound1}
\bnu\trans\left(\int_{0}^{1} \Sigma(\bpsi^* + u\bnu) du\right)\bnu = & \int_{0}^{1} \bnu\trans\Sigma(\bpsi^* + u\bnu)\bnu du
\\ \geq & \int_{0}^1\lambda_{\min}\left(\Sigma(\bpsi^* + u\bnu) \right)\|\bnu\|_2^2du.
\end{align}
Now, under Assumption~\ref{ass:hess}, the population Hessian matrix minimum eigenvalue is bounded away from 0 by $\rho$ over the ball $\|\bnu\|_2 \leq{} R$. For any $u \in [0,1]$ then $\bpsi^* + u\bnu$ will be within that ball, yielding the bound
\begin{align}
\bnu\trans\left(\int_{0}^{1} \Sigma(\bpsi^* + u\bnu) du\right)\bnu \geq & \int_{0}^1\rho\|\bnu\|_2^2du = \rho\|\bnu\|_2^2 > 0.
\end{align}

\noindent \ul{Step 2: Bound elementwise maximum deviation between sample and population Hessian}

To simplify notation, define the matrix-valued deviation function $G^{\bnu}(u) = \nabla^2 \ell(\bpsi^* + u\bnu) - \Sigma( \bpsi^* + u\bnu)$ as in Lemma~\ref{lem:hessbound}. For an arbitrary matrix $A$, denote $A_{jl}$ its $(j,l)$th element, and $\|A\|_{\max} = \max_{j,l}|A_{jl}|$. Then repeatedly using the triangle inequality and H{\"o}lder's inequality yields
\begin{align}
\left|\bnu\trans \left(\int_{0}^{1}G^{\bnu}(u)du\right)\bnu\right| \leq{} & \|\bnu\|_{1}^2\int_{0}^1\left\|G^{\bnu}(u)\right\|_{\max}du.
\end{align} 
Our goal is then to find an upper bound for $\left\|G(u)\right\|_{\max}$ over $u$, as 
\begin{align}
\int_{0}^1\left\|G^{\bnu}(u)\right\|_{\max}du \leq \sup_{u\in [0,1]} \left\|G(u)\right\|_{\max}.
\end{align}
Lemma~\ref{lem:gradbound} shows an approach to bounding the maximum of a collection of bounded random variables, but now we must bound the maximum both over the elements and over the interval $u\in [0,1]$. For this we use an $\epsilon$-net argument. Define a grid of points $u_m = m/n$ for $m = 1,\dots,n$, then for any $u$ there exists $m$ such that $|u-u_m|\leq{} \frac{1}{n}$. Then
\begin{equation}
\begin{aligned}\label{eq:enet}
\left\|G^{\bnu}(u)\right\|_{\max} \leq{} & \sup_{u\in [0,1]} \max_{j,l} |[G^{\bnu}(u)]_{jl}| 
\\ \leq{} & \max_{1\leq{} m\leq{} n} \max_{j,l} |[G^{\bnu}(u_m)]_{jl}|  + \sup_{|u-u_m| \leq{} 1/n} \max_{j,l} |[G^{\bnu}(u)]_{jl}-[G^{\bnu}(u_m)]_{jl}| .
\end{aligned}
\end{equation}

By Lemma~\ref{lem:hessbound}, the first term on the RHS of \eqref{eq:enet} is bounded with probability $1-\epsilon$ by

\begin{equation}
\begin{aligned}\label{eq:enetterm1}
\max_{1\leq{} m \leq{} n} \|G^{\bnu}(u_m)\|_{\max} \leq{} & \sqrt{\frac{2c_3^2\log(6n[(d+k+1)^2 - k^2]/\epsilon)}{n}} 
\\ & + \sqrt{\frac{2c_4^2\log(6nk^2/\epsilon)}{n}} + \sqrt{\frac{12k^2c_5^2c_6}{n\epsilon}}.
\end{aligned}
\end{equation}

To bound the final term of \eqref{eq:enet}, each of the $(j,l)$th elements $G^{\bnu}_{jl}(u)$ is continuous in $u$ over the closed interval $[0,1]$ and is therefore locally Lipschitz on the interval. Define $c'<\infty$ such that for all $j,l$, $G^{\bnu}_{jl}$ is $c'$-Lipschitz smooth on the interval $[0,1]$. Then
\begin{align}
\sup_{|u-u_m| \leq{} 1/n} \max_{j,l} |G^{\bnu}_{jl}(u)-G^{\bnu}_{jl}(u_m)| \leq{} \sup_{|u-u_m| \leq{} 1/n} c'|u-u_m| \leq{} \frac{c'}{n}.
\end{align}
Then putting this result and \eqref{eq:enetterm1} into \eqref{eq:enet}, we have that with probability $1-\epsilon$, 
\begin{equation}
\begin{aligned}\label{eq:hessbound2}
\sup_{u\in [0,1]} \max_{j,l} |G^{\bnu}_{jl}(u)|  \leq{} & \sqrt{\frac{2c_3^2\log(6n[(d+k+1)^2 - k^2]/\epsilon)}{n}} 
\\ & + \sqrt{\frac{2c_4^2\log(6nk^2/\epsilon)}{n}} + \sqrt{\frac{12k^2c_5^2c_6}{n\epsilon}} + \frac{c'}{n} .
\end{aligned}
\end{equation}

\noindent \ul{Step 3: Plug bounds into final expression}

By plugging \eqref{eq:hessbound1} and \eqref{eq:hessbound2} into \eqref{eq:hessdecomp}, we conclude that with probability $1-\epsilon$ there exists some positive constant $c$ such that

\begin{equation}
\langle \nabla \ell(\bpsi^* + \bnu) - \nabla \ell(\bpsi^*), \bnu \rangle \geq \rho\|\bnu\|_2^2 - c \sqrt{\frac{\log(dn)}{n}}\|\bnu\|_1^2 \quad \forall \|\bnu\|_2 \leq{} R .
\end{equation}

\end{proof}


\newpage
\section{Additional Simulation Details}\label{sec:addlsims}

In this section, we present additional details on simulation settings, and simulation results. We specifically point out that the non-zero elements of $\bbeta^*$ shown below are always within the first 15 elements of $\bX$, so due to the AR(0.25) serial auto-correlation specification of the covariates $\bX$, the correlation between covariates with non-zero effects ranges from 0.25 for neighboring covariates to $0.25^{-14}$ between $X_{1}$ and $X_{15}$. 

\begin{table}[H]
\caption{Summary of simulation settings run for each sample size $n=300,500,1000$.\label{tab:simsettings}}
\begin{center}
\resizebox{\textwidth}{!}{\begin{tabular}{l|c|c|c|c|c|c|c|c}
\hline
& \multicolumn{8}{c}{Simulation Settings}\\
\cline{2-9}
& 1 & 2 & 3 & 4 & 5 & 6 & 7 & 8 \\
\hline
{\textbf{True Baseline Parameters}} &&&&&&&& \\ 
~~{\textit{Moderate Non-Terminal Event Rate}}  &&&&&&&& \\ 
~~~~$\begin{pmatrix} \bphi^*_{1} \\ \bphi^*_{2} \\ \bphi^*_{3} \end{pmatrix} = \begin{pmatrix} (0.005,0.015,0.050,0.0125)\trans \\
(0.010,0.040,0.075,0.0500)\trans \\
(0.010,0.040,0.075,0.0750)\trans \end{pmatrix}$  & X & X & X& X&&&&\\ \hline
~~{\textit{Low Non-Terminal Event Rate}}  &&&&&&&&\\ 
~~~~$\begin{pmatrix} \bphi^*_{1} \\ \bphi^*_{2} \\ \bphi^*_{3} \end{pmatrix} = \begin{pmatrix} (0.035,0.025,0.020,0.025)\trans \\
(0.005,0.010,0.025,0.018)\trans \\
(0.008,0.015,0.024,0.024)\trans \end{pmatrix}$  &&&&&X&X&X&X \\ \hline
{\textbf{True Regression Parameters}}  &&&&&&&& \\ 
~~{\textit{Shared Support}}  &&&&&&&& \\ 
~~~~$\begin{pmatrix} \bbeta^*_1 \\ \bbeta^*_2 \\ \bbeta^*_3 \end{pmatrix} = \begin{pmatrix} (0.3,-0.4,0.5,0.2,-0.4,0.3,-0.4,0.5,0.2,-0.4,\mathbf{0})\trans \\
(0.8,-1.0,0.6,0.3,-0.5,0.8,-1.0,0.6,0.3,-0.5,\mathbf{0})\trans \\
(0.6,-0.7,0.7,0.4,-0.3,0.6,-0.7,0.7,0.4,-0.3,\mathbf{0})\trans \end{pmatrix}$   &X&X &&&X&X&& \\  \hline
~~{\textit{Partially Shared Support}}  &&&&&&&& \\ 
~~~~$\begin{pmatrix} \bbeta^*_1 \\ \bbeta^*_2 \\ \bbeta^*_3 \end{pmatrix} = \begin{pmatrix} (0.3,-0.4,0.5,0.2,-0.4, 0.3,-0.4,0.5,0.2,-0.4,\mathbf{0})\trans \\ 
(\mathbf{0}_5, 0.6,-0.7,0.5,0.2,-0.4, 0.3,-0.4,0.5,0.2,-0.4,\mathbf{0})\trans \\ 
(\mathbf{0}_5, 0.6,-0.7,0.7,0.4,-0.3,0.6,-0.7,0.7,0.4,-0.3,\mathbf{0})\trans \end{pmatrix}$  &&&X&X&&&X&X \\
\hline
{\textbf{Dimensionality Regime}}  &&&&&&&& \\ 
~~{\textit{Low Dimension} ($d_g = 25$)}  &X&&X&&X&&X& \\ \hline
~~{\textit{High Dimension} ($d_g = 350$)}  &&X&&X&&X&&X \\ \hline
\end{tabular}}
\end{center}
\end{table}

\subsection{Failures of MLE and Forward Selection Optimization}
The MLE and forward selection comparator methods use the \texttt{optim} function built into \texttt{R}, which flags failure to converge. In Table~\ref{tab:failures} we summarize the proportion of MLE and forward selection simulations that failed, by simulation setting. In all subsequent tables, we report results from the subset of successful iterations.

\begin{table}[H]
\caption{\label{tab:failures} Optimization failure proportions for maximum likelihood and forward selection comparator models, by specification and simulation setting. Maximum likelihood estimates only available for low-dimensional setting.}
\begin{center}
\resizebox{0.9\textwidth}{!}{\begin{tabular}{lrrrrrr}
  \hline
 & \multicolumn{2}{c}{$n=300$} & \multicolumn{2}{c}{$n=500$} & \multicolumn{2}{c}{$n=1000$} \\ 
\ul{Weibull} & MLE & Forward & MLE & Forward & MLE & Forward \\ 
  \hline
\multicolumn{7}{l}{\textbf{Moderate Non-Terminal Event Rate}} \\ 
\multicolumn{7}{l}{\textit{~~Shared Support}} \\ 
  ~~~~Low-Dimension & 0.00 & 0.00 & 0.00 & 0.00 & 0.00 & 0.00 \\ 
  ~~~~High-Dimension & --- & 0.67 & --- & 0.09 & --- & 0.00 \\ 
\multicolumn{7}{l}{\textit{~~Partially Non-Overlapping Support}} \\
  ~~~~Low-Dimension & 0.00 & 0.00 & 0.00 & 0.00 & 0.00 & 0.00 \\ 
  ~~~~High-Dimension & --- & 0.70 & --- & 0.10 & --- & 0.00 \\ 
\multicolumn{7}{l}{\textbf{Low Non-Terminal Event Rate}} \\ 
\multicolumn{7}{l}{\textit{~~Shared Support}} \\  
  ~~~~Low-Dimension & 0.01 & 0.10 & 0.00 & 0.06 & 0.00 & 0.03 \\ 
  ~~~~High-Dimension & --- & 0.90 & --- & 0.47 & --- & 0.07 \\ 
\multicolumn{7}{l}{\textit{~~Partially Non-Overlapping Support}} \\
  ~~~~Low-Dimension & 0.00 & 0.00 & 0.00 & 0.00 & 0.00 & 0.00 \\ 
  ~~~~High-Dimension & --- & 0.83 & --- & 0.23 & --- & 0.00 \\ 
\hline
 & \multicolumn{2}{c}{$n=300$} & \multicolumn{2}{c}{$n=500$} & \multicolumn{2}{c}{$n=1000$} \\ 
\ul{Piecewise Constant} & MLE & Forward & MLE & Forward & MLE & Forward \\ 
\hline
\multicolumn{7}{l}{\textbf{Moderate Non-Terminal Event Rate}} \\ 
\multicolumn{7}{l}{\textit{~~Shared Support}} \\ 
  ~~~~Low-Dimension & 0.00 & 0.00 & 0.00 & 0.00 & 0.00 & 0.00 \\ 
  ~~~~High-Dimension & --- & 0.02 & --- & 0.00 & --- & 0.00 \\ 
\multicolumn{7}{l}{\textit{~~Partially Non-Overlapping Support}} \\
  ~~~~Low-Dimension & 0.00 & 0.00 & 0.00 & 0.00 & 0.00 & 0.00 \\ 
  ~~~~High-Dimension & --- & 0.01 & --- & 0.00 & --- & 0.00 \\
\multicolumn{7}{l}{\textbf{Low Non-Terminal Event Rate}} \\ 
\multicolumn{7}{l}{\textit{~~Shared Support}} \\  
  ~~~~Low-Dimension & 0.01 & 0.09 & 0.00 & 0.06 & 0.00 & 0.03 \\ 
  ~~~~High-Dimension & --- & 0.55 & --- & 0.24 & --- & 0.06 \\ 
\multicolumn{7}{l}{\textit{~~Partially Non-Overlapping Support}} \\
  ~~~~Low-Dimension & 0.00 & 0.01 & 0.00 & 0.00 & 0.00 & 0.00 \\ 
  ~~~~High-Dimension & --- & 0.26 & --- & 0.02 & --- & 0.00 \\ 
\hline
\end{tabular}}
\end{center}
\end{table}

\subsection{Estimation Error Results}

\subsubsection{Main Results ($n=500,1000$), Piecewise Constant Model}

\begin{table}[H]
\caption{\label{tab:l2errorpw}Mean $\ell_2$ estimation error of $\bbetahat$, piecewise constant baseline hazard specification. Maximum likelihood estimates only available for low-dimensional setting.}
\begin{center}
\resizebox{\textwidth}{!}{\begin{tabular}{lccccccc}
\hline
$n=500$ & Oracle & MLE & Forward & Lasso & SCAD & Lasso + Fusion & SCAD + Fusion \\ 
\hline
\multicolumn{8}{l}{\textbf{Moderate Non-Terminal Event Rate}} \\ 
\multicolumn{8}{l}{\textit{~~Shared Support}} \\ 
~~~~Low-Dimension & 0.74 & 1.31 & 1.48 & 2.06 & 1.37 & 1.51 & 0.98 \\ 
~~~~High-Dimension & 0.74 & --- & 2.83 & 2.76 & 2.22 & 2.22 & 1.22 \\ 
\multicolumn{8}{l}{\textit{~~Partially Non-Overlapping Support}} \\
~~~~Low-Dimension & 0.71 & 1.27 & 1.32 & 1.89 & 1.25 & 1.61 & 1.16 \\ 
~~~~High-Dimension & 0.73 & --- & 2.71 & 2.49 & 2.18 & 2.28 & 1.43 \\ 
\multicolumn{8}{l}{\textbf{Low Non-Terminal Event Rate}} \\ 
\multicolumn{8}{l}{\textit{~~Shared Support}} \\ 
~~~~Low-Dimension & 0.92 & 1.78 & 1.88 & 2.23 & 1.91 & 1.47 & 1.18 \\ 
~~~~High-Dimension & 0.89 & --- & 3.87 & 2.56 & 2.21 & 2.34 & 1.30 \\ 
\multicolumn{8}{l}{\textit{~~Partially Non-Overlapping Support}} \\
~~~~Low-Dimension & 0.83 & 1.55 & 1.54 & 2.04 & 1.49 & 1.71 & 1.25 \\ 
~~~~High-Dimension & 0.82 & --- & 3.44 & 2.43 & 2.18 & 2.32 & 1.58 \\ 
\hline
$n=1000$ & Oracle & MLE & Forward & Lasso & SCAD & Lasso + Fusion & SCAD + Fusion \\ 
\hline
\multicolumn{8}{l}{\textbf{Moderate Non-Terminal Event Rate}} \\ 
\multicolumn{8}{l}{\textit{~~Shared Support}} \\ 
~~~~Low-Dimension & 0.50 & 0.82 & 0.76 & 1.53 & 0.73 & 1.39 & 0.75 \\ 
~~~~High-Dimension & 0.49 & --- & 1.22 & 2.41 & 1.48 & 1.83 & 0.80 \\ 
\multicolumn{8}{l}{\textit{~~Partially Non-Overlapping Support}} \\
~~~~Low-Dimension & 0.48 & 0.80 & 0.71 & 1.27 & 0.71 & 1.35 & 0.83 \\ 
~~~~High-Dimension & 0.48 & --- & 1.15 & 2.18 & 1.22 & 1.84 & 0.98 \\ 
\multicolumn{8}{l}{\textbf{Low Non-Terminal Event Rate}} \\ 
\multicolumn{8}{l}{\textit{~~Shared Support}} \\ 
~~~~Low-Dimension & 0.58 & 0.97 & 1.24 & 2.03 & 1.12 & 1.34 & 0.82 \\ 
~~~~High-Dimension & 0.59 & --- & 1.85 & 2.39 & 2.09 & 1.78 & 0.89 \\ 
\multicolumn{8}{l}{\textit{~~Partially Non-Overlapping Support}} \\
~~~~Low-Dimension & 0.55 & 0.92 & 0.85 & 1.48 & 0.85 & 1.42 & 0.94 \\ 
~~~~High-Dimension & 0.55 & --- & 1.42 & 2.22 & 1.76 & 1.93 & 1.04 \\ 
\hline
\end{tabular}}
\end{center}
\end{table}

\subsubsection{Small Sample Result ($n=300$), Weibull and Piecewise Constant  Models}

\begin{table}[H]
\caption{\label{tab:l2errorsmall}Mean $\ell_2$ estimation error of $\bbetahat$, Weibull and piecewise constant baseline hazard specifications, sample size $n=300$. Maximum likelihood estimates only available for low-dimensional setting.}
\begin{center}
\resizebox{\textwidth}{!}{\begin{tabular}{lccccccc}
\hline
\multicolumn{8}{l}{\ul{Weibull}} \\ 
$n=300$ & Oracle & MLE & Forward & Lasso & SCAD & Lasso + Fusion & SCAD + Fusion \\ 
\hline
\multicolumn{8}{l}{\textbf{Moderate Non-Terminal Event Rate}} \\ 
\multicolumn{8}{l}{\textit{~~Shared Support}} \\ 
~~~~Low-Dimension & 1.05 & 2.24 & 1.94 & 2.53 & 1.90 & 1.64 & 1.23 \\ 
~~~~High-Dimension & 1.05 & --- & 6.02 & 2.96 & 22.04 & 2.88 & 22.04 \\ 
\multicolumn{8}{l}{\textit{~~Partially Non-Overlapping Support}} \\
~~~~Low-Dimension & 1.02 & 2.12 & 1.86 & 2.37 & 1.79 & 2.02 & 1.47 \\ 
~~~~High-Dimension & 1.00 & --- & 6.98 & 2.53 & 20.47 & 2.52 & 20.47 \\ 
\multicolumn{8}{l}{\textbf{Low Non-Terminal Event Rate}} \\ 
\multicolumn{8}{l}{\textit{~~Shared Support}} \\ 
~~~~Low-Dimension & 1.33 & 5.49 & 2.11 & 2.37 & 2.15 & 1.78 & 1.53 \\ 
~~~~High-Dimension & 1.33 & --- & 5.67 & 2.75 & 20.90 & 2.71 & 20.87 \\ 
\multicolumn{8}{l}{\textit{~~Partially Non-Overlapping Support}} \\
~~~~Low-Dimension & 1.15 & 3.18 & 1.99 & 2.28 & 2.04 & 2.04 & 1.56 \\ 
~~~~High-Dimension & 1.14 & --- & 5.96 & 2.50 & 20.93 & 2.52 & 20.93 \\ 
\hline
\multicolumn{8}{l}{\ul{Piecewise Constant}} \\ 
$n=300$ & Oracle & MLE & Forward & Lasso & SCAD & Lasso + Fusion & SCAD + Fusion \\ 
\hline
\multicolumn{8}{l}{\textbf{Moderate Non-Terminal Event Rate}} \\ 
\multicolumn{8}{l}{\textit{~~Shared Support}} \\ 
~~~~Low-Dimension & 1.02 & 2.09 & 1.94 & 2.52 & 1.88 & 1.67 & 1.25 \\ 
~~~~High-Dimension & 1.02 & --- & 6.72 & 2.97 & 25.92 & 2.86 & 24.97 \\
\multicolumn{8}{l}{\textit{~~Partially Non-Overlapping Support}} \\
~~~~Low-Dimension & 0.99 & 2.01 & 1.84 & 2.37 & 1.75 & 2.03 & 1.47 \\ 
~~~~High-Dimension & 0.98 & --- & 7.11 & 2.53 & 26.37 & 2.52 & 26.05 \\
\multicolumn{8}{l}{\textbf{Low Non-Terminal Event Rate}} \\ 
\multicolumn{8}{l}{\textit{~~Shared Support}} \\ 
~~~~Low-Dimension & 1.32 & 4.65 & 2.11 & 2.35 & 2.14 & 1.76 & 1.52 \\ 
~~~~High-Dimension & 1.31 & --- & 7.42 & 2.77 & 21.46 & 2.71 & 19.28 \\ 
\multicolumn{8}{l}{\textit{~~Partially Non-Overlapping Support}} \\
~~~~Low-Dimension & 1.16 & 2.88 & 2.00 & 2.28 & 2.01 & 2.04 & 1.54 \\ 
~~~~High-Dimension & 1.16 & --- & 7.78 & 2.51 & 24.34 & 2.52 & 23.57 \\ 
\hline
\end{tabular}}
\end{center}
\end{table}

\subsection{Sign Inconsistency Results}

\subsubsection{Main Results ($n=500,1000$), Piecewise Constant Model}

\begin{table}[H]
\caption{\label{tab:signincpw}Mean count of sign-inconsistent $\bbetahat$ estimates, piecewise constant baseline hazard specification.}
\begin{center}
\resizebox{\textwidth}{!}{\begin{tabular}{lcccccc}
\hline
$n=500$ & Oracle & Forward & Lasso & SCAD & Lasso + Fusion & SCAD + Fusion \\ 
\hline
\multicolumn{7}{l}{\textbf{Moderate Non-Terminal Event Rate}} \\ 
\multicolumn{7}{l}{\textit{~~Shared Support}} \\ 
~~~~Low-Dimension & 0.12 & 11.52 & 15.21 & 10.46 & 11.55 & 3.56 \\ 
~~~~High-Dimension & 0.13 & 35.43 & 26.20 & 30.09 & 21.51 & 20.79 \\ 
\multicolumn{7}{l}{\textit{~~Partially Non-Overlapping Support}} \\
~~~~Low-Dimension & 0.13 & 10.70 & 15.47 & 10.10 & 15.21 & 8.19 \\ 
~~~~High-Dimension & 0.15 & 33.50 & 27.06 & 35.51 & 22.86 & 26.78 \\
\multicolumn{7}{l}{\textbf{Low Non-Terminal Event Rate}} \\ 
\multicolumn{7}{l}{\textit{~~Shared Support}} \\ 
~~~~Low-Dimension & 0.28 & 16.02 & 19.45 & 16.98 & 13.05 & 5.28 \\ 
~~~~High-Dimension & 0.26 & 40.10 & 26.66 & 25.24 & 23.64 & 14.43 \\ 
\multicolumn{7}{l}{\textit{~~Partially Non-Overlapping Support}} \\
~~~~Low-Dimension & 0.23 & 12.42 & 17.34 & 12.17 & 16.76 & 9.27 \\ 
~~~~High-Dimension & 0.20 & 37.24 & 26.66 & 29.50 & 24.50 & 23.75 \\ 
\hline
$n=1000$ & Oracle & Forward & Lasso & SCAD & Lasso + Fusion & SCAD + Fusion \\ 
\hline
\multicolumn{7}{l}{\textbf{Moderate Non-Terminal Event Rate}} \\ 
\multicolumn{7}{l}{\textit{~~Shared Support}} \\ 
~~~~Low-Dimension & 0.01 & 4.05 & 13.58 & 3.96 & 8.33 & 1.66 \\ 
~~~~High-Dimension & 0.00 & 15.59 & 20.81 & 18.45 & 19.05 & 6.38 \\ 
\multicolumn{7}{l}{\textit{~~Partially Non-Overlapping Support}} \\
~~~~Low-Dimension & 0.02 & 4.10 & 13.40 & 4.27 & 12.14 & 3.94 \\ 
~~~~High-Dimension & 0.03 & 15.24 & 19.87 & 16.66 & 21.40 & 8.24 \\ 
\multicolumn{7}{l}{\textbf{Low Non-Terminal Event Rate}} \\ 
\multicolumn{7}{l}{\textit{~~Shared Support}} \\ 
~~~~Low-Dimension & 0.05 & 8.81 & 16.64 & 8.54 & 8.76 & 1.65 \\ 
~~~~High-Dimension & 0.06 & 22.83 & 24.67 & 22.04 & 19.69 & 6.01 \\ 
\multicolumn{7}{l}{\textit{~~Partially Non-Overlapping Support}} \\
~~~~Low-Dimension & 0.03 & 5.14 & 14.84 & 5.63 & 13.53 & 4.79 \\ 
~~~~High-Dimension & 0.02 & 17.00 & 22.64 & 22.43 & 22.94 & 8.49 \\ 
\hline
\end{tabular}}
\end{center}
\end{table}

\subsubsection{Small Sample Result ($n=300$),  Weibull and Piecewise Constant Models}

\begin{table}[H]
\caption{\label{tab:signincsmall}Mean count of sign-inconsistent $\bbetahat$ estimates, Weibull and piecewise constant baseline hazard specifications, sample size $n=300$.}
\begin{center}
\resizebox{\textwidth}{!}{\begin{tabular}{lcccccc}
\hline
\multicolumn{7}{l}{\ul{Weibull}} \\ 
$n=300$ & Oracle & Forward & Lasso & SCAD & Lasso + Fusion & SCAD + Fusion \\ 
\hline
\multicolumn{7}{l}{\textbf{Moderate Non-Terminal Event Rate}} \\ 
\multicolumn{7}{l}{\textit{~~Shared Support}} \\ 
~~~~Low-Dimension & 0.48 & 17.11 & 21.26 & 16.21 & 14.17 & 6.93 \\ 
~~~~High-Dimension & 0.48 & 50.56 & 28.73 & 209.17 & 26.88 & 209.17 \\ 
\multicolumn{7}{l}{\textit{~~Partially Non-Overlapping Support}} \\
~~~~Low-Dimension & 0.50 & 16.96 & 23.95 & 15.80 & 19.35 & 12.75 \\ 
~~~~High-Dimension & 0.55 & 50.21 & 28.72 & 208.05 & 28.49 & 208.05 \\ 
\multicolumn{7}{l}{\textbf{Low Non-Terminal Event Rate}} \\ 
\multicolumn{7}{l}{\textit{~~Shared Support}} \\ 
~~~~Low-Dimension & 0.91 & 19.07 & 21.01 & 19.84 & 16.38 & 10.24 \\ 
~~~~High-Dimension & 0.93 & 44.83 & 27.66 & 217.74 & 26.89 & 217.48 \\ 
\multicolumn{7}{l}{\textit{~~Partially Non-Overlapping Support}} \\
~~~~Low-Dimension & 0.66 & 18.31 & 22.05 & 17.72 & 19.45 & 13.61 \\ 
~~~~High-Dimension & 0.68 & 48.14 & 28.23 & 231.49 & 28.63 & 231.49 \\ 
\hline
\multicolumn{7}{l}{\ul{Piecewise Constant}} \\ 
$n=300$ & Oracle & Forward & Lasso & SCAD & Lasso + Fusion & SCAD + Fusion \\ 
\hline
\multicolumn{7}{l}{\textbf{Moderate Non-Terminal Event Rate}} \\ 
\multicolumn{7}{l}{\textit{~~Shared Support}} \\ 
~~~~Low-Dimension & 0.47 & 17.12 & 21.15 & 16.06 & 13.84 & 7.14 \\ 
~~~~High-Dimension & 0.49 & 53.07 & 28.79 & 177.87 & 27.06 & 173.04 \\ 
\multicolumn{7}{l}{\textit{~~Partially Non-Overlapping Support}} \\
~~~~Low-Dimension & 0.49 & 16.85 & 23.90 & 15.51 & 19.37 & 12.63 \\ 
~~~~High-Dimension & 0.54 & 53.45 & 28.65 & 184.61 & 28.72 & 183.47 \\ 
\multicolumn{7}{l}{\textbf{Low Non-Terminal Event Rate}} \\ 
\multicolumn{7}{l}{\textit{~~Shared Support}} \\ 
~~~~Low-Dimension & 0.88 & 19.05 & 21.14 & 19.63 & 16.75 & 10.27 \\ 
~~~~High-Dimension & 0.90 & 49.42 & 28.18 & 174.62 & 27.17 & 158.78 \\ 
\multicolumn{7}{l}{\textit{~~Partially Non-Overlapping Support}} \\
~~~~Low-Dimension & 0.67 & 18.17 & 22.16 & 17.85 & 20.02 & 13.35 \\ 
~~~~High-Dimension & 0.67 & 52.44 & 28.60 & 191.24 & 28.74 & 186.23 \\ 
\hline
\end{tabular}}
\end{center}
\end{table}

\subsection{False Inclusion Results}

\subsubsection{Main Results ($n=500,1000$),  Weibull and Piecewise Constant Models}

\begin{table}[H]
\caption{Mean count of false inclusions, Weibull model specification.\label{tab:falseinclwb}}
\begin{center}
\resizebox{0.95\textwidth}{!}{\begin{tabular}{lcccccc}

\hline
$n=500$ & Forward & Lasso & SCAD & Lasso + Fusion & SCAD + Fusion \\ 
\hline
\multicolumn{6}{l}{\textbf{Moderate Non-Terminal Event Rate}} \\ 
\multicolumn{6}{l}{\textit{~~Shared Support}} \\ 
~~~~Low-Dimension & 0.77 & 6.18 & 1.70 & 11.66 & 0.79 \\ 
~~~~High-Dimension & 23.68 & 3.52 & 18.93 & 18.23 & 16.13 \\ 
\multicolumn{6}{l}{\textit{~~Partially Non-Overlapping Support}} \\
~~~~Low-Dimension & 0.82 & 5.97 & 1.97 & 12.24 & 1.53 \\ 
~~~~High-Dimension & 23.59 & 0.86 & 23.95 & 6.18 & 16.89 \\ 
\multicolumn{6}{l}{\textbf{Low Non-Terminal Event Rate}} \\ 
\multicolumn{6}{l}{\textit{~~Shared Support}} \\ 
~~~~Low-Dimension & 0.94 & 3.97 & 1.82 & 11.41 & 0.89 \\ 
~~~~High-Dimension & 23.15 & 5.51 & 4.28 & 11.97 & 10.61 \\ 
\multicolumn{6}{l}{\textit{~~Partially Non-Overlapping Support}} \\
~~~~Low-Dimension & 0.82 & 4.13 & 1.99 & 11.06 & 1.67 \\ 
~~~~High-Dimension & 22.49 & 2.10 & 9.97 & 4.64 & 12.62 \\ 
\hline
$n=1000$ & Forward & Lasso & SCAD & Lasso + Fusion & SCAD + Fusion \\ 
\hline
\multicolumn{6}{l}{\textbf{Moderate Non-Terminal Event Rate}} \\ 
\multicolumn{6}{l}{\textit{~~Shared Support}} \\ 
~~~~Low-Dimension & 0.54 & 11.37 & 1.13 & 8.45 & 0.60 \\ 
~~~~High-Dimension & 11.79 & 7.32 & 10.41 & 19.43 & 6.47 \\ 
\multicolumn{6}{l}{\textit{~~Partially Non-Overlapping Support}} \\
~~~~Low-Dimension & 0.56 & 12.39 & 1.05 & 11.99 & 0.82 \\ 
~~~~High-Dimension & 11.57 & 6.93 & 12.16 & 17.41 & 5.02 \\ 
\multicolumn{6}{l}{\textbf{Low Non-Terminal Event Rate}} \\ 
\multicolumn{6}{l}{\textit{~~Shared Support}} \\ 
~~~~Low-Dimension & 0.55 & 5.25 & 2.08 & 8.28 & 0.41 \\ 
~~~~High-Dimension & 13.06 & 6.48 & 3.67 & 18.06 & 5.57 \\ 
\multicolumn{6}{l}{\textit{~~Partially Non-Overlapping Support}} \\
~~~~Low-Dimension & 0.43 & 10.79 & 1.67 & 11.83 & 1.06 \\ 
~~~~High-Dimension & 11.54 & 5.65 & 8.91 & 15.09 & 4.05 \\ 
\hline
\end{tabular}}
\end{center}
\end{table}

\begin{table}[H]
\caption{Mean count of false inclusions, piecewise constant model specification.\label{tab:falseinclpw}}
\begin{center}
\resizebox{\textwidth}{!}{\begin{tabular}{lcccccc}
\hline
$n=500$ & Forward & Lasso & SCAD & Lasso + Fusion & SCAD + Fusion \\ 
\hline
\multicolumn{6}{l}{\textbf{Moderate Non-Terminal Event Rate}} \\ 
\multicolumn{6}{l}{\textit{~~Shared Support}} \\ 
~~~~Low-Dimension & 0.74 & 6.28 & 1.61 & 11.30 & 0.81 \\ 
~~~~High-Dimension & 23.12 & 3.60 & 13.24 & 17.62 & 18.15 \\ 
\multicolumn{6}{l}{\textit{~~Partially Non-Overlapping Support}} \\
~~~~Low-Dimension & 0.86 & 5.87 & 1.81 & 12.48 & 1.45 \\ 
~~~~High-Dimension & 22.43 & 0.89 & 20.48 & 6.36 & 19.38 \\ 
\multicolumn{6}{l}{\textbf{Low Non-Terminal Event Rate}} \\ 
\multicolumn{6}{l}{\textit{~~Shared Support}} \\ 
~~~~Low-Dimension & 0.94 & 4.13 & 1.72 & 11.44 & 0.85 \\ 
~~~~High-Dimension & 23.96 & 6.06 & 4.68 & 11.73 & 9.92 \\ 
\multicolumn{6}{l}{\textit{~~Partially Non-Overlapping Support}} \\
~~~~Low-Dimension & 0.84 & 4.46 & 2.12 & 11.48 & 1.59 \\ 
~~~~High-Dimension & 23.69 & 2.23 & 10.14 & 5.18 & 13.05 \\ 
\hline
$n=1000$ & Forward & Lasso & SCAD & Lasso + Fusion & SCAD + Fusion \\ 
\hline
\multicolumn{6}{l}{\textbf{Moderate Non-Terminal Event Rate}} \\ 
\multicolumn{6}{l}{\textit{~~Shared Support}} \\ 
~~~~Low-Dimension & 0.53 & 11.44 & 1.00 & 8.29 & 0.61 \\ 
~~~~High-Dimension & 11.54 & 7.21 & 9.27 & 18.70 & 5.53 \\ 
\multicolumn{6}{l}{\textit{~~Partially Non-Overlapping Support}} \\
~~~~Low-Dimension & 0.56 & 11.90 & 0.94 & 11.56 & 0.84 \\ 
~~~~High-Dimension & 11.28 & 7.01 & 8.96 & 17.65 & 4.75 \\ 
\multicolumn{6}{l}{\textbf{Low Non-Terminal Event Rate}} \\ 
\multicolumn{6}{l}{\textit{~~Shared Support}} \\ 
~~~~Low-Dimension & 0.53 & 5.57 & 2.05 & 8.59 & 0.51 \\ 
~~~~High-Dimension & 13.18 & 6.78 & 3.67 & 18.73 & 4.64 \\ 
\multicolumn{6}{l}{\textit{~~Partially Non-Overlapping Support}} \\
~~~~Low-Dimension & 0.47 & 11.38 & 1.55 & 12.21 & 0.93 \\ 
~~~~High-Dimension & 11.77 & 6.47 & 9.83 & 16.45 & 3.89 \\ 
\hline
\end{tabular}}
\end{center}
\end{table}

\subsubsection{Small Sample Result ($n=300$),  Weibull and Piecewise Constant Models}

\begin{table}[H]
\caption{Mean count of false inclusions, Weibull and piecewise constant baseline hazard specifications, sample size $n=300$.\label{tab:falseinclsmall}}
\begin{center}
\resizebox{0.95\textwidth}{!}{\begin{tabular}{lcccccc}
\hline
\multicolumn{6}{l}{\ul{Weibull}} \\ 
$n=300$ & Forward & Lasso & SCAD & Lasso + Fusion & SCAD + Fusion \\ 
\hline
\multicolumn{6}{l}{\textbf{Moderate Non-Terminal Event Rate}} \\ 
\multicolumn{6}{l}{\textit{~~Shared Support}} \\ 
~~~~Low-Dimension & 1.22 & 2.76 & 1.94 & 12.70 & 1.27 \\ 
~~~~High-Dimension & 32.84 & 1.14 & 194.58 & 2.67 & 194.58 \\ 
\multicolumn{6}{l}{\textit{~~Partially Non-Overlapping Support}} \\
~~~~Low-Dimension & 1.08 & 1.20 & 2.29 & 6.69 & 1.75 \\ 
~~~~High-Dimension & 32.56 & 1.03 & 193.72 & 0.98 & 193.72 \\ 
\multicolumn{6}{l}{\textbf{Low Non-Terminal Event Rate}} \\ 
\multicolumn{6}{l}{\textit{~~Shared Support}} \\ 
~~~~Low-Dimension & 1.14 & 3.38 & 1.87 & 10.47 & 1.32 \\ 
~~~~High-Dimension & 25.76 & 4.15 & 200.79 & 5.56 & 200.57 \\ 
\multicolumn{6}{l}{\textit{~~Partially Non-Overlapping Support}} \\
~~~~Low-Dimension & 1.18 & 2.08 & 2.34 & 6.15 & 2.14 \\ 
~~~~High-Dimension & 29.12 & 1.71 & 216.09 & 1.08 & 216.09 \\ 
\hline
\multicolumn{6}{l}{\ul{Piecewise Constant}} \\ 
$n=300$ & Forward & Lasso & SCAD & Lasso + Fusion & SCAD + Fusion \\ 
\hline
\multicolumn{6}{l}{\textbf{Moderate Non-Terminal Event Rate}} \\ 
\multicolumn{6}{l}{\textit{~~Shared Support}} \\ 
~~~~Low-Dimension & 1.20 & 2.89 & 2.20 & 12.13 & 1.21 \\ 
~~~~High-Dimension & 35.64 & 1.09 & 162.64 & 3.42 & 158.58 \\ 
\multicolumn{6}{l}{\textit{~~Partially Non-Overlapping Support}} \\
~~~~Low-Dimension & 1.12 & 1.31 & 2.08 & 6.81 & 1.73 \\ 
~~~~High-Dimension & 35.82 & 1.17 & 169.88 & 1.10 & 169.27 \\ 
\multicolumn{6}{l}{\textbf{Low Non-Terminal Event Rate}} \\ 
\multicolumn{6}{l}{\textit{~~Shared Support}} \\ 
~~~~Low-Dimension & 1.18 & 3.66 & 1.76 & 10.72 & 1.27 \\ 
~~~~High-Dimension & 30.79 & 4.17 & 156.88 & 5.90 & 142.89 \\ 
\multicolumn{6}{l}{\textit{~~Partially Non-Overlapping Support}} \\
~~~~Low-Dimension & 1.22 & 2.05 & 2.45 & 6.42 & 2.17 \\ 
~~~~High-Dimension & 33.62 & 1.56 & 174.97 & 1.04 & 170.26 \\ 
\hline
\end{tabular}}
\end{center}
\end{table}

\subsection{False Exclusion Results}

\subsubsection{Main Results ($n=500,1000$), Weibull and Piecewise Constant Models}

\begin{table}[H]
\caption{Mean count of false exclusions, Weibull model specification.\label{tab:falseexclwb}}
\begin{center}
\resizebox{0.95\textwidth}{!}{\begin{tabular}{lcccccc}
\hline
$n=500$ & Forward & Lasso & SCAD & Lasso + Fusion & SCAD + Fusion \\ 
\hline
\multicolumn{6}{l}{\textbf{Moderate Non-Terminal Event Rate}} \\ 
\multicolumn{6}{l}{\textit{~~Shared Support}} \\ 
~~~~Low-Dimension & 10.74 & 8.85 & 8.58 & 0.24 & 2.66 \\ 
~~~~High-Dimension & 12.12 & 22.87 & 16.29 & 3.66 & 2.76 \\ 
\multicolumn{6}{l}{\textit{~~Partially Non-Overlapping Support}} \\
~~~~Low-Dimension & 9.90 & 9.40 & 8.21 & 2.80 & 6.37 \\ 
~~~~High-Dimension & 10.96 & 26.35 & 14.57 & 17.18 & 7.46 \\ 
\multicolumn{6}{l}{\textbf{Low Non-Terminal Event Rate}} \\ 
\multicolumn{6}{l}{\textit{~~Shared Support}} \\ 
~~~~Low-Dimension & 15.29 & 15.34 & 15.32 & 1.35 & 4.56 \\ 
~~~~High-Dimension & 16.33 & 20.57 & 20.65 & 11.43 & 3.81 \\ 
\multicolumn{6}{l}{\textit{~~Partially Non-Overlapping Support}} \\
~~~~Low-Dimension & 11.88 & 12.97 & 10.43 & 5.24 & 7.80 \\ 
~~~~High-Dimension & 13.71 & 23.98 & 19.54 & 19.89 & 10.26 \\ 
\hline
$n=1000$ & Forward & Lasso & SCAD & Lasso + Fusion & SCAD + Fusion \\ 
\hline
\multicolumn{6}{l}{\textbf{Moderate Non-Terminal Event Rate}} \\ 
\multicolumn{6}{l}{\textit{~~Shared Support}} \\ 
~~~~Low-Dimension & 3.49 & 2.15 & 2.86 & 0.04 & 0.91 \\ 
~~~~High-Dimension & 4.04 & 13.60 & 9.51 & 0.39 & 0.74 \\ 
\multicolumn{6}{l}{\textit{~~Partially Non-Overlapping Support}} \\
~~~~Low-Dimension & 3.56 & 1.42 & 3.22 & 0.59 & 3.18 \\ 
~~~~High-Dimension & 4.01 & 12.93 & 7.14 & 3.87 & 3.42 \\ 
\multicolumn{6}{l}{\textbf{Low Non-Terminal Event Rate}} \\ 
\multicolumn{6}{l}{\textit{~~Shared Support}} \\ 
~~~~Low-Dimension & 9.16 & 11.33 & 6.94 & 0.15 & 1.45 \\ 
~~~~High-Dimension & 10.67 & 17.81 & 18.49 & 1.02 & 1.37 \\ 
\multicolumn{6}{l}{\textit{~~Partially Non-Overlapping Support}} \\
~~~~Low-Dimension & 4.86 & 3.75 & 4.08 & 1.43 & 3.91 \\ 
~~~~High-Dimension & 5.42 & 16.32 & 13.27 & 6.89 & 4.66 \\ 
\hline
\end{tabular}}
\end{center}
\end{table}

\begin{table}[H]
\caption{Mean count of false exclusions, piecewise constant model specification.\label{tab:falseexclpw}}
\begin{center}
\resizebox{\textwidth}{!}{\begin{tabular}{lcccccc}
\hline
$n=500$ & Forward & Lasso & SCAD & Lasso + Fusion & SCAD + Fusion \\ 
\hline
\multicolumn{6}{l}{\textbf{Moderate Non-Terminal Event Rate}} \\ 
\multicolumn{6}{l}{\textit{~~Shared Support}} \\ 
~~~~Low-Dimension & 10.77 & 8.89 & 8.85 & 0.25 & 2.75 \\ 
~~~~High-Dimension & 12.30 & 22.60 & 16.85 & 3.89 & 2.64 \\ 
\multicolumn{6}{l}{\textit{~~Partially Non-Overlapping Support}} \\
~~~~Low-Dimension & 9.84 & 9.59 & 8.29 & 2.72 & 6.74 \\ 
~~~~High-Dimension & 11.07 & 26.18 & 15.02 & 16.50 & 7.40 \\ 
\multicolumn{6}{l}{\textbf{Low Non-Terminal Event Rate}} \\ 
\multicolumn{6}{l}{\textit{~~Shared Support}} \\ 
~~~~Low-Dimension & 15.08 & 15.30 & 15.25 & 1.61 & 4.42 \\ 
~~~~High-Dimension & 16.11 & 20.60 & 20.56 & 11.91 & 4.52 \\ 
\multicolumn{6}{l}{\textit{~~Partially Non-Overlapping Support}} \\
~~~~Low-Dimension & 11.57 & 12.87 & 10.05 & 5.24 & 7.69 \\ 
~~~~High-Dimension & 13.54 & 24.43 & 19.36 & 19.31 & 10.70 \\ 
\hline
$n=1000$ & Forward & Lasso & SCAD & Lasso + Fusion & SCAD + Fusion \\ 
\hline
\multicolumn{6}{l}{\textbf{Moderate Non-Terminal Event Rate}} \\ 
\multicolumn{6}{l}{\textit{~~Shared Support}} \\ 
~~~~Low-Dimension & 3.52 & 2.14 & 2.95 & 0.04 & 1.05 \\ 
~~~~High-Dimension & 4.05 & 13.60 & 9.18 & 0.36 & 0.85 \\ 
\multicolumn{6}{l}{\textit{~~Partially Non-Overlapping Support}} \\
~~~~Low-Dimension & 3.54 & 1.50 & 3.33 & 0.57 & 3.10 \\ 
~~~~High-Dimension & 3.96 & 12.86 & 7.71 & 3.75 & 3.49 \\ 
\multicolumn{6}{l}{\textbf{Low Non-Terminal Event Rate}} \\ 
\multicolumn{6}{l}{\textit{~~Shared Support}} \\ 
~~~~Low-Dimension & 8.28 & 11.05 & 6.49 & 0.17 & 1.14 \\ 
~~~~High-Dimension & 9.65 & 17.90 & 18.37 & 0.96 & 1.36 \\ 
\multicolumn{6}{l}{\textit{~~Partially Non-Overlapping Support}} \\
~~~~Low-Dimension & 4.67 & 3.45 & 4.08 & 1.32 & 3.86 \\ 
~~~~High-Dimension & 5.23 & 16.18 & 12.61 & 6.49 & 4.61 \\ 
\hline
\end{tabular}}
\end{center}
\end{table}

\subsubsection{Small Sample Result ($n=300$),  Weibull and Piecewise Constant Models}

\begin{table}[H]
\caption{Mean count of false exclusions, Weibull and piecewise constant baseline hazard specifications, sample size $n=300$.\label{tab:falseexclsmall}}
\begin{center}
\resizebox{0.95\textwidth}{!}{\begin{tabular}{lcccccc}
\hline
\multicolumn{6}{l}{\ul{Weibull}} \\ 
$n=300$ & Forward & Lasso & SCAD & Lasso + Fusion & SCAD + Fusion \\ 
\hline
\multicolumn{6}{l}{\textbf{Moderate Non-Terminal Event Rate}} \\ 
\multicolumn{6}{l}{\textit{~~Shared Support}} \\ 
~~~~Low-Dimension & 15.88 & 18.48 & 14.25 & 1.47 & 5.66 \\ 
~~~~High-Dimension & 17.69 & 27.59 & 14.30 & 24.21 & 14.30 \\ 
\multicolumn{6}{l}{\textit{~~Partially Non-Overlapping Support}} \\
~~~~Low-Dimension & 15.87 & 22.75 & 13.50 & 12.65 & 11.00 \\ 
~~~~High-Dimension & 17.64 & 27.69 & 14.06 & 27.51 & 14.06 \\ 
\multicolumn{6}{l}{\textbf{Low Non-Terminal Event Rate}} \\ 
\multicolumn{6}{l}{\textit{~~Shared Support}} \\ 
~~~~Low-Dimension & 17.90 & 17.59 & 17.92 & 5.89 & 8.90 \\ 
~~~~High-Dimension & 19.00 & 23.51 & 16.59 & 21.32 & 16.55 \\ 
\multicolumn{6}{l}{\textit{~~Partially Non-Overlapping Support}} \\
~~~~Low-Dimension & 17.12 & 19.96 & 15.34 & 13.28 & 11.46 \\ 
~~~~High-Dimension & 18.98 & 26.51 & 15.06 & 27.55 & 15.06 \\ 
\hline
\multicolumn{6}{l}{\ul{Piecewise Constant}} \\ 
$n=300$ & Forward & Lasso & SCAD & Lasso + Fusion & SCAD + Fusion \\ 
\hline
\multicolumn{6}{l}{\textbf{Moderate Non-Terminal Event Rate}} \\ 
\multicolumn{6}{l}{\textit{~~Shared Support}} \\ 
~~~~Low-Dimension & 15.91 & 18.23 & 13.83 & 1.71 & 5.92 \\ 
~~~~High-Dimension & 17.39 & 27.70 & 14.95 & 23.64 & 14.19 \\ 
\multicolumn{6}{l}{\textit{~~Partially Non-Overlapping Support}} \\
~~~~Low-Dimension & 15.72 & 22.58 & 13.42 & 12.55 & 10.90 \\ 
~~~~High-Dimension & 17.59 & 27.49 & 14.48 & 27.61 & 13.95 \\ 
\multicolumn{6}{l}{\textbf{Low Non-Terminal Event Rate}} \\ 
\multicolumn{6}{l}{\textit{~~Shared Support}} \\ 
~~~~Low-Dimension & 17.83 & 17.43 & 17.82 & 6.02 & 8.99 \\ 
~~~~High-Dimension & 18.56 & 24.00 & 17.45 & 21.27 & 15.64 \\ 
\multicolumn{6}{l}{\textit{~~Partially Non-Overlapping Support}} \\
~~~~Low-Dimension & 16.94 & 20.11 & 15.38 & 13.57 & 11.18 \\ 
~~~~High-Dimension & 18.78 & 27.04 & 16.02 & 27.70 & 15.74 \\ 
\hline
\end{tabular}}
\end{center}
\end{table}


\newpage
\section{Additional Algorithmic Details}\label{sec:algodetails}

\subsection{Tuning the Nesterov smoothing parameter $\mu$}
The Nesterov smoothing parameter $\mu$ defined in \eqref{eq:pensmooth} is important for optimization under the proposed structured fusion penalty. If $\mu$ is too large then the approximation may be too loose to induce fusion, while if $\mu$ is too small the approximation will be insufficiently smooth, and the optimization algorithm will exhibit poor performance. 

In our proposed optimization algorithm, we therefore avoid pre-specifying a single $\mu$ value, and instead adopt a simple algorithmic approach following \citet{hahn2020framework} called `progressive smoothing'. For each fixed set of regularization parameters $(\lambda_1,\lambda_2)$, we first iterate proximal gradient descent to convergence with $\mu$ large (e.g., $10^{-2}$ in our applications), and then decrease $\mu$ and further iterate to convergence, and so on until $\mu$ is sufficiently small. (In our applications, we decrease $\mu$ over the sequence $\{10^{-2},10^{-3},10^{-4},10^{-5},10^{-6}\}$ ). This approach allows for tight approximations of the fusion penalty to be achieved, while remaining computationally efficient because optimizing for each new $\mu$ typically requires only a handful of proximal gradient descent iterations from the preceding $\mu$. This progressive smoothing procedure is fully implemented in our \texttt{SemiCompRisksPen} package.

\section{Spline-Based Baseline Hazard Specifications}\label{sec:splines}

In this section, we summarize several spline-based specifications for the baseline hazards in the illness-death model, to which the proposed estimation framework can be readily extended. 

\subsection{Polynomial B-Spline on Log-Hazard Scale}

One choice is to specify each log baseline hazard as a polynomial B-spline function of degree $b_g$. This approach specifies 
\begin{equation}\label{eq:bs}
\log h_{0g}(t) = \sum_{j=1}^{k_{g}} \phi_{gj}B_{gj}(t)
\end{equation} where $k_g \geq b_g+1$ is the number of desired baseline parameters for the $g$th hazard, and $B_{gj}(\cdot)$ is the $j$th B-spline basis function \citep{deboor2001practical}. These basis functions are defined according to the placement of $k_g + b_g - 1$ knots, such that linear combinations of these basis functions can express a range of shapes constrained at those knots. Usually $b_g=3$ such that the resulting spline function must be continuously twice-differentiable. B-splines are only well-defined over the interval spanned by the knots, so the boundary knots are typically fixed at 0 and the maximum observed time in the $g$th hazard, with inner knots typically placed at appropriate deciles. While extremely flexible, this specification requires numerical integration to compute the cumulative hazard, making it slower in practice than other specifications.

Using general notation over $g\in \{1,2,3\}$, $r\in \{1,2,3\}$, $j=1,\dots,k_{g}$, and $l=1,\dots,k_{r}$, the first two derivatives of the cause-specific log-baseline are
\begin{align}
\frac{\partial}{\partial\phi_{gj}} \log h_{0g}(t) = & B_{gj}(t)
\\ \frac{\partial^2}{\partial\phi_{gj}\partial\phi_{rl}} \log h_{0g}(t) = & 0
\end{align}
The cause-specific cumulative hazard and its first two derivatives do not have closed form, but can be written as
\begin{align}
H_{0g}(t) = & \int_{0}^t \exp\left(\sum_{j=1}^{k_{g}} \phi_{gj}B_{gj}(s)\right)  ds
\\ \frac{\partial}{\partial\phi_{gj}} H_{0g}(t) = & \int_{0}^t B_{gj}(s)\exp\left(\sum_{j'=1}^{k_{g}} \phi_{gj'}B_{gj'}(s)\right)  ds
\\ \frac{\partial^2}{\partial\phi_{gj}\partial\phi_{rl}} H_{0g}(t) = & \int_{0}^t B_{gj}(s)B_{gl}(s)\exp\left(\sum_{j'=1}^{k_{g}} \phi_{gj'}B_{gj'}(s)\right)\mathbb{I}(g=r)  ds
\end{align}

\subsection{Restricted Cubic Spline on Log-Cumulative Hazard Scale}

A final spline-based approach follows \citet{royston2002flexible} in specifying the log-cumulative baseline hazard as a natural (or `restricted') cubic spline function. Whereas B-splines are only defined over range spanned by the chosen knots, natural cubic splines extend linearly beyond the boundary knots.  Set $z = \log(t)$, and to generate $k_g$ total parameters, consider a set of $k_g$ knots $0 \leq z_g^{(1)} < \dots < z_g^{(k_g)}$. Define $v_{g1}(z)=1$, $v_{g2}(z)=z$, and for $j = 3, \dots, k_{g}$ define natural cubic spline basis functions as
\begin{align}
v_{gj}(z) = (z-z_g^{(j)})_{+}^3 - \zeta_j(z-z_g^{(1)})_{+}^3 - (1-\zeta_j)(z-z_g^{(k_g)})_{+}^3,
\end{align} where $\zeta_{gj} = (z_g^{(k_g)}-z_g^{(j)})/(z_g^{(k_g)}-z_g^{(1)})$ and $(z)_{+} = \max(0,z)$.  Then this model specifies the baseline hazard as
\begin{equation}\label{eq:rp}
\log H_{0g}(t) = \sum_{j=1}^{k_g}\phi_{gj}v_{gj}(z).
\end{equation} By specifying the spline model of the log-cumulative hazard, evaluation does not require numerical integration. However, while the log-cumulative hazard is constrained to be monotonically increasing, there is no such constraint inherently on the natural cubic spline. While in principle this might require putting formal constraints on $\bphi_g$ during optimization, in practice standard unconstrained methods suffice as long as the starting point is feasible and there is a modest amount of data \citep{herndon1990restricted,royston2002flexible}. The Weibull model is a special case of the Royston-Parmar model, fixing no internal knots (e.g., setting $k_g=2$). 

The cumulative cause-specific hazard and its first two derivatives are
\begin{align}
H_{0g}(t) = & \exp\left(\sum_{j=1}^{k_g}\phi_{gj}v_{gj}(z)\right)
\\ \frac{\partial}{\partial\phi_{gj}} H_{0g}(t) = & v_{gj}(z)\exp\left(\sum_{j=1}^{k_g}\phi_{gj}v_{gj}(z)\right)
\\ \frac{\partial^2}{\partial\phi_{gj}\partial\phi_{rl}} H_{0g}(t) = & v_{gj}(z)v_{gl}(z)\exp\left(\sum_{j=1}^{k_g}\phi_{gj}v_{gj}(z)\right)\mathbb{I}(g=r)
\end{align}

Now, define $v'_{g1}(z) = 0$, and $v'_{g2}(z)=1$, and for $j = 3, \dots, k_{g}$, 
\begin{align}
v'_{gj}(z) = 3(z-z_g^{(j)})_{+}^2 - 3\zeta_j(z-z_g^{(1)})_{+}^2 - 3(1-\zeta_j)(z-z_g^{(k_g)})_{+}^2.
\end{align}
Using general notation over $g\in \{1,2,3\}$, $r\in \{1,2,3\}$, $j=1,\dots,k_{g}$, and $l=1,\dots,k_{r}$, the cause-specific log-baseline hazard and its first two derivatives are
\begin{align}
\log h_{0g}(t) = & \log\left(\sum_{j=1}^{k_g}\phi_{gj}v'_{gj}(z)\right) + \sum_{j=1}^{k_g}\phi_{gj}v_{gj}(z) - \log t
\\ \frac{\partial}{\partial\phi_{gj}} \log h_{0g}(t) = & \frac{v'_{gj}(z)}{\sum_{j'=1}^{k_g}\phi_{gj'}v'_{gj'}(z)} + v_{gj}(z)
\\ \frac{\partial^2}{\partial\phi_{gj}\partial\phi_{rl}} \log h_{0g}(t) = &-  \frac{v'_{gj}(z)v'_{gl}(z)}{\left(\sum_{j'=1}^{k_g}\phi_{gj'}v'_{gj'}(z)\right)^2}\mathbb{I}(g=r)
\end{align}

\end{document}

%% file: GrandMacros.tex
\newcommand{\ptilde}{{\widetilde p}}

\newcommand{\ytilde}{{\widetilde y}}

\newcommand{\Ytilde}{{\widetilde Y}}

\newcommand{\bx}{\mathbf{ x}}

\newcommand{\bz}{\mathbf{ z}}

\newcommand{\bD}{\mathbf{ D}}

\newcommand{\bM}{\mathbf{ M}}

\newcommand{\bS}{\mathbf{ S}}
\newcommand{\bT}{\mathbf{ T}}

\newcommand{\bX}{\mathbf{ X}}
\newcommand{\bY}{\mathbf{ Y}}

\newcommand{\bbE}{\mathbb{ E}}

\newcommand{\bbI}{\mathbb{ I}}

\newcommand{\bbR}{\mathbb{ R}}

\newcommand{\deltatilde}{{\widetilde\delta}}

\newcommand{\Deltatilde}{{\widetilde\Delta}}

\newcommand{\Omegatilde}{{\widetilde\Omega}}
\newcommand{\bbeta}{\boldsymbol{\beta}}

\newcommand{\blambda}{\boldsymbol{\lambda}}

\newcommand{\bnu}{\boldsymbol{\nu}}

\newcommand{\bpi}{\boldsymbol{\pi}}

\newcommand{\bphi}{\boldsymbol{\phi}}

\newcommand{\bpsi}{\boldsymbol{\psi}}

\newcommand{\bDelta}{\boldsymbol{\Delta}}

\newcommand{\bbetahat}{{\widehat\bbeta}}

\newcommand{\bpsihat}{{\widehat\bpsi}}

\DeclareMathOperator\Var{Var}

\DeclareMathOperator\sign{sign}

\DeclareMathOperator*{\argmin}{arg\,min}
\newcommand{\sumin}{\sum_{i=1}^n}

\newcommand{\transpose}{{\mathsf{\scriptscriptstyle T}}}
\newcommand{\trans}{^{\transpose}}

\newcommand{\indep}{\mathrel{\perp\mspace{-10mu}\perp}}